\documentclass{svmult}

\usepackage{amsmath}
\usepackage{amsfonts}
\usepackage{amssymb}
\usepackage{color}
\usepackage[bottom]{footmisc}
\usepackage{graphicx}
\usepackage{subfig}
\usepackage{wrapfig}
 
\definecolor{lightbg}{rgb}{0.85,0.85,0.85}

\begin{document}

\title*{The B36/S125 ``2x2'' Life-Like Cellular Automaton}
\author{Nathaniel Johnston\inst{1}}
\institute{$^1$Department of Mathematics \& Statistics, University of Guelph, Guelph, ON, Canada N1G 2W1
\texttt{njohns01@uoguelph.ca}}

\maketitle

\begin{abstract}
	The B36/S125 (or ``2x2'') cellular automaton is one that takes place on a 2D square lattice much like Conway's Game of Life. Although it exhibits high-level behaviour that is similar to Life, such as chaotic but eventually stable evolution and the existence of a natural diagonal glider, the individual objects that the rule contains generally look very different from their Life counterparts. In this article, a history of notable discoveries in the 2x2 rule is provided, and the fundamental patterns of the automaton are described. Some theoretical results are derived along the way, including a proof that the speed limits for diagonal and orthogonal spaceships in this rule are $c/3$ and $c/2$, respectively. A Margolus block cellular automaton that 2x2 emulates is investigated, and in particular a family of oscillators made up entirely of $2 \times 2$ blocks are analyzed and used to show that there exist oscillators with period $2^\ell(2^k - 1)$ for any integers $k,\ell \geq 1$.
\end{abstract}

  \section{Introduction}\label{sec:intro}

  One cellular automaton that has drawn a fair amount of interest recently is the one that takes place on a grid like Conway's Game of Life, except dead cells are born if they have 3 or 6 live neighbours, and alive cells survive if they have 1, 2, or 5 live neighbours -- this information is conveyed by its rulestring ``B36/S125''. This rule exhibits many qualities that are similar to those of Life -- for example, evolution seems unpredictable and random patterns tend to evolve into ``ash fields'' consisting of several small stable patterns (known as \emph{still lifes}) and periodic patterns (known as \emph{oscillators}).
  
  The B36/S125 automaton has become known as ``2x2'' because of the fact that it emulates a simpler cellular automaton that acts on $2 \times 2$ blocks of cells. In particular, this means that patterns that are initially made up of $2 \times 2$ blocks will forever be made up of $2 \times 2$ blocks under this evolution rule. Because of the simplicity of the emulated block cellular automaton, it has many properties in common with elementary cellular automata~\cite{W83,W02} and in particular emulates Wolfram's rule 90 \cite{WR90}.
  
  Although the rough behaviour and statistics of 2x2 are similar to those of Life, the patterns of 2x2 have completely different structure and thus are interesting in their own right. Furthermore, many questions that have been answered about Life remain open in 2x2, such as whether or not it contains guns or replicators. It has a basic $c/8$ diagonal glider that occurs fairly commonly, though it is larger than the standard Life glider and is thus more difficult to construct. The first infinitely-growing pattern to be discovered was a wickstretcher based on the glider which, despite its simplicity, was not stumbled upon until June 2009. This wickstretcher is displayed in Figure~\ref{fig:2x2_wickstretcher} with alive cells in black and dead cells in white.

\begin{figure}[ht]
\begin{minipage}[t]{0.485\textwidth}
\centering
		  \includegraphics[scale=0.2]{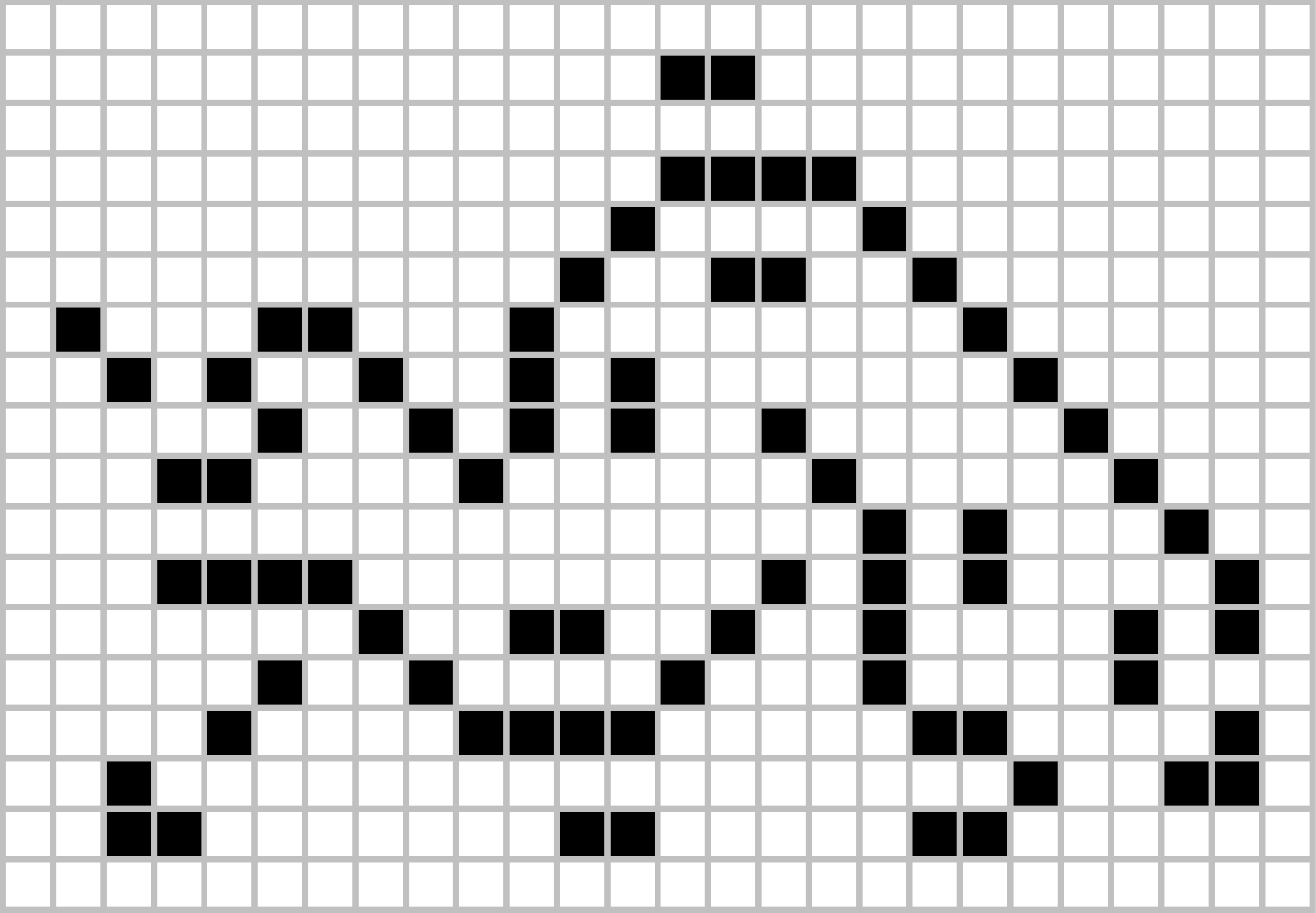}
		  \caption{A large still life}\label{fig:2x2_still_life} 
\end{minipage}
\hspace{0.03\textwidth}
\begin{minipage}[t]{0.485\textwidth}
\centering
		  \includegraphics[scale=0.2]{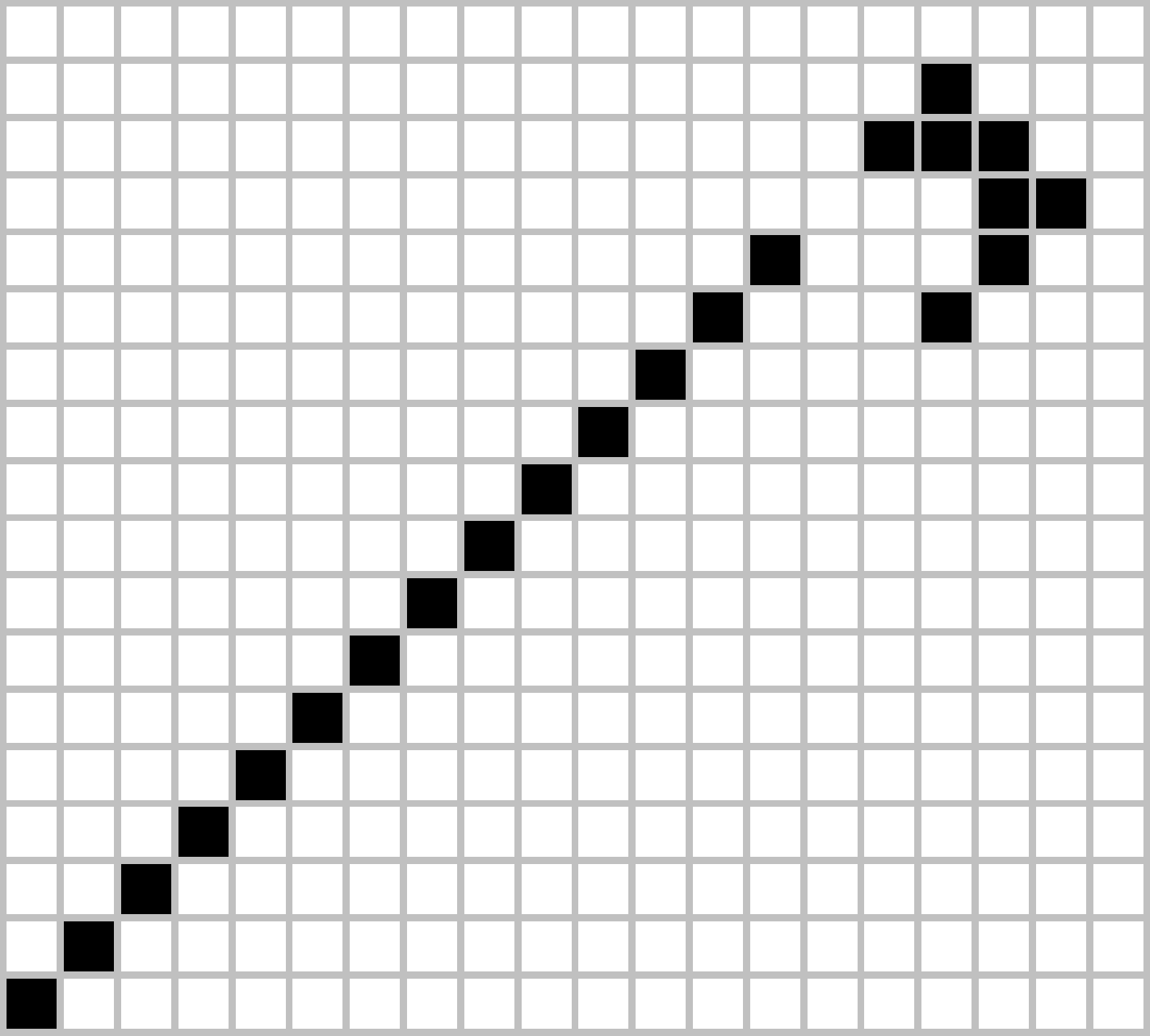}
		  \caption{The wickstretcher}\label{fig:2x2_wickstretcher}
\end{minipage}
\end{figure}

  Other spaceships were found via computerized searches carried out by Alan Hensel, Dean Hickerson, David Bell, and David Eppstein in the 1990s. One of the most important spaceship discoveries was a $c/3$ diagonal spaceship, which showed that it is possible for spaceships to travel faster in the 2x2 universe than in the standard Life universe, despite most of the easy-to-find spaceships being quite slow. We will see that $c/3$ is the diagonal speed limit in 2x2, much like $c/4$ is the diagonal speed limit in Life. Several derived results will also apply to other Life-like cellular automata, and we will note when this is the case, although our focus and motivation will be the 2x2 rule.

  \section{Ash and Common Patterns}\label{sec:ash}
  
  One of those most interesting aspects of 2x2 is the large variety of still lifes and oscillators that appear naturally as a result of evolving randomly-generated starting patterns (known as \emph{soup}). Many simple still lifes are familiar from the standard Life rules, such as the tub, beehive, aircraft carrier, loaf, and pond. More commonly-occurring, however, are simple ``sparse'' still lifes that are not stable in Life, such as a horizontal or diagonal row of $2$ cells. Table~\ref{tab:common_2x2_still} shows the 20 most commonly-occurring still lifes in 2x2.\footnote{Computed by evolving $22,846,665$ random patterns of size $20 \times 20$ and initial density $37.5\%$. A total of $255,689,477$ (non-distinct) still lifes were catalogued as a result. Statistics compiled by the \emph{Online Life-Like CA Soup Search}. http://www.conwaylife.com/soup}
  	
		\begin{table}
		\center
		\caption{The $20$ most common naturally-occurring still lifes in the 2x2 rule and their approximate frequency (out of $1.000$) relative to all still lifes}
		\label{tab:common_2x2_still}
		\begin{tabular}{cc}
		\begin{tabular}[t]{lcc}
		\hline\noalign{\smallskip}
		{\bf \#} \ & \ {\bf Pattern} \ & \ {\bf Rel. Frequency} \\
		\noalign{\smallskip}\hline\noalign{\smallskip}
		$\mathbf{1}$ & \includegraphics[scale=0.02]{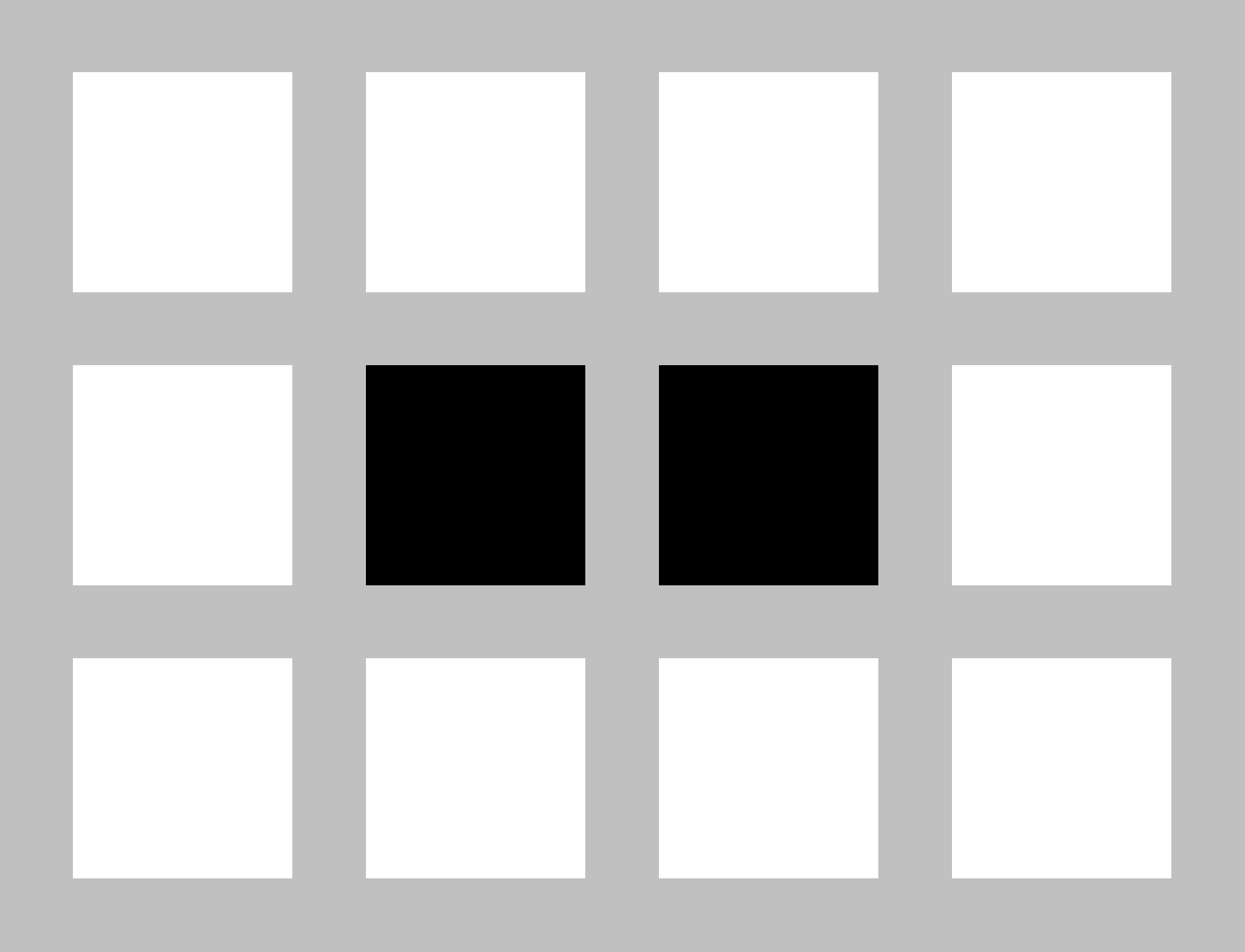} & $6.076 \times 10^{-1}$ \\
		$\mathbf{2}$ & \includegraphics[scale=0.02]{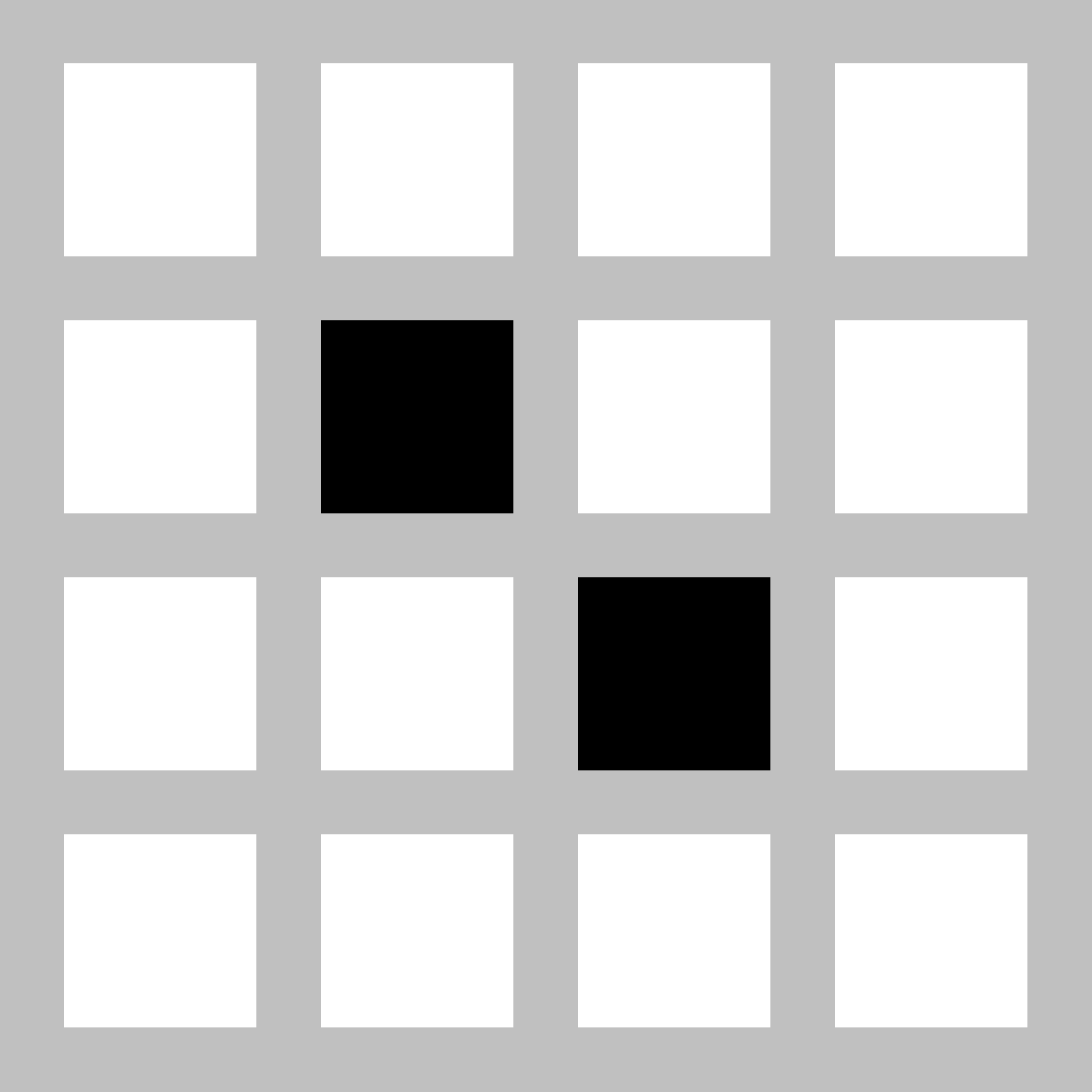} & $2.130 \times 10^{-1}$ \\
		$\mathbf{3}$ & \includegraphics[scale=0.02]{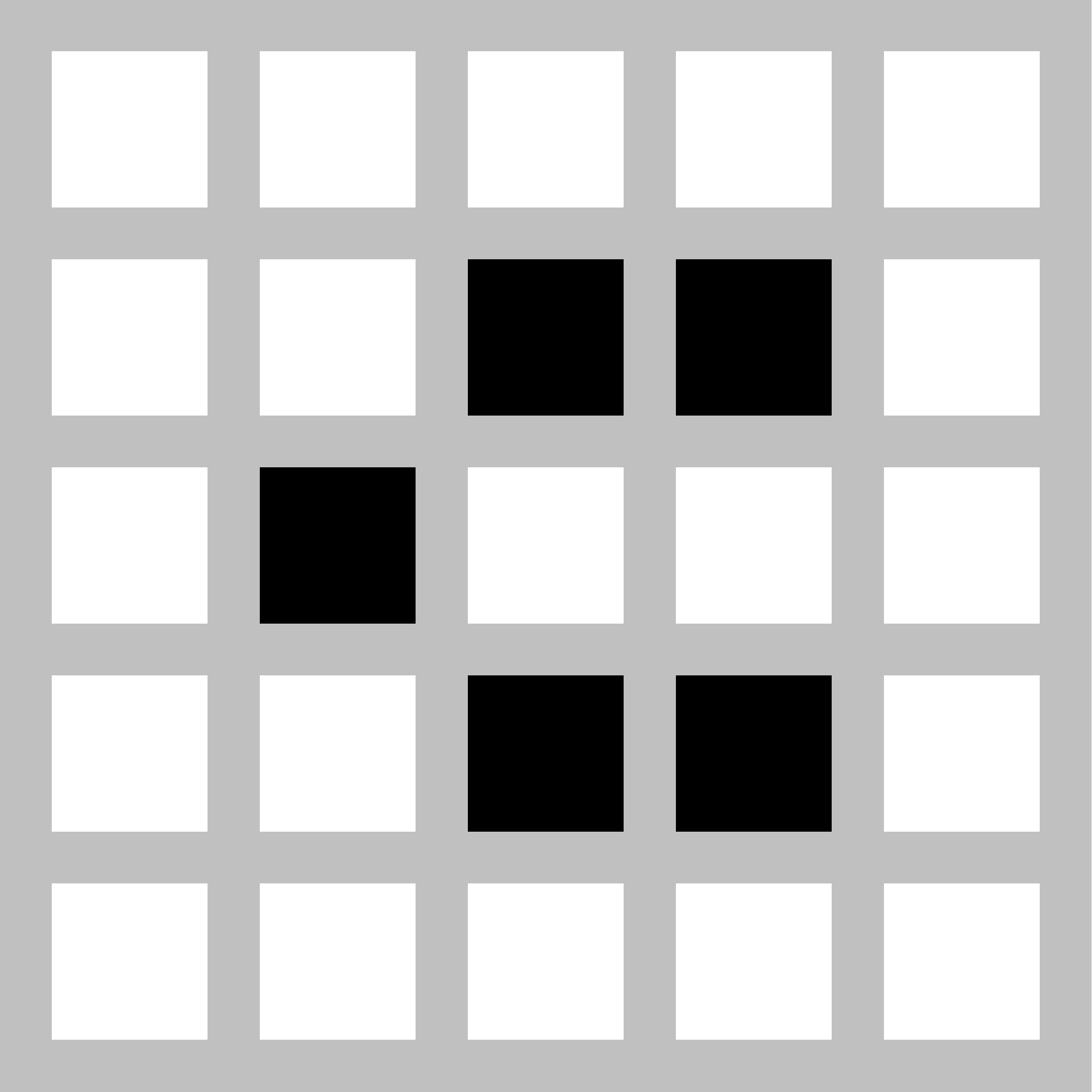} & $5.038 \times 10^{-2}$ \\
		$\mathbf{4}$ & \includegraphics[scale=0.02]{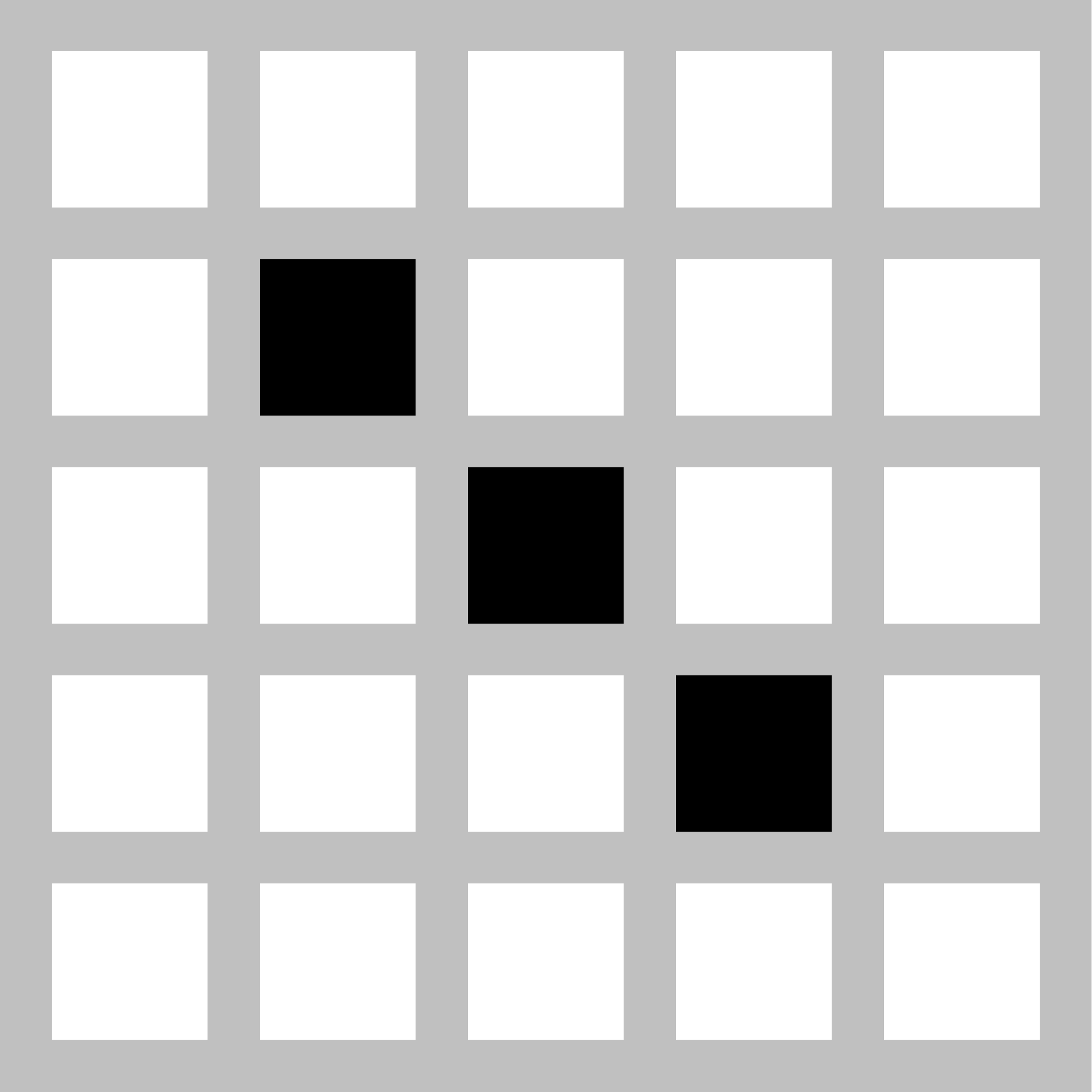} & $4.898 \times 10^{-2}$ \\
		$\mathbf{5}$ & \includegraphics[scale=0.02]{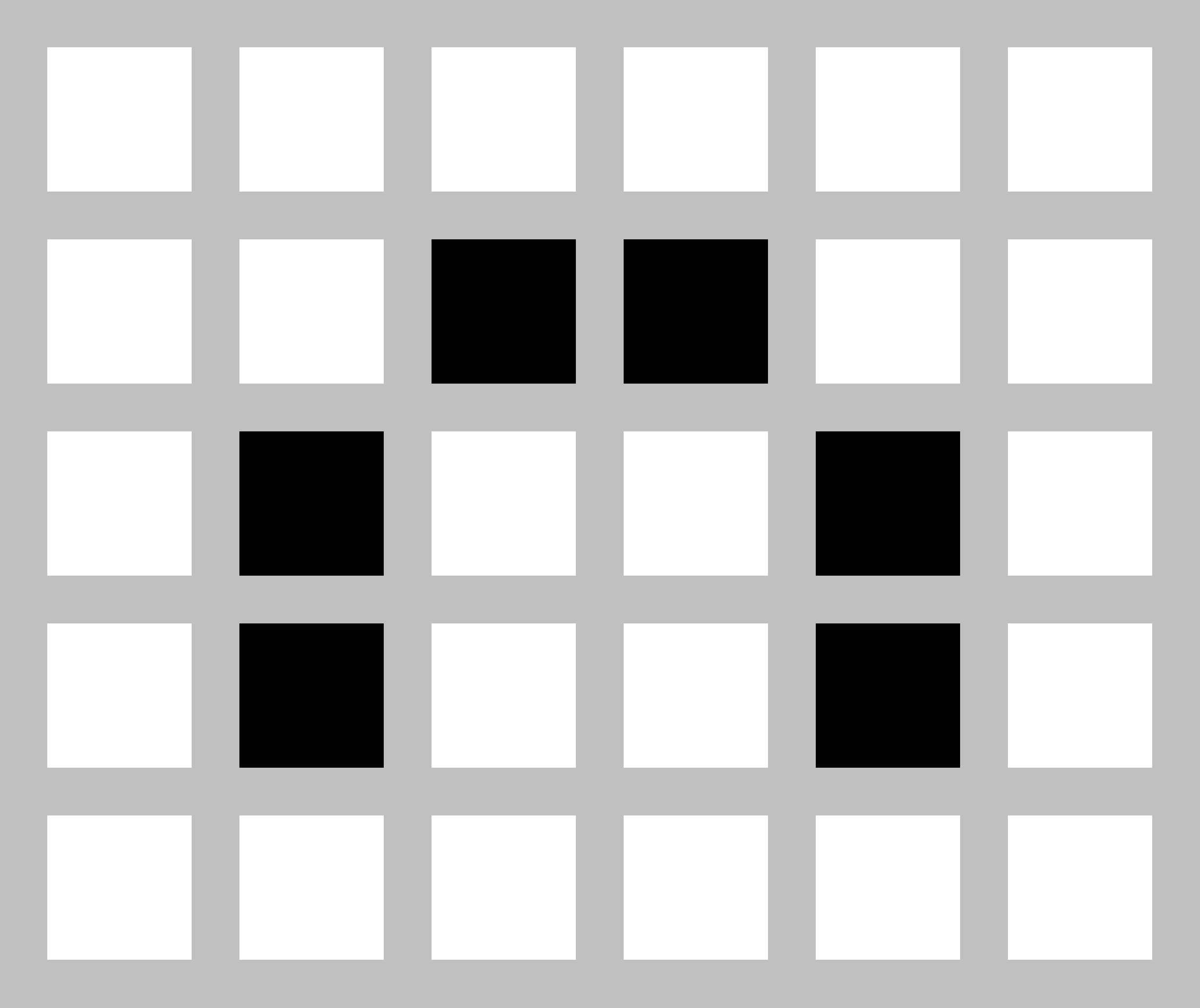} & $2.964 \times 10^{-2}$ \\
		$\mathbf{6}$ & \includegraphics[scale=0.02]{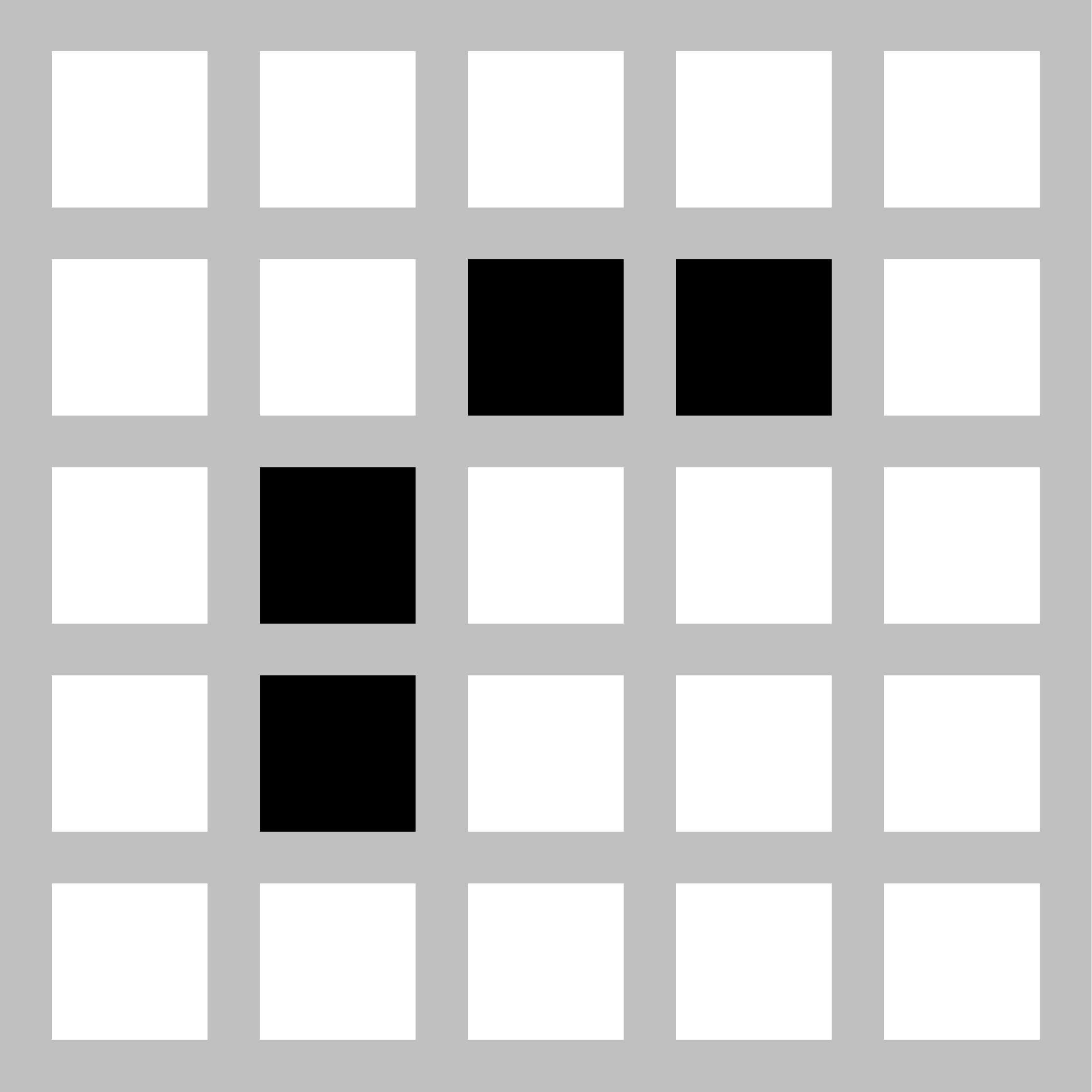} & $2.575 \times 10^{-2}$ \\
		$\mathbf{7}$ & \includegraphics[scale=0.02]{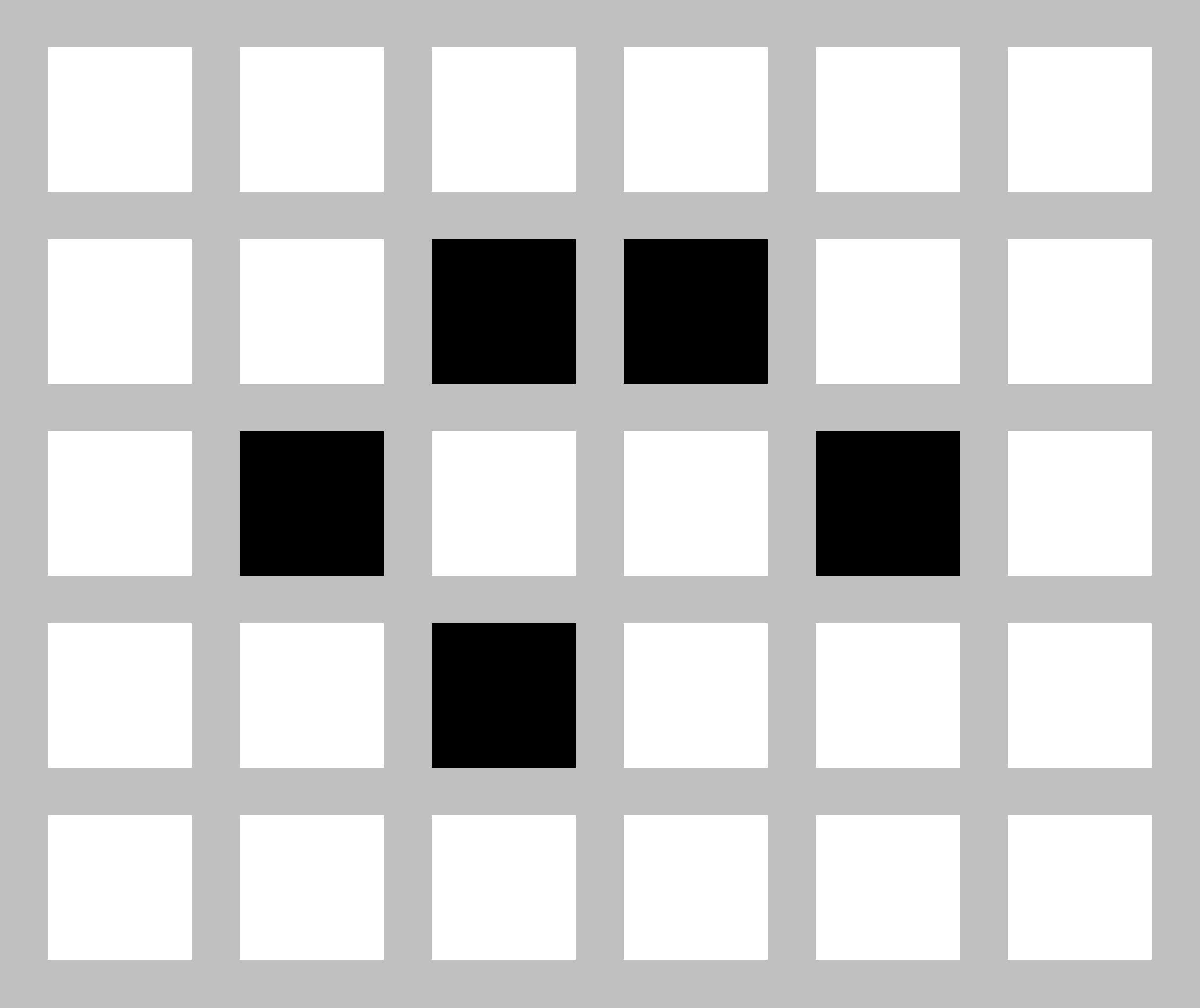} & $9.784 \times 10^{-3}$ \\
		$\mathbf{8}$ & \includegraphics[scale=0.02]{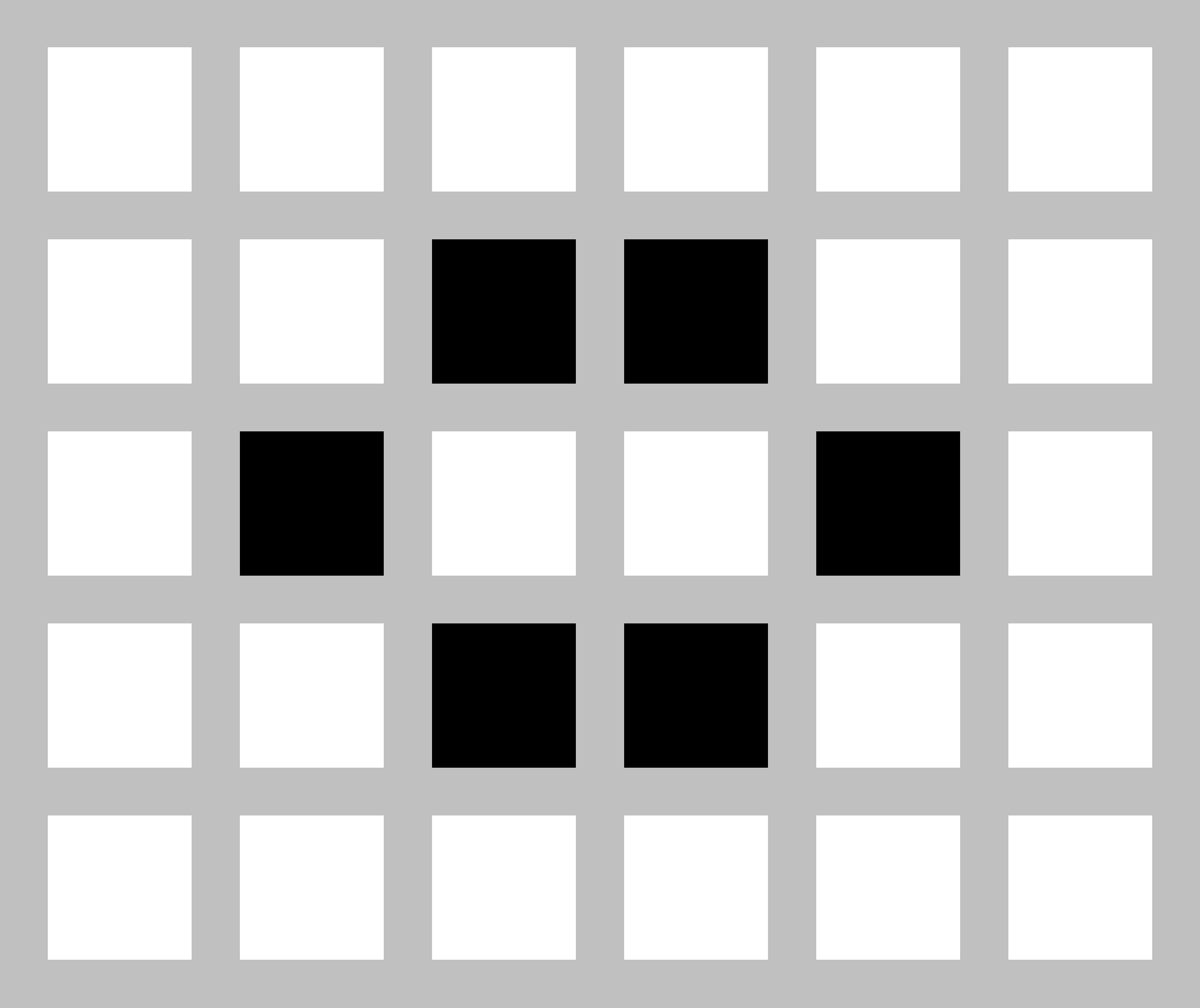} & $4.232 \times 10^{-3}$ \\
		$\mathbf{9}$ & \includegraphics[scale=0.02]{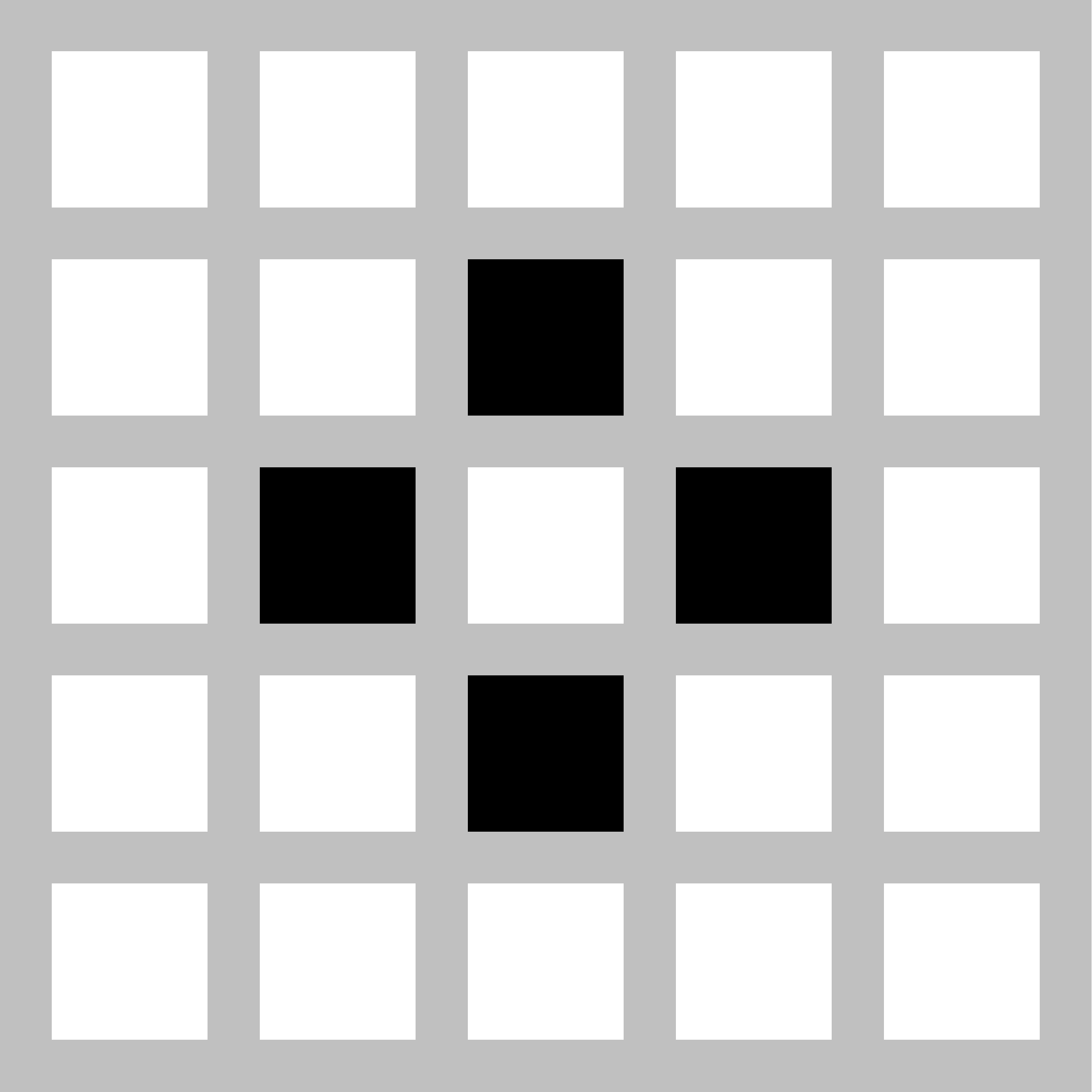} & $4.077 \times 10^{-3}$ \\
		$\mathbf{10}$ & \includegraphics[scale=0.02]{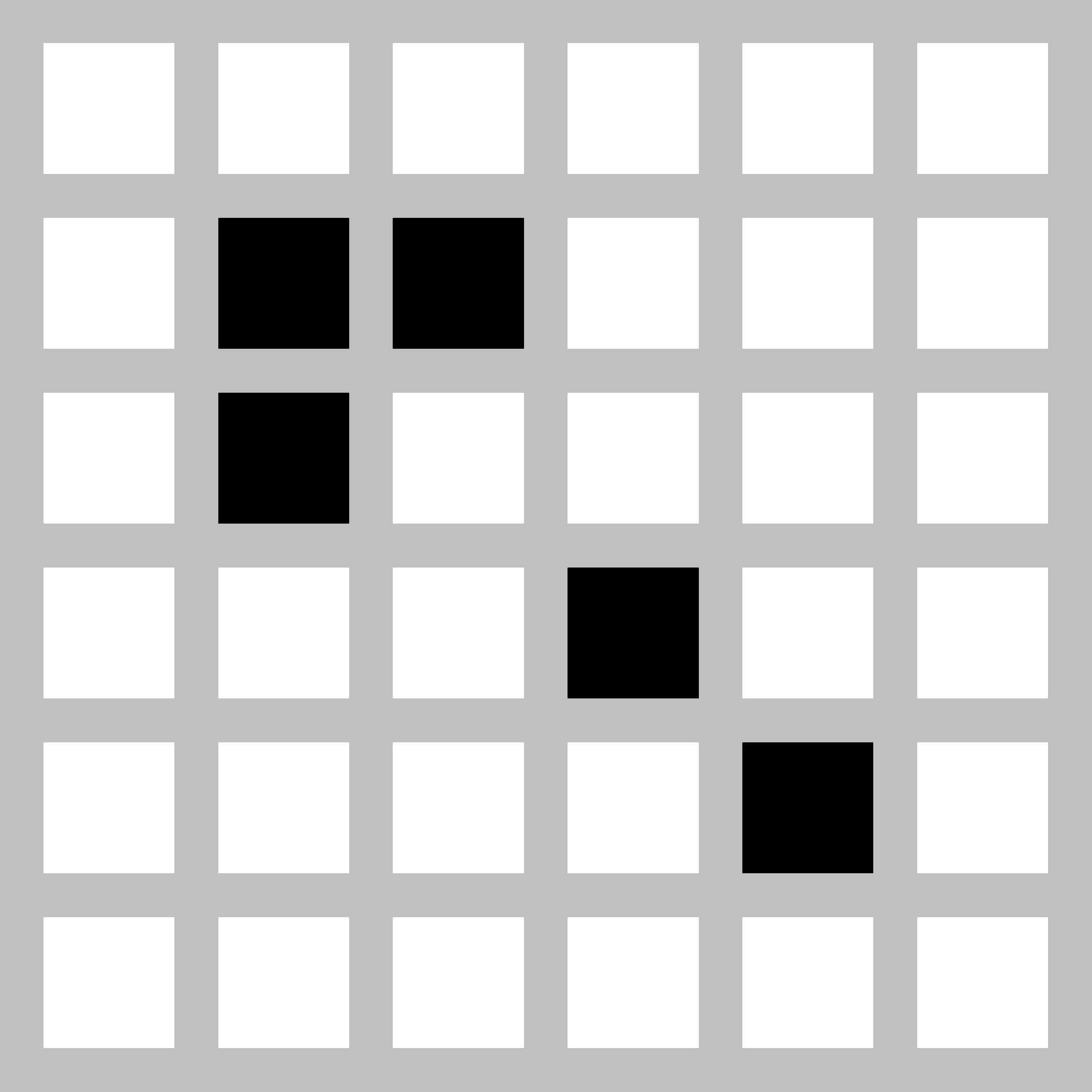} & $1.949 \times 10^{-3}$ \\
		\noalign{\smallskip}\hline
		\end{tabular} \ \ & \ \ \begin{tabular}[t]{lcc}
		\hline\noalign{\smallskip}
		{\bf \#} \ & \ {\bf Pattern} \ & \ {\bf Rel. Frequency} \\
		\noalign{\smallskip}\hline\noalign{\smallskip}
		$\mathbf{11}$ & \includegraphics[scale=0.02]{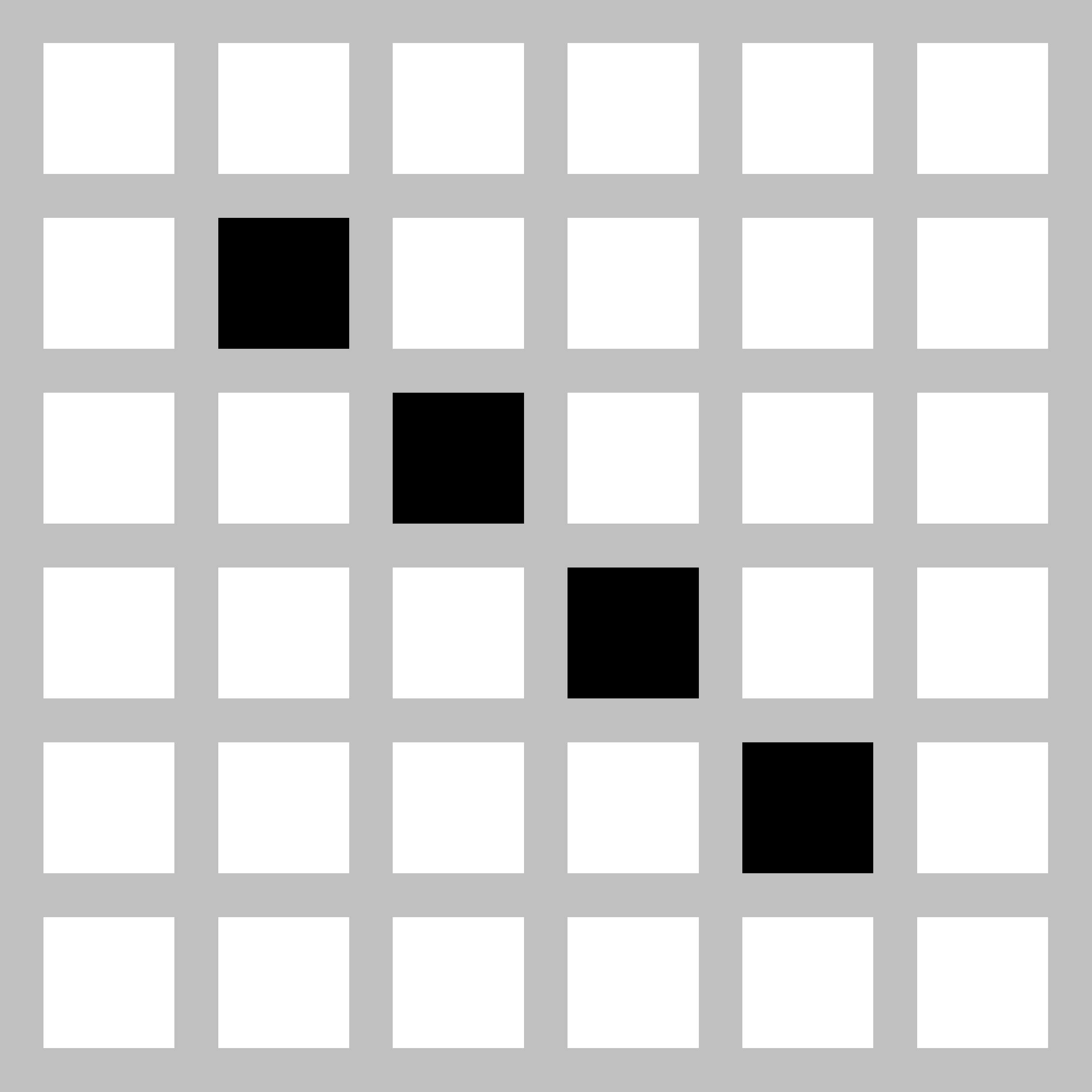} & $8.895 \times 10^{-4}$ \\
		$\mathbf{12}$ & \includegraphics[scale=0.02]{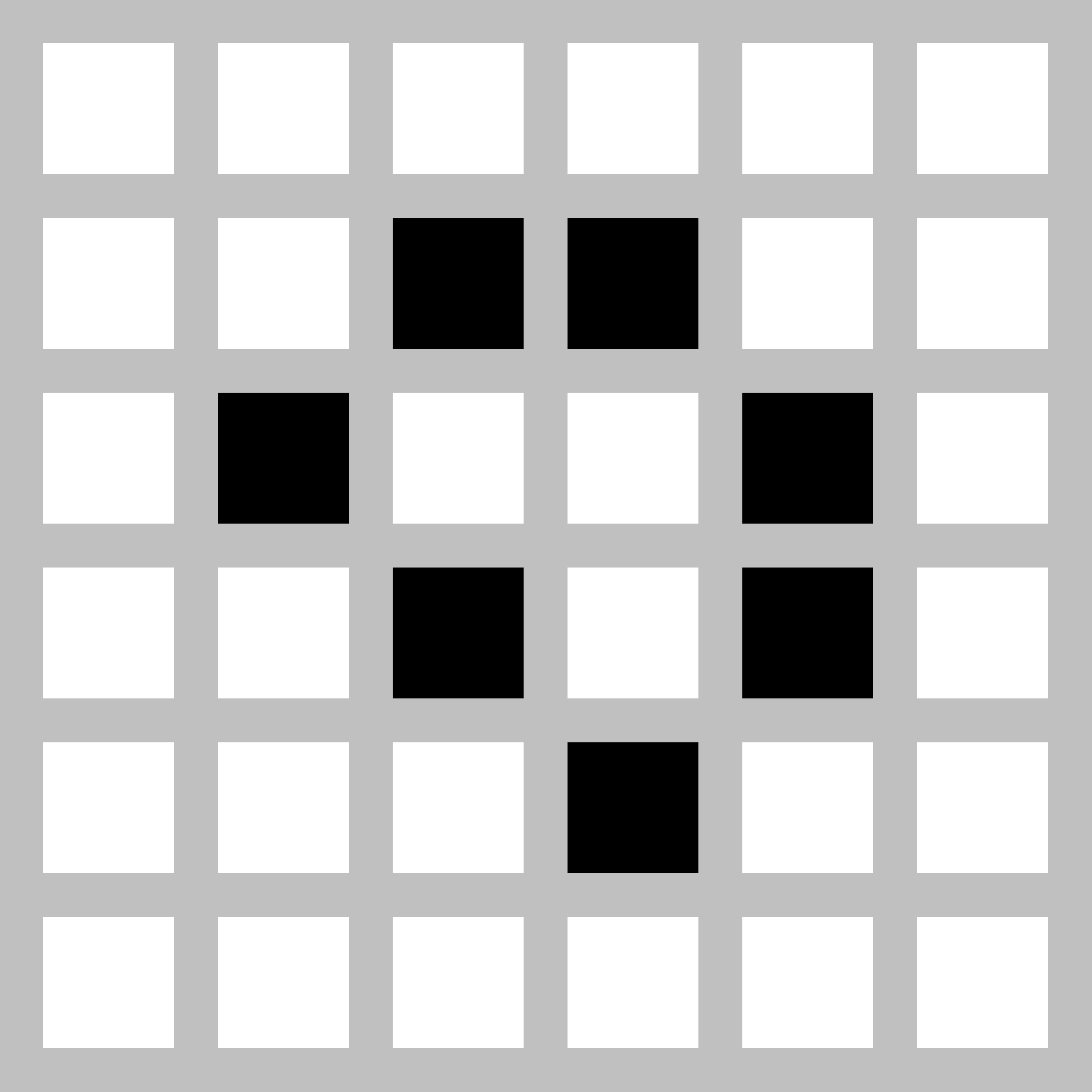} & $7.738 \times 10^{-4}$ \\
		$\mathbf{13}$ & \includegraphics[scale=0.02]{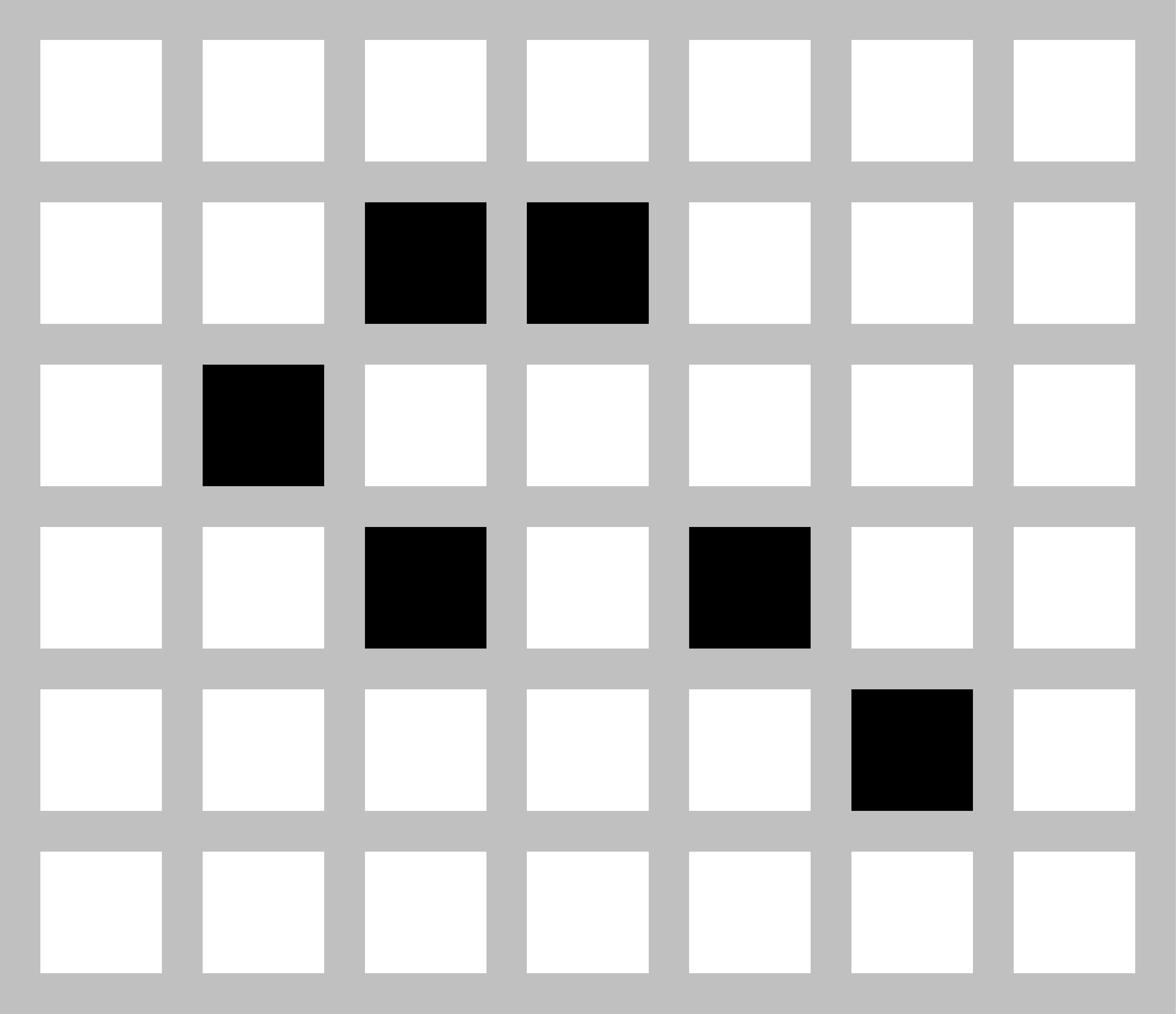} & $6.618 \times 10^{-4}$ \\
		$\mathbf{14}$ & \includegraphics[scale=0.02]{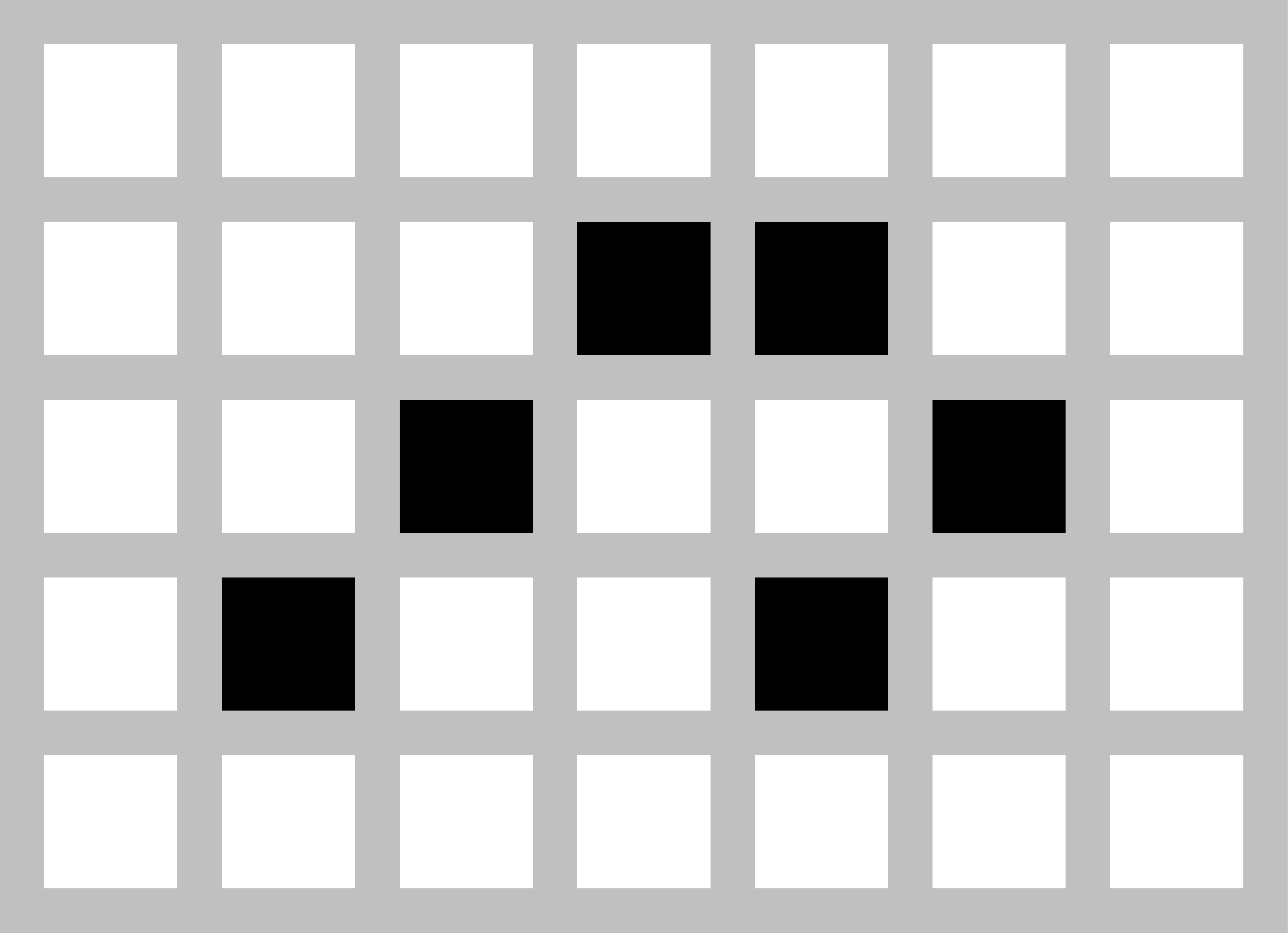} & $5.604 \times 10^{-4}$ \\
		$\mathbf{15}$ & \includegraphics[scale=0.02]{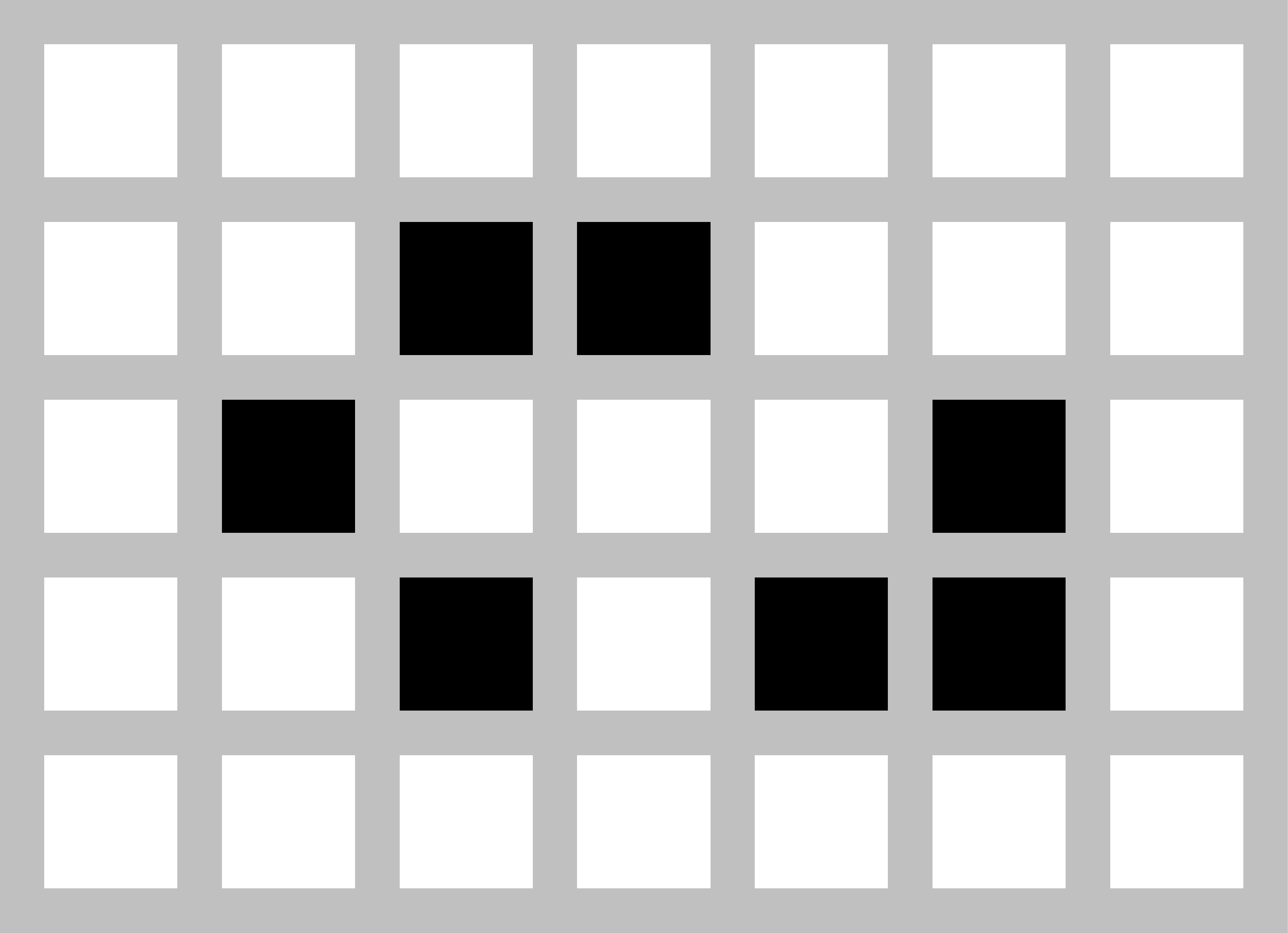} & $3.640 \times 10^{-4}$ \\
		$\mathbf{16}$ & \includegraphics[scale=0.02]{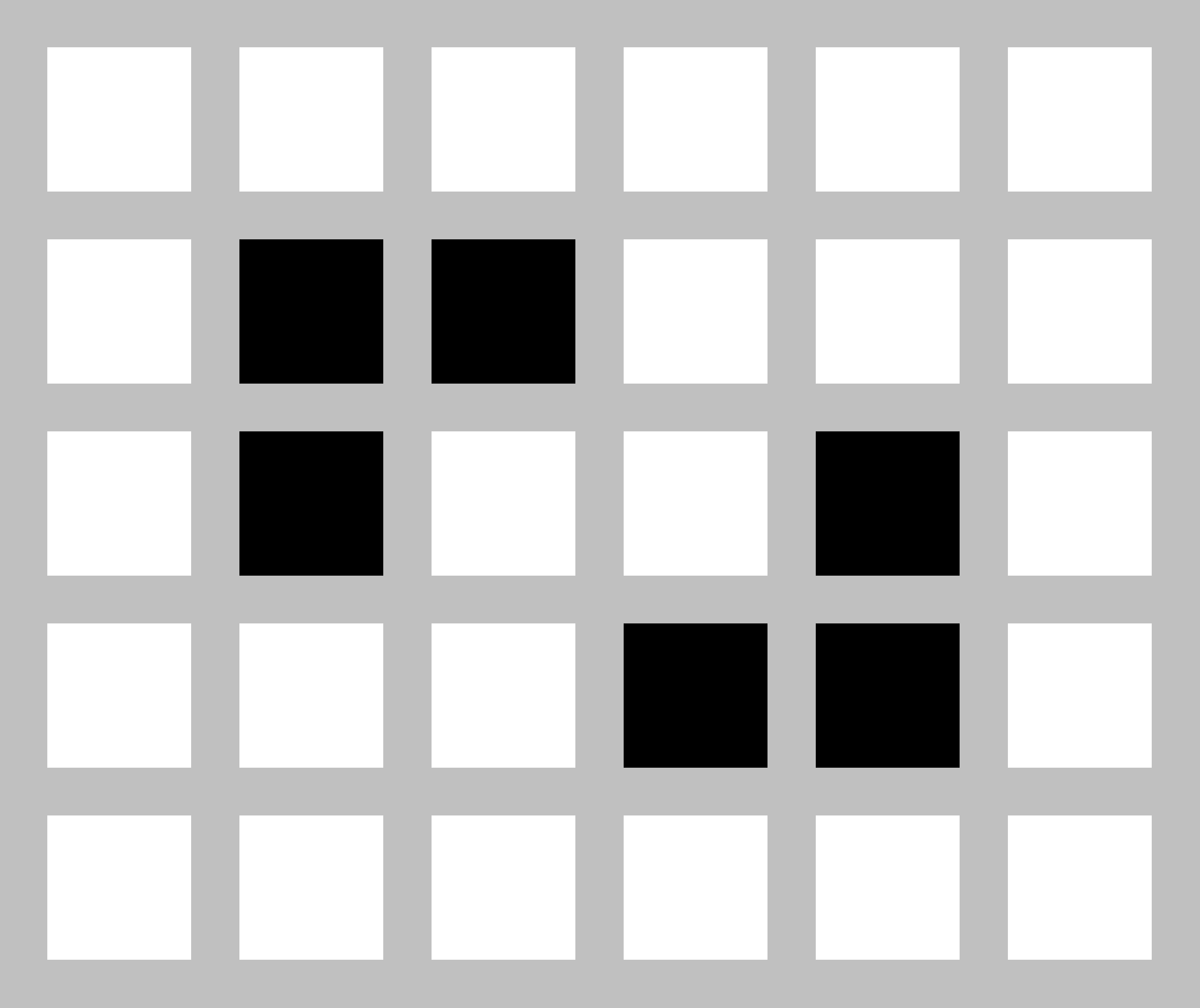} & $2.831 \times 10^{-4}$ \\
		$\mathbf{17}$ & \includegraphics[scale=0.02]{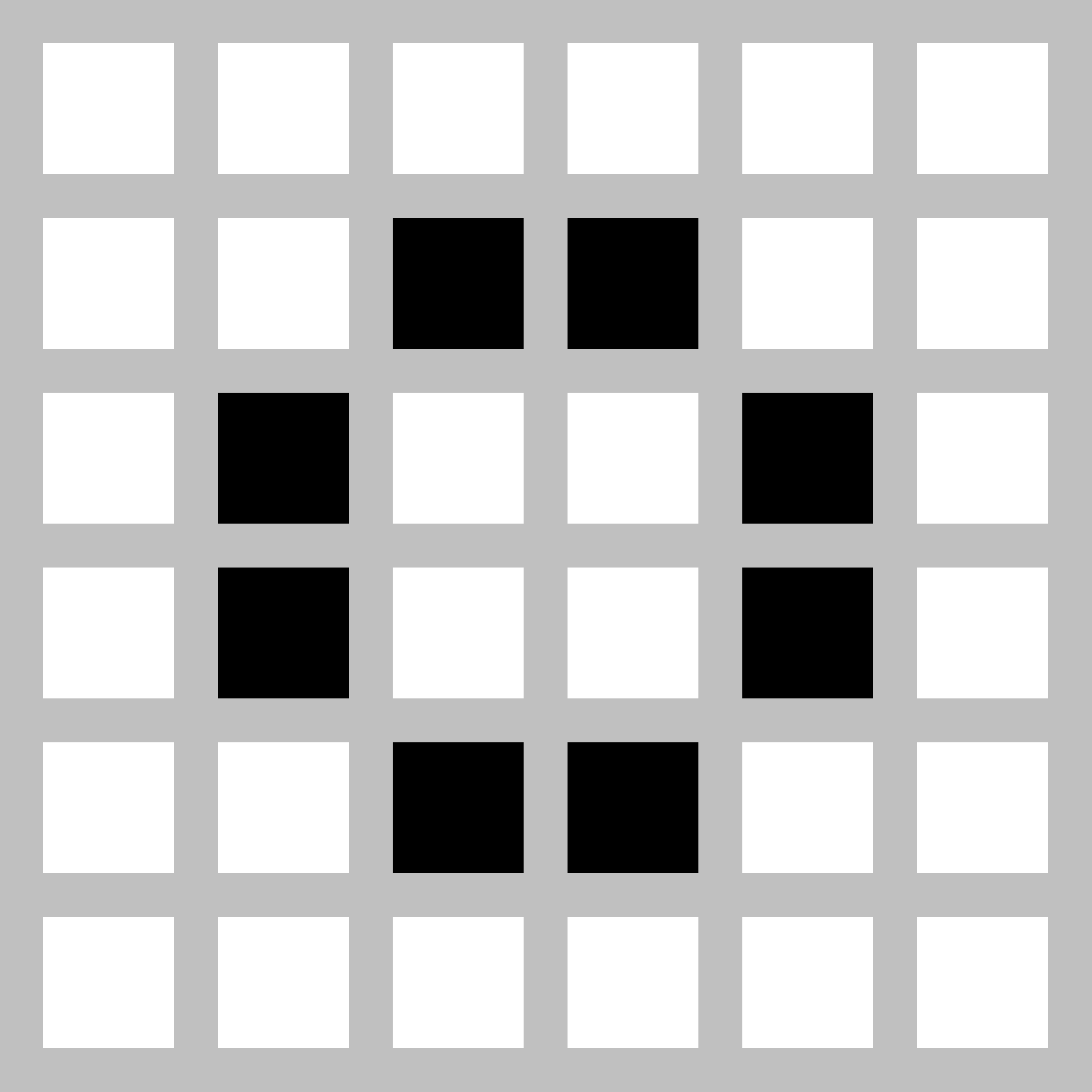} & $2.487 \times 10^{-4}$ \\
		$\mathbf{18}$ & \includegraphics[scale=0.02]{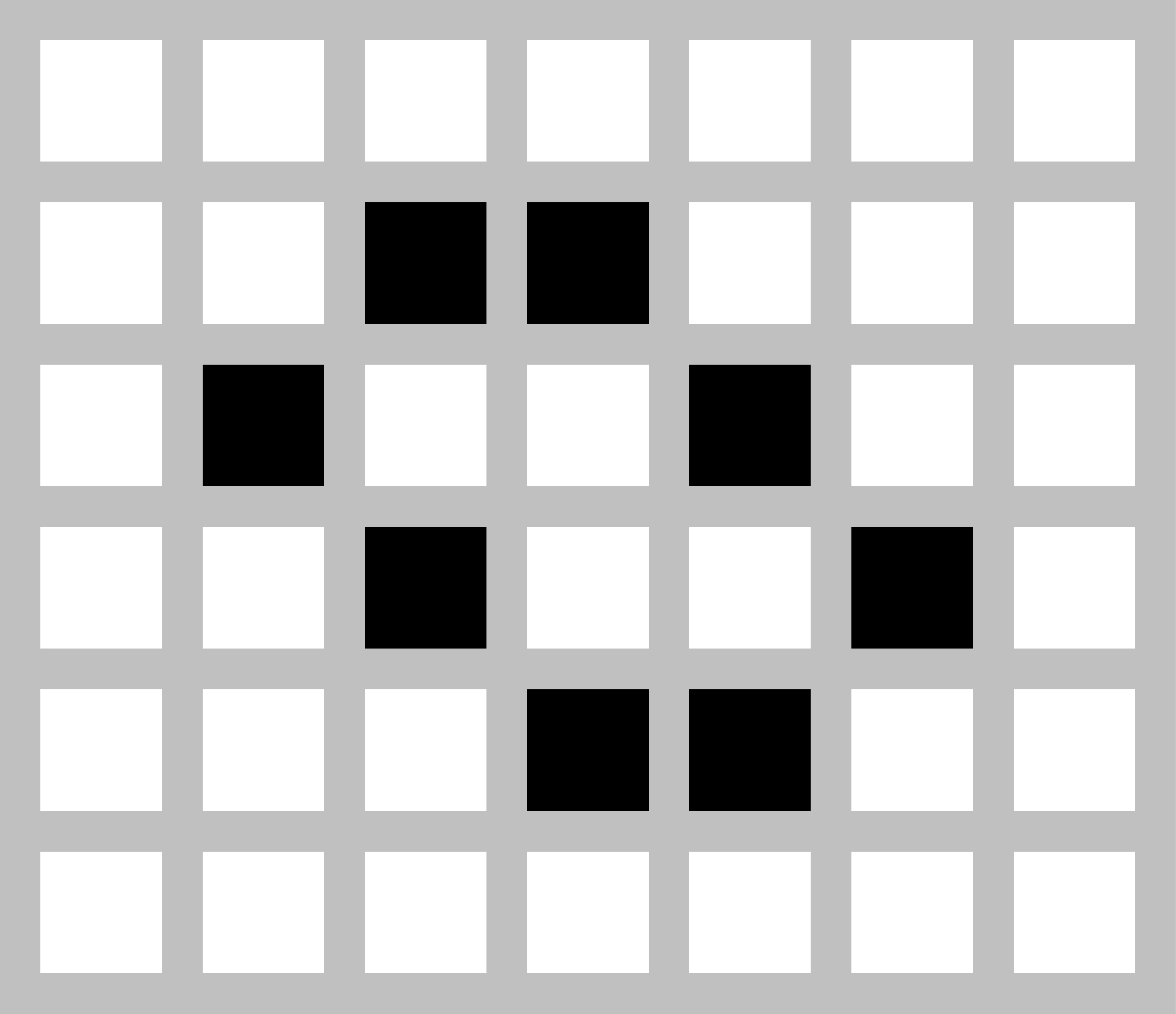} & $1.385 \times 10^{-4}$ \\
		$\mathbf{19}$ & \includegraphics[scale=0.02]{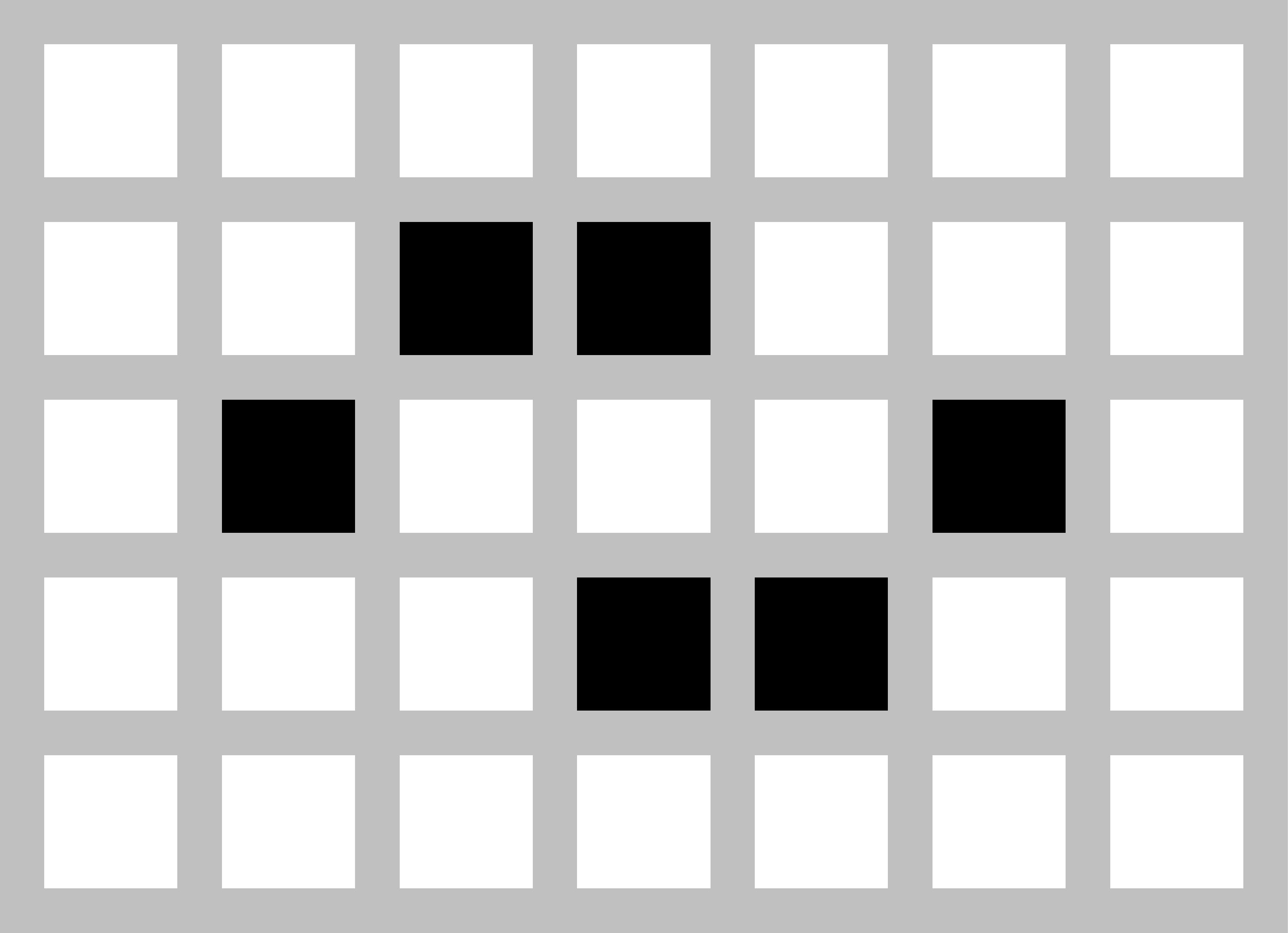} & $1.053 \times 10^{-4}$ \\
		$\mathbf{20}$ & \includegraphics[scale=0.02]{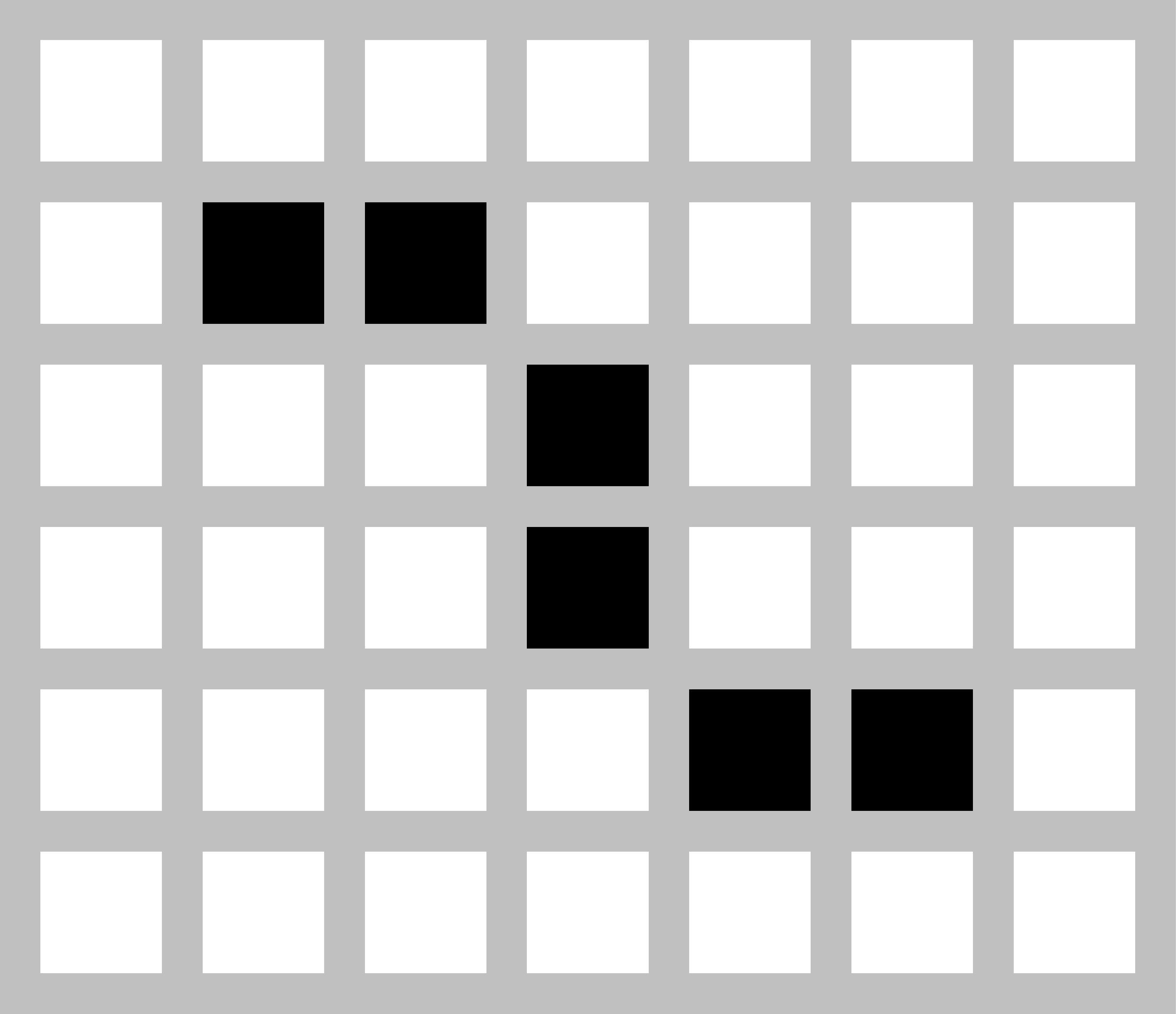} & $7.571 \times 10^{-5}$ \\
		\noalign{\smallskip}\hline
		\end{tabular}
		\end{tabular}
		\end{table}
  	
		Observe that the fact that cells stay alive if they have only one live neighbour results in an abundance of small still lifes, most of which are made up of islands, with each island being a chain or loop of a few cells. This leads to a simple grammar for constructing large still lifes -- see Figure~\ref{fig:2x2_still_life}. The number of distinct strict still lifes with $n$ cells for $n = 1, 2, 3, \ldots$ is given by $0,2,1,3,4,9,10,27,48,126,\ldots$\footnote{Sloane's A166476 -- the still lifes with $9$ or fewer cells are shown in Appendix I}. Compare this with the corresponding sequence for Life, which is $0, 0, 0, 2, 1, 5, 4, 9, 10, 25,\ldots$\footnote{Sloane's A019473}.
  	
  	The oscillators that occur naturally in 2x2 do not occur in Life. The majority of common oscillators have period 2 or 4, but some small patterns give rise to very high-period oscillators. For example, the fourth most common oscillator is simply the stairstep hexomino in one of its phases, yet it has period 26. The thirteenth most common oscillator, which it might be appropriate to name the ``decathlon,'' has period 10 and evolves out of a horizontal row of $5$ adjacent cells, much like the period 15 ``pentadecathlon'' of Life evolves out of a horizontal row of $10$ adjacent cells. Oscillators with periods 14 and 22 are also relatively frequent, as demonstrated by Table~\ref{tab:common_2x2_osc}.\footnote{Based on data from the \emph{Online Life-Like CA Soup Search}. A total of $11,270,020$ (non-distinct) oscillators were catalogued.}

		\begin{table}
		\center
		\caption{The $20$ most common naturally-occurring oscillators in the 2x2 rule and their approximate frequency (out of $1.000$) relative to all oscillators}
		\label{tab:common_2x2_osc}
		\begin{tabular}{cc}
		\begin{tabular}[t]{lccc}
		\hline\noalign{\smallskip}
		{\bf \#} \ & \ {\bf Pattern} \ & \ {\bf Period} \ & \ {\bf Rel. Freq.} \\
		\noalign{\smallskip}\hline\noalign{\smallskip}
		$\mathbf{1}$ & \includegraphics[scale=0.02]{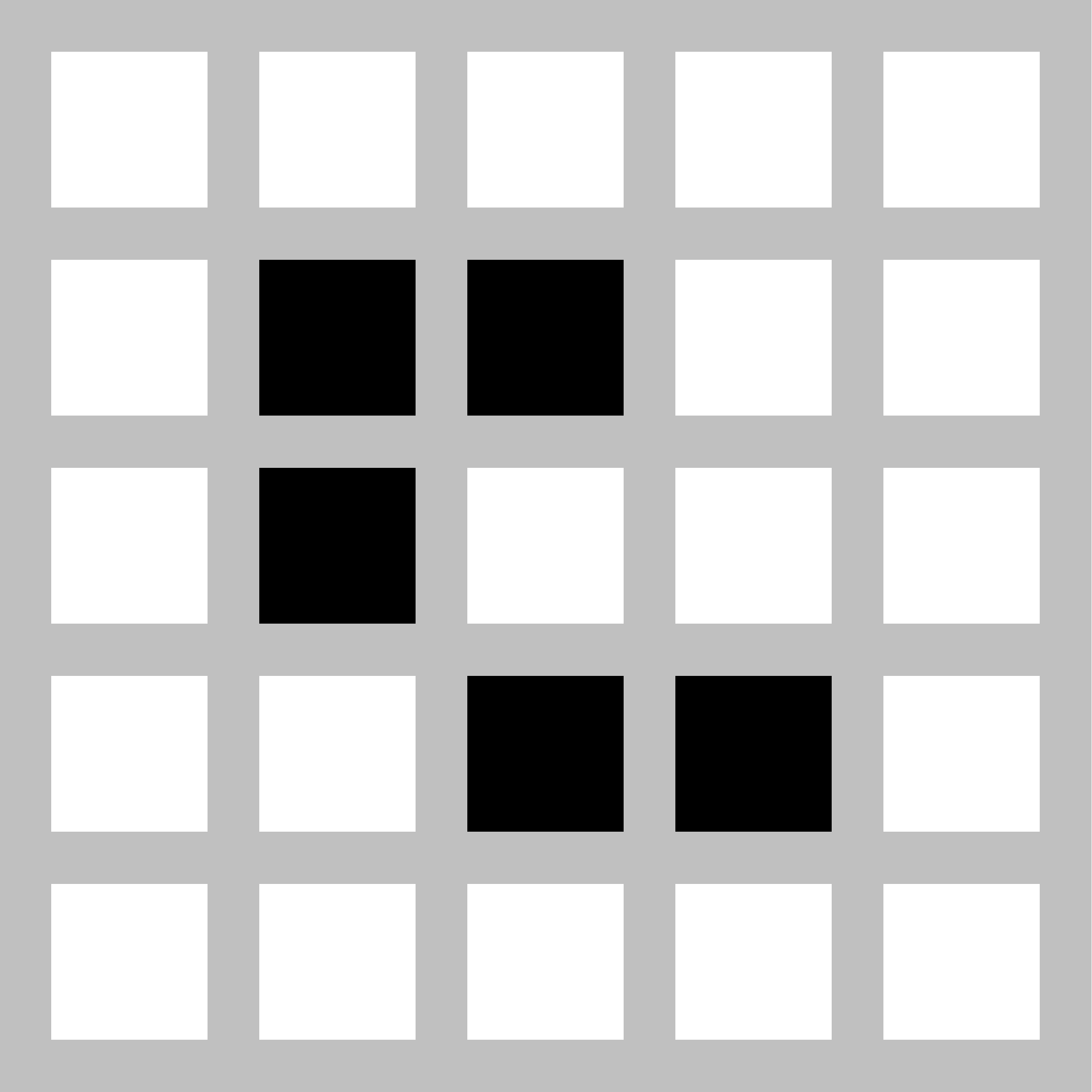} & $2$ & $4.824 \times 10^{-1}$ \\
		$\mathbf{2}$ & \includegraphics[scale=0.02]{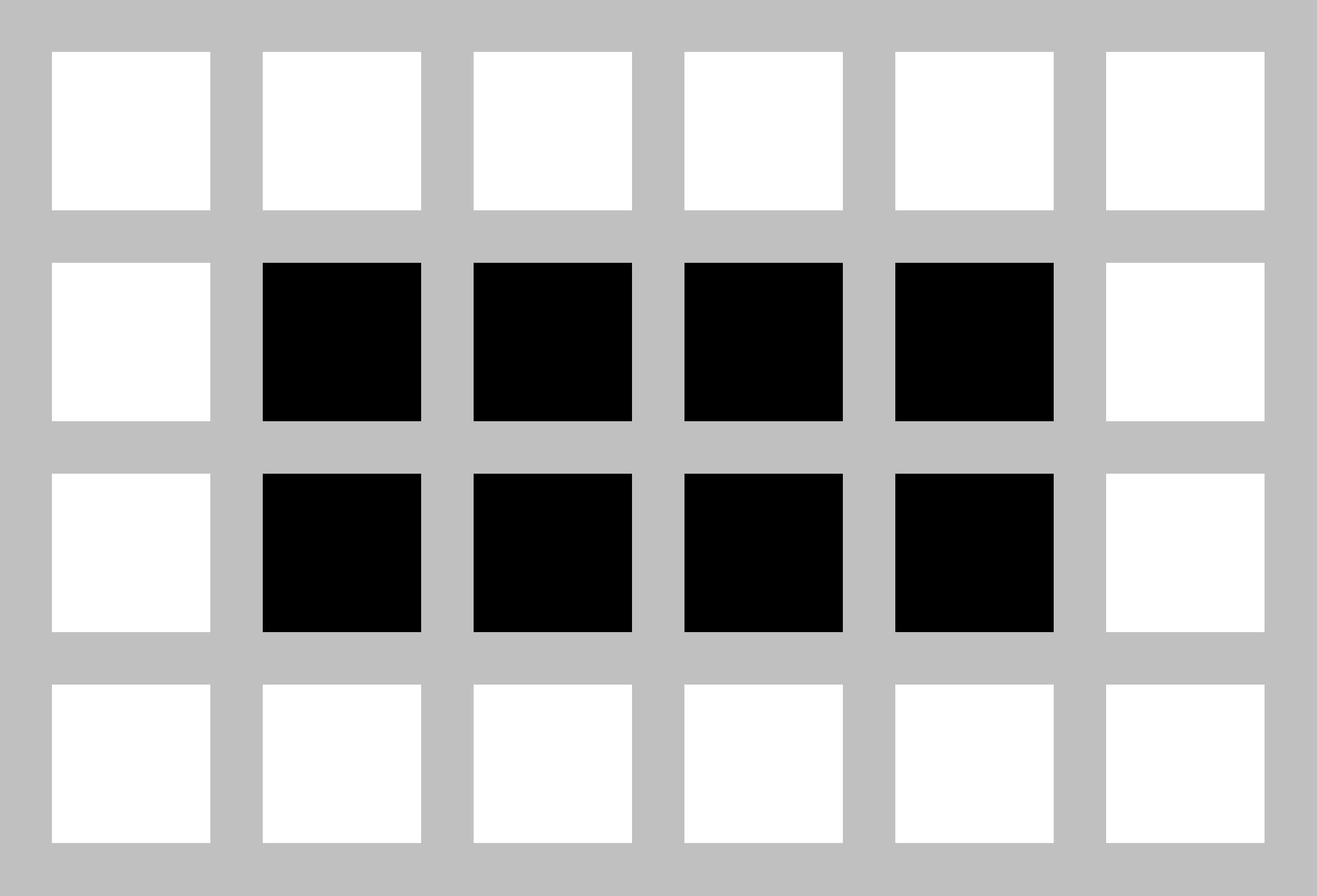} & $2$ & $2.170 \times 10^{-1}$ \\
		$\mathbf{3}$ & \includegraphics[scale=0.02]{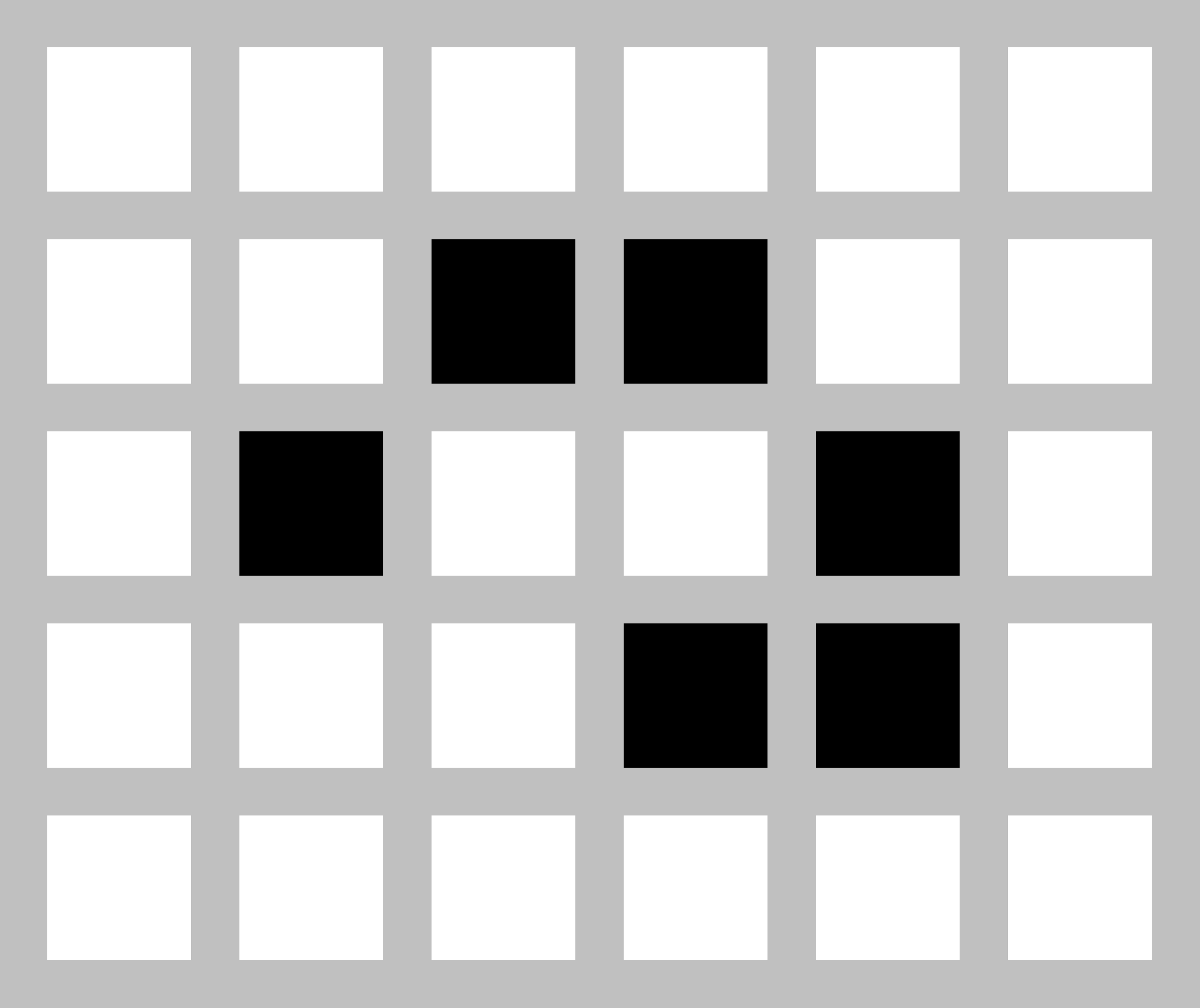} & $2$ & $5.741 \times 10^{-2}$ \\
		$\mathbf{4}$ & \includegraphics[scale=0.02]{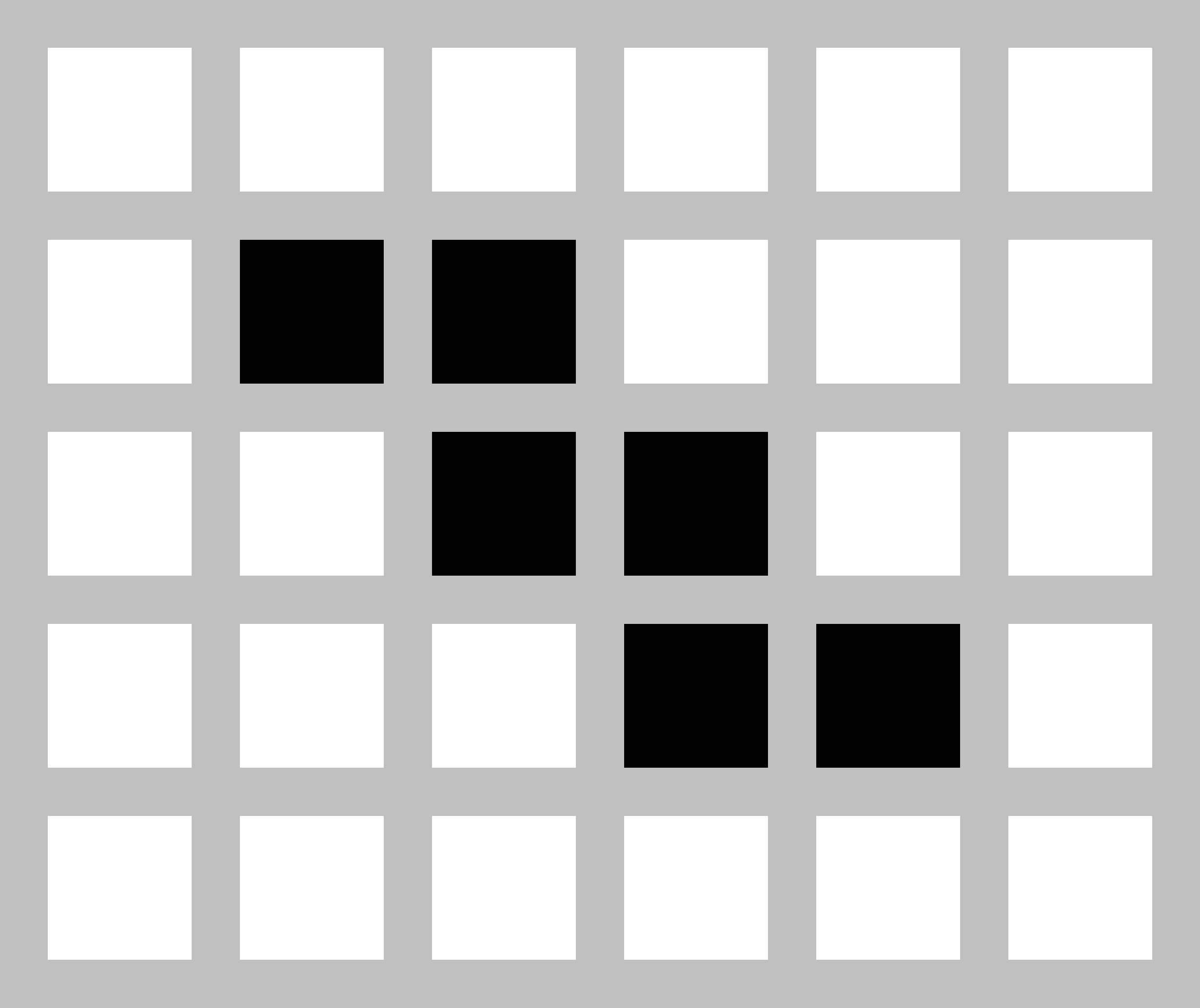} & $26$ & $5.515 \times 10^{-2}$ \\
		$\mathbf{5}$ & \includegraphics[scale=0.02]{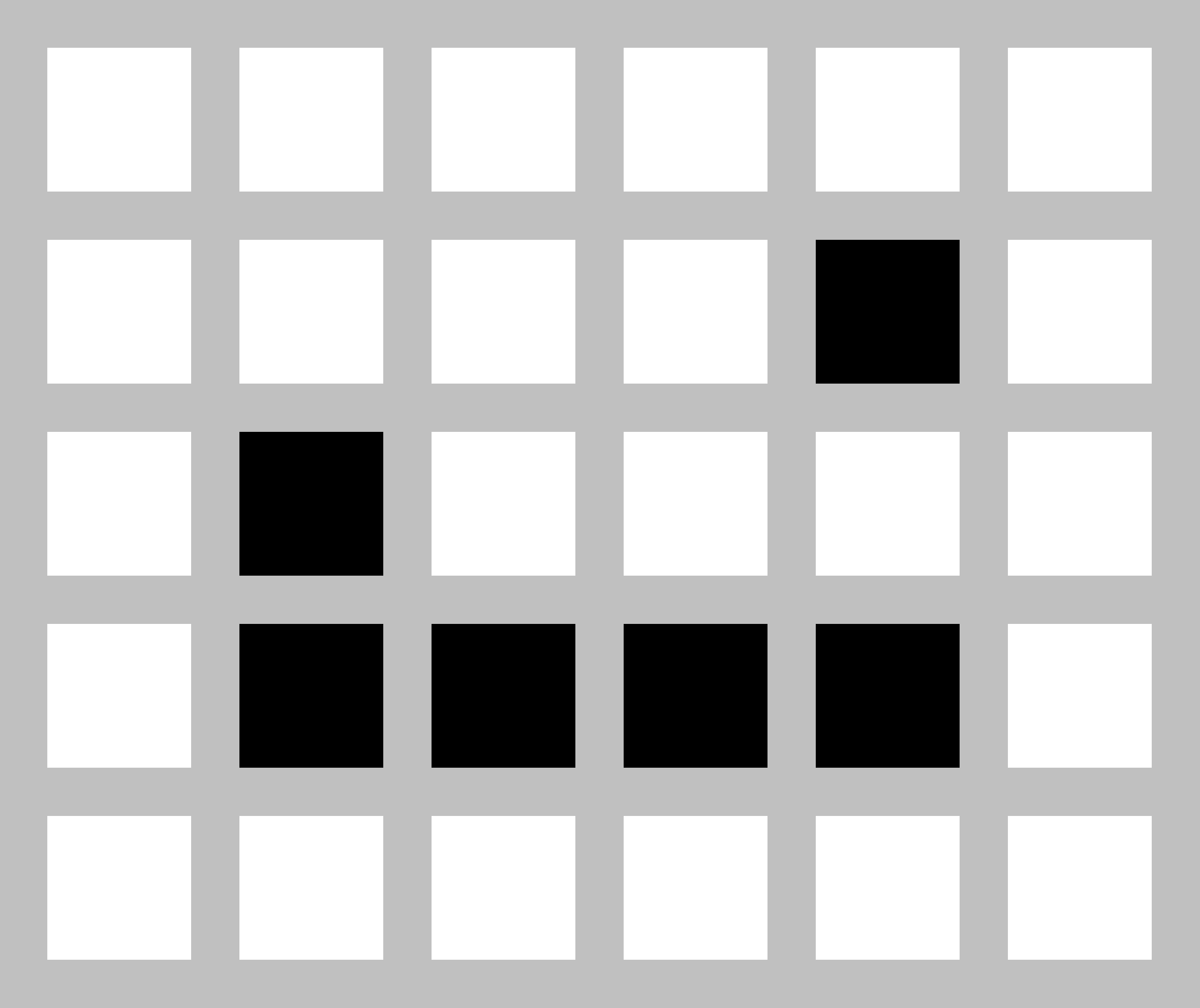} & $4$ & $3.718 \times 10^{-2}$ \\
		$\mathbf{6}$ & \includegraphics[scale=0.02]{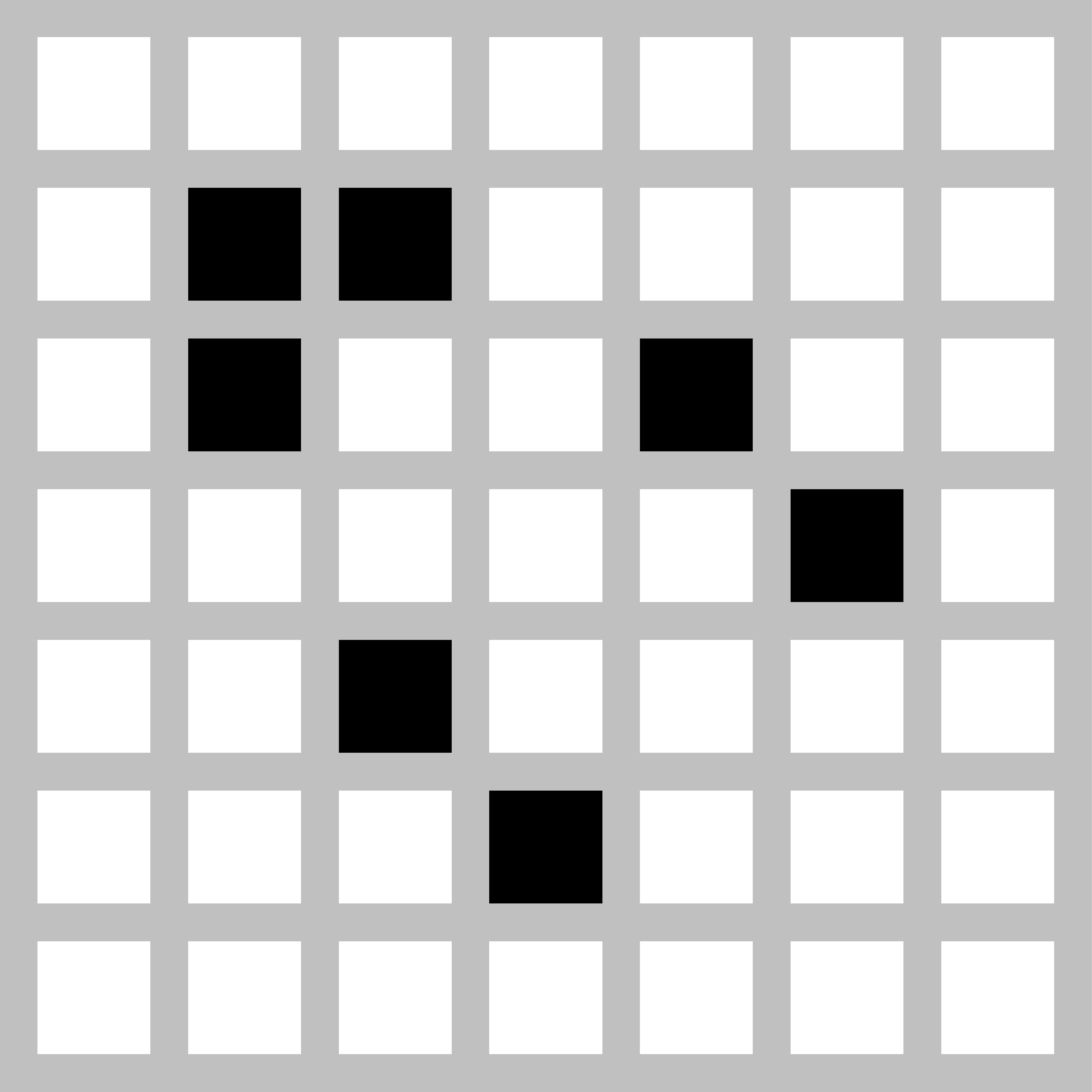} & $14$ & $3.364 \times 10^{-2}$ \\
		$\mathbf{7}$ & \includegraphics[scale=0.02]{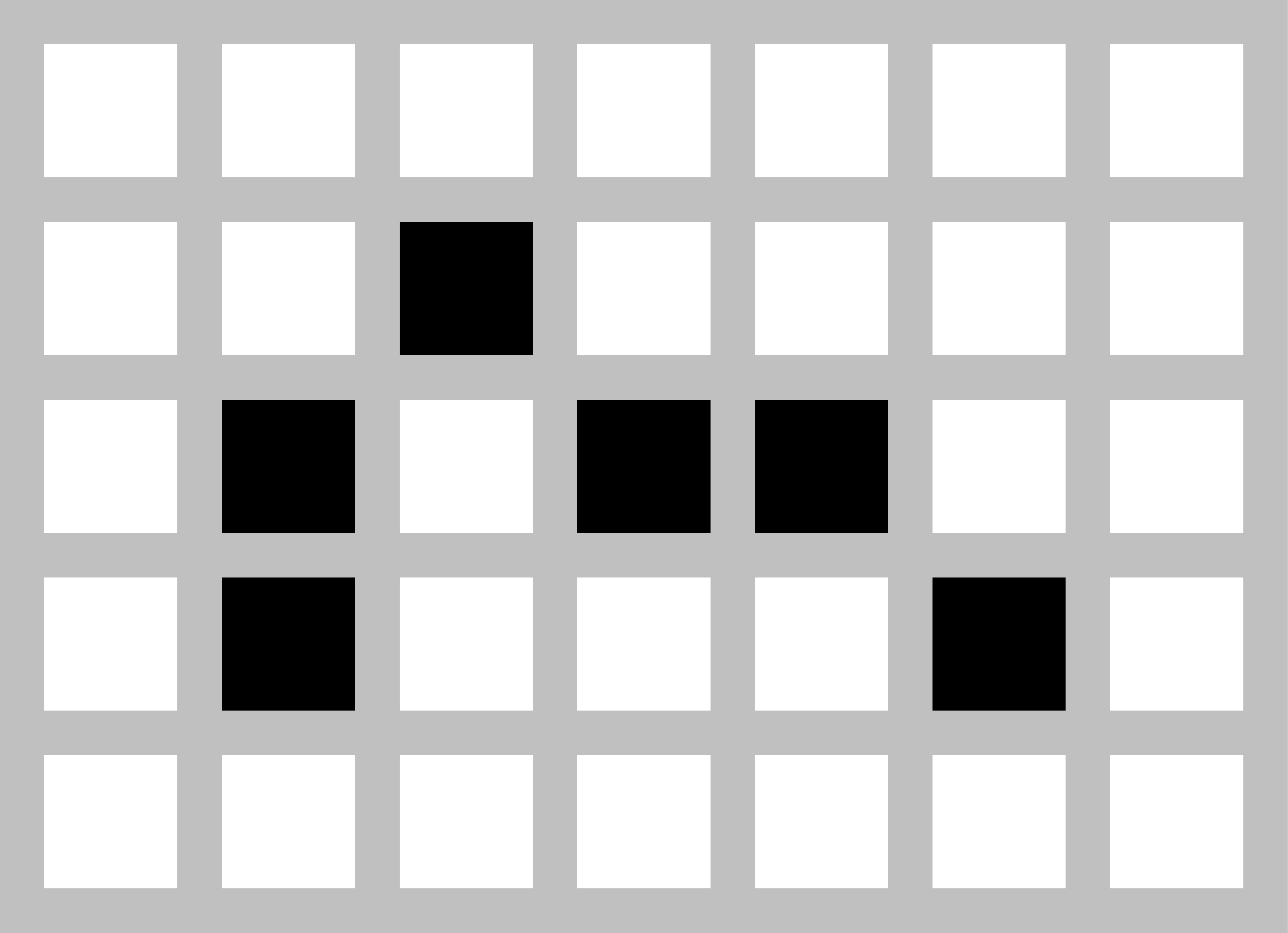} & $4$ & $3.104 \times 10^{-2}$ \\
		$\mathbf{8}$ & \includegraphics[scale=0.02]{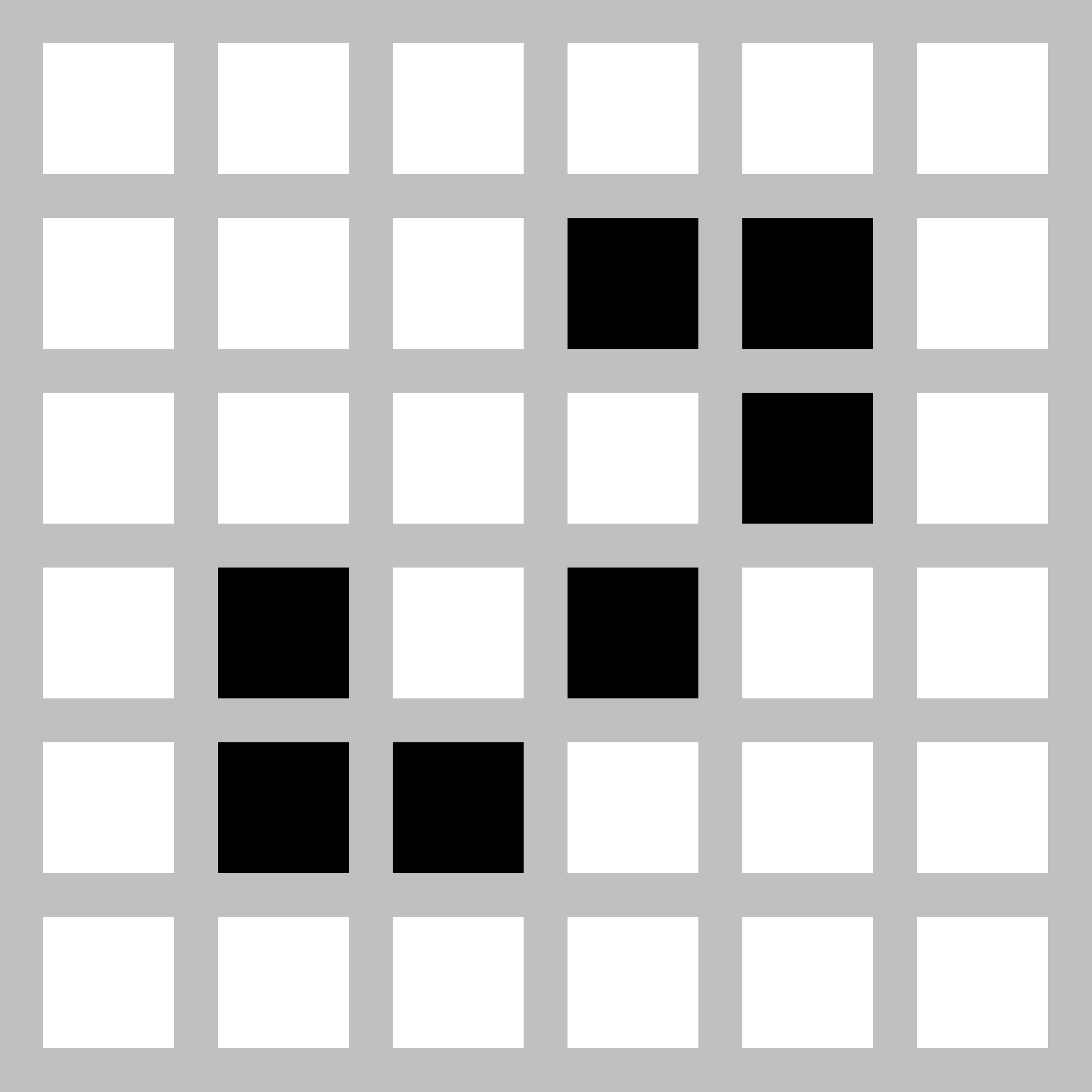} & $2$ & $1.795 \times 10^{-2}$ \\
		$\mathbf{9}$ & \includegraphics[scale=0.02]{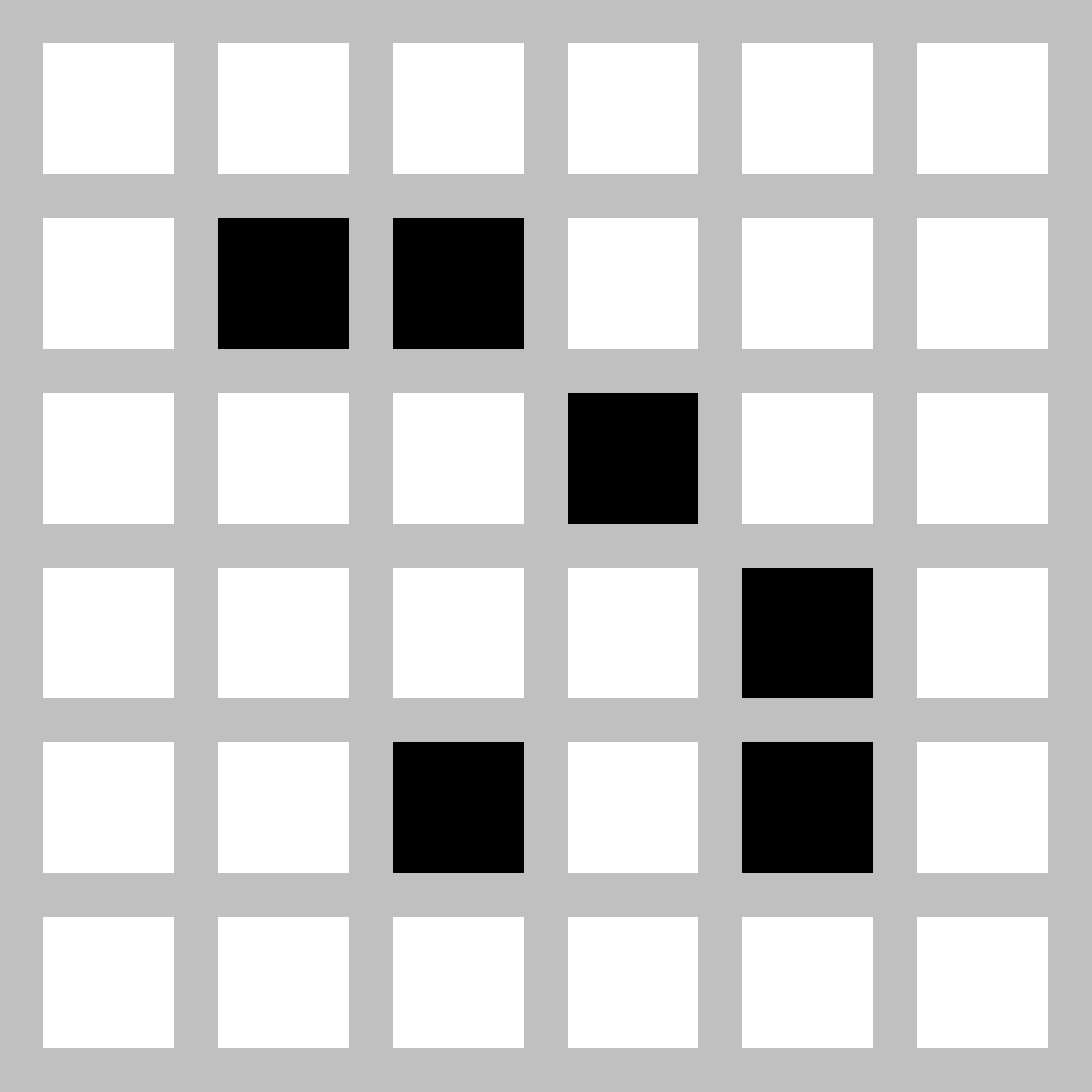} & $4$ & $1.766 \times 10^{-2}$ \\
		$\mathbf{10}$ & \includegraphics[scale=0.02]{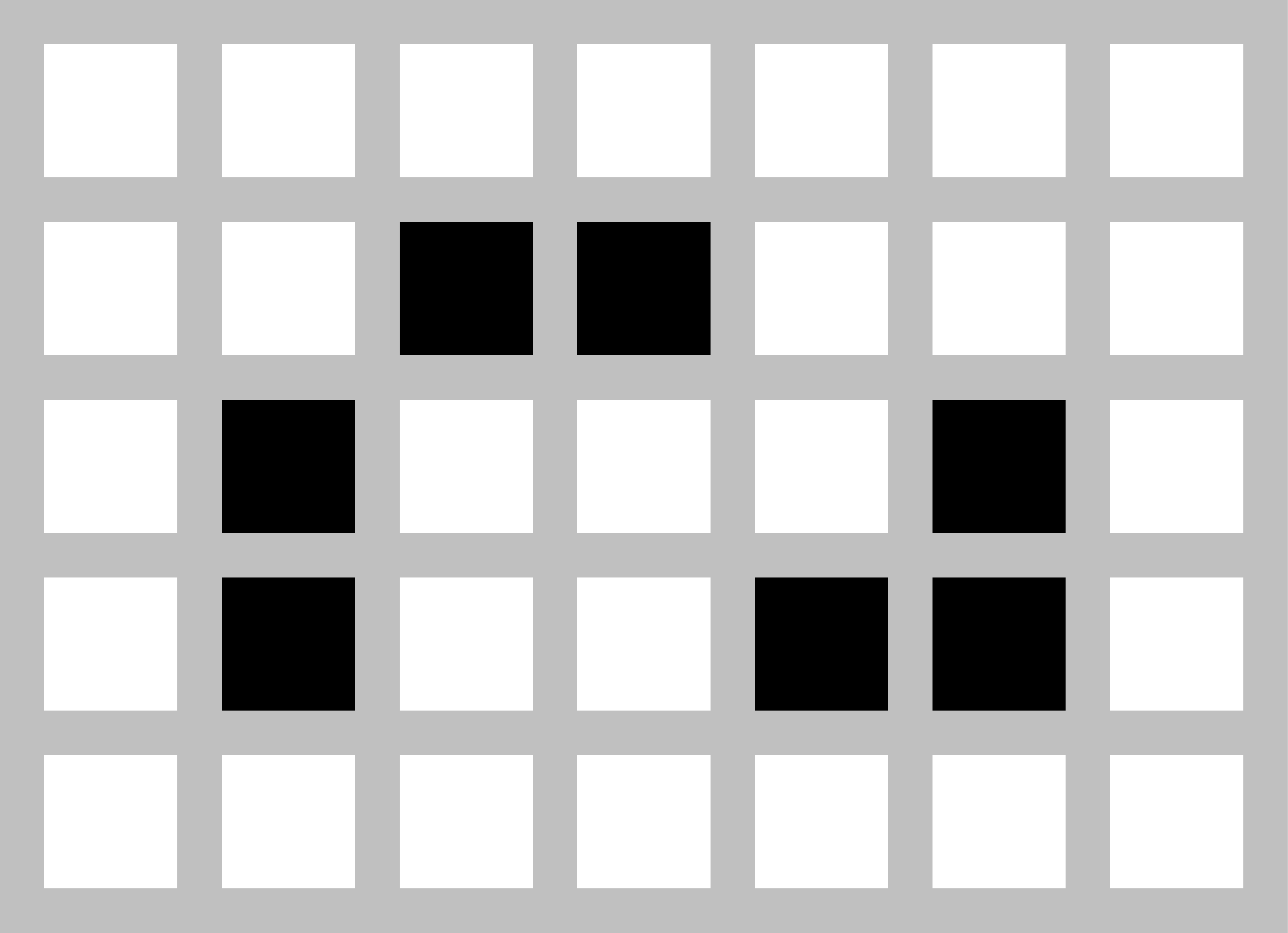} & $4$ & $1.745 \times 10^{-2}$ \\
		\noalign{\smallskip}\hline
		\end{tabular} \ \ & \ \ \begin{tabular}[t]{lccc}
		\hline\noalign{\smallskip}
		{\bf \#} \ & \ {\bf Pattern} \ & \ {\bf Period} \ & \ {\bf Rel. Freq.} \\
		\noalign{\smallskip}\hline\noalign{\smallskip}
		$\mathbf{11}$ & \includegraphics[scale=0.02]{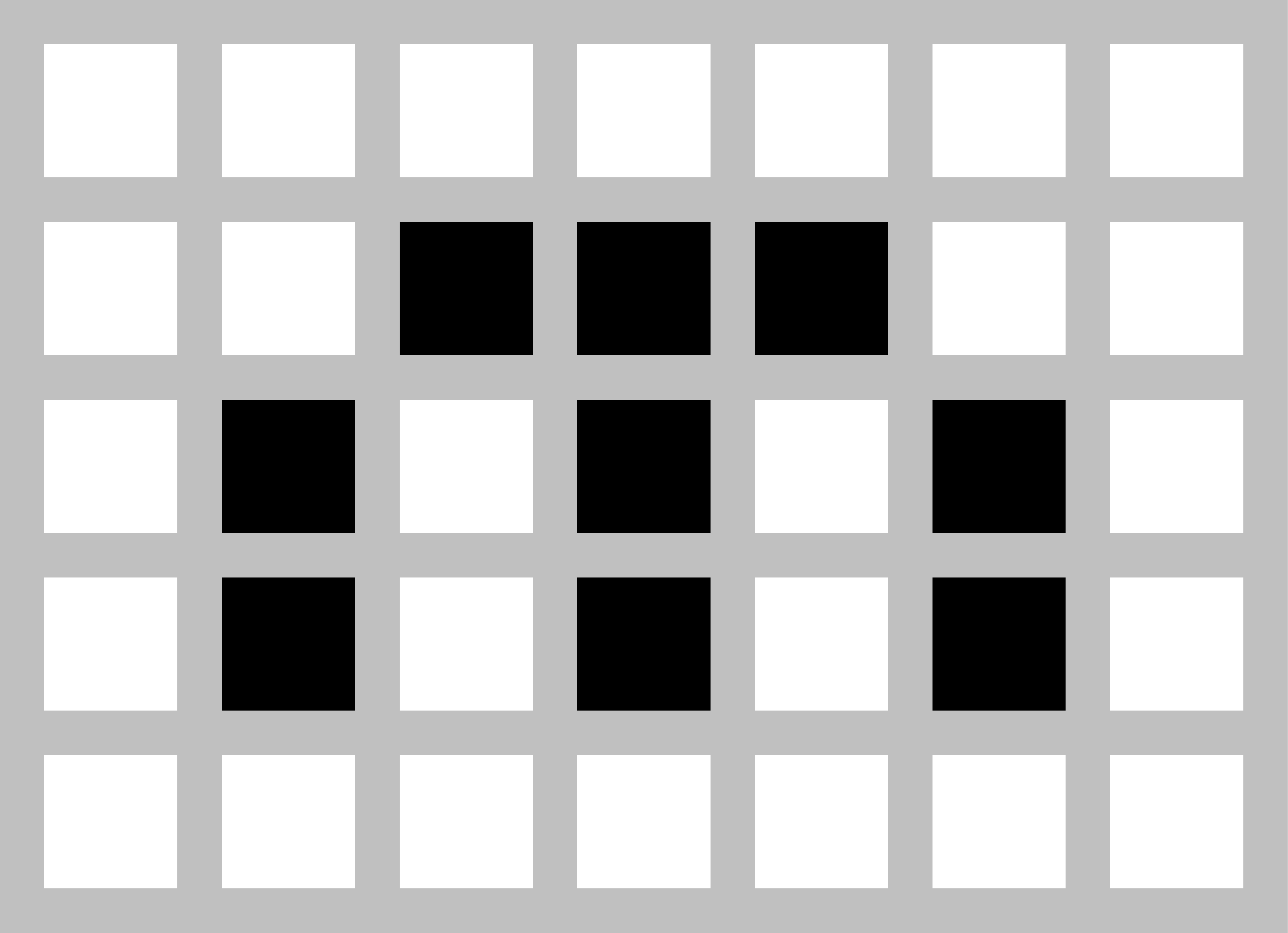} & $2$ & $9.330 \times 10^{-3}$ \\
		$\mathbf{12}$ & \includegraphics[scale=0.02]{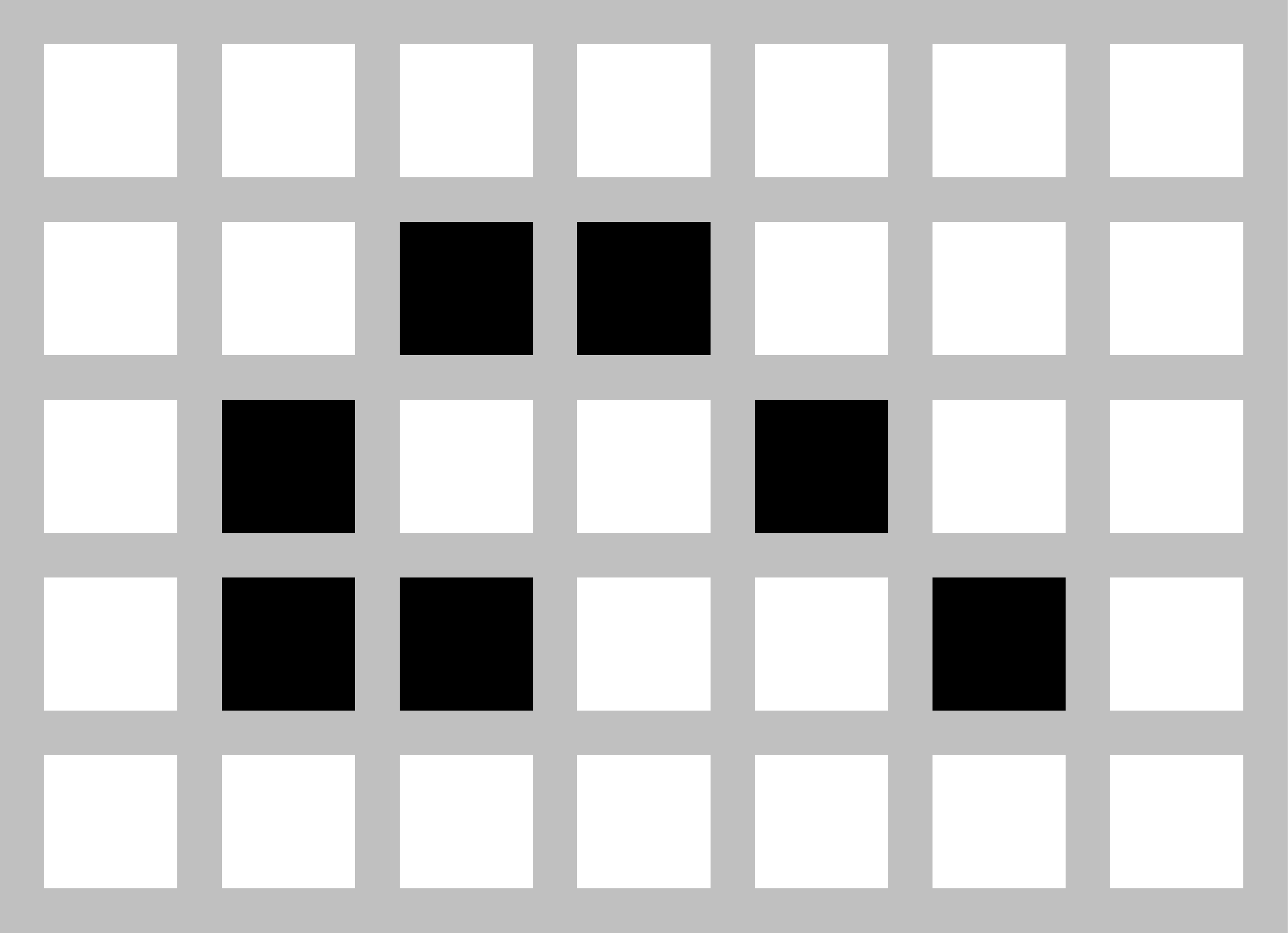} & $2$ & $7.766 \times 10^{-3}$ \\
		$\mathbf{13}$ & \includegraphics[scale=0.02]{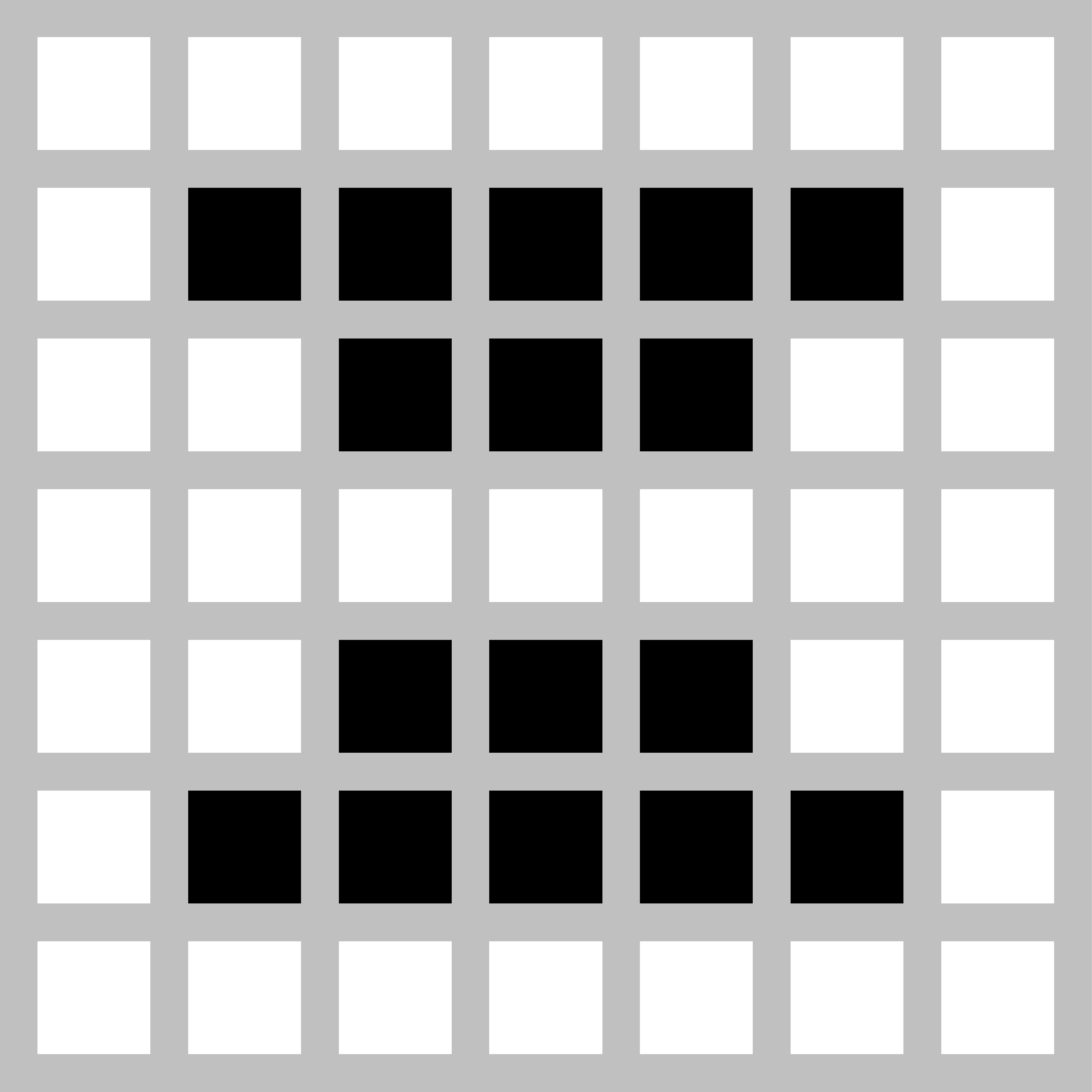} & $10$ & $5.188 \times 10^{-3}$ \\
		$\mathbf{14}$ & \includegraphics[scale=0.02]{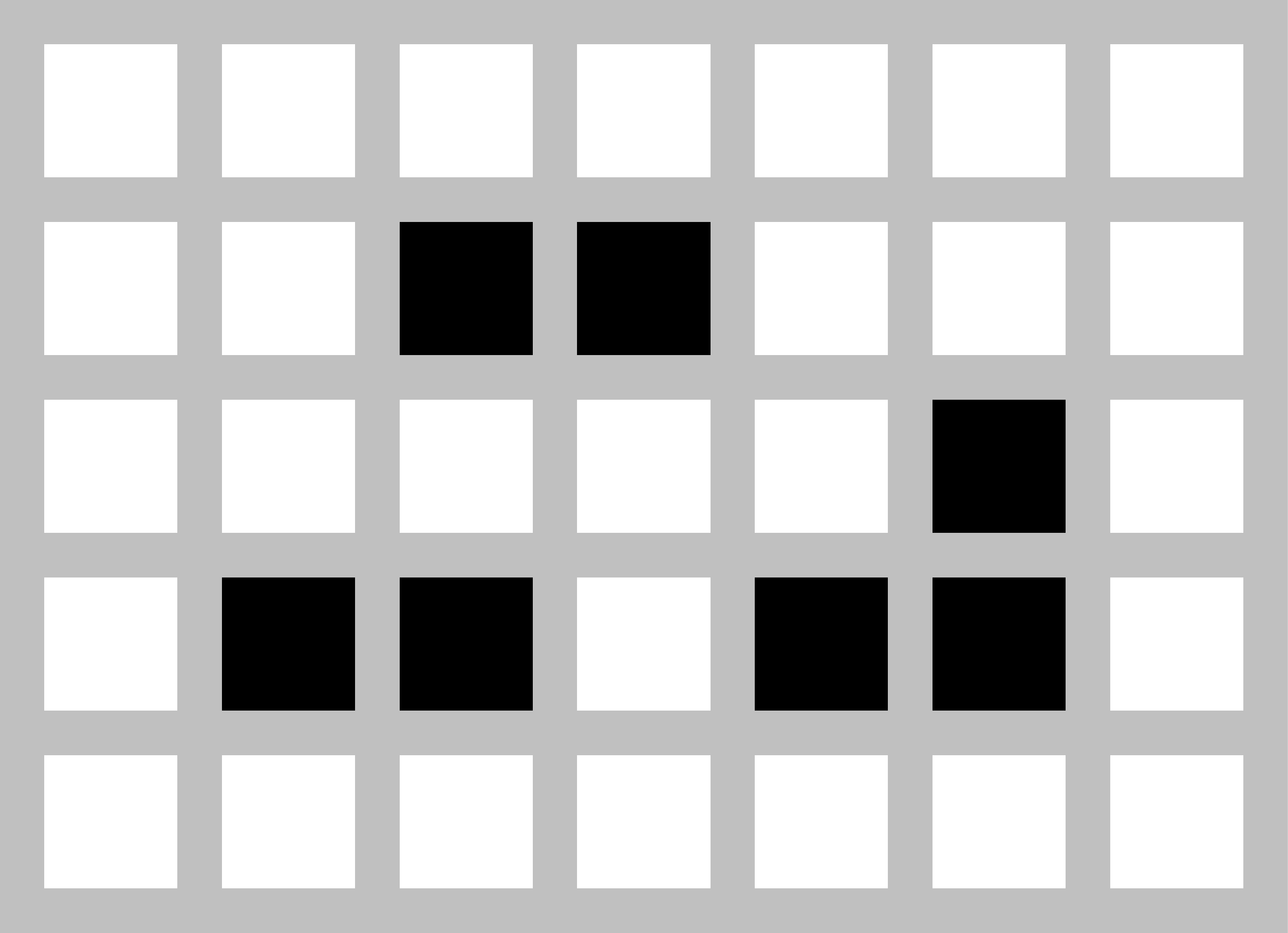} & $2$ & $2.042 \times 10^{-3}$ \\
		$\mathbf{15}$ & \includegraphics[scale=0.02]{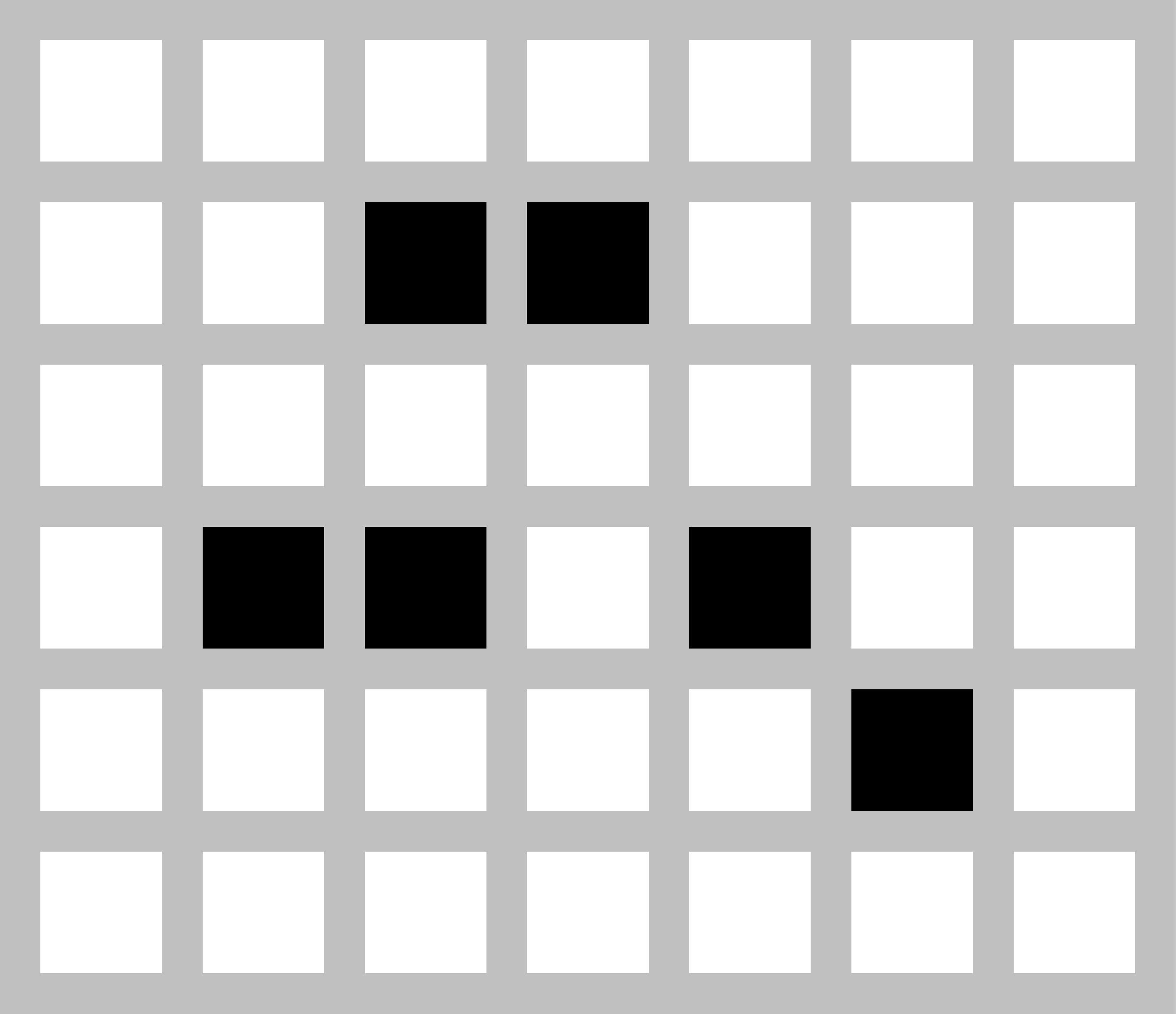} & $2$ & $1.633 \times 10^{-3}$ \\
		$\mathbf{16}$ & \includegraphics[scale=0.02]{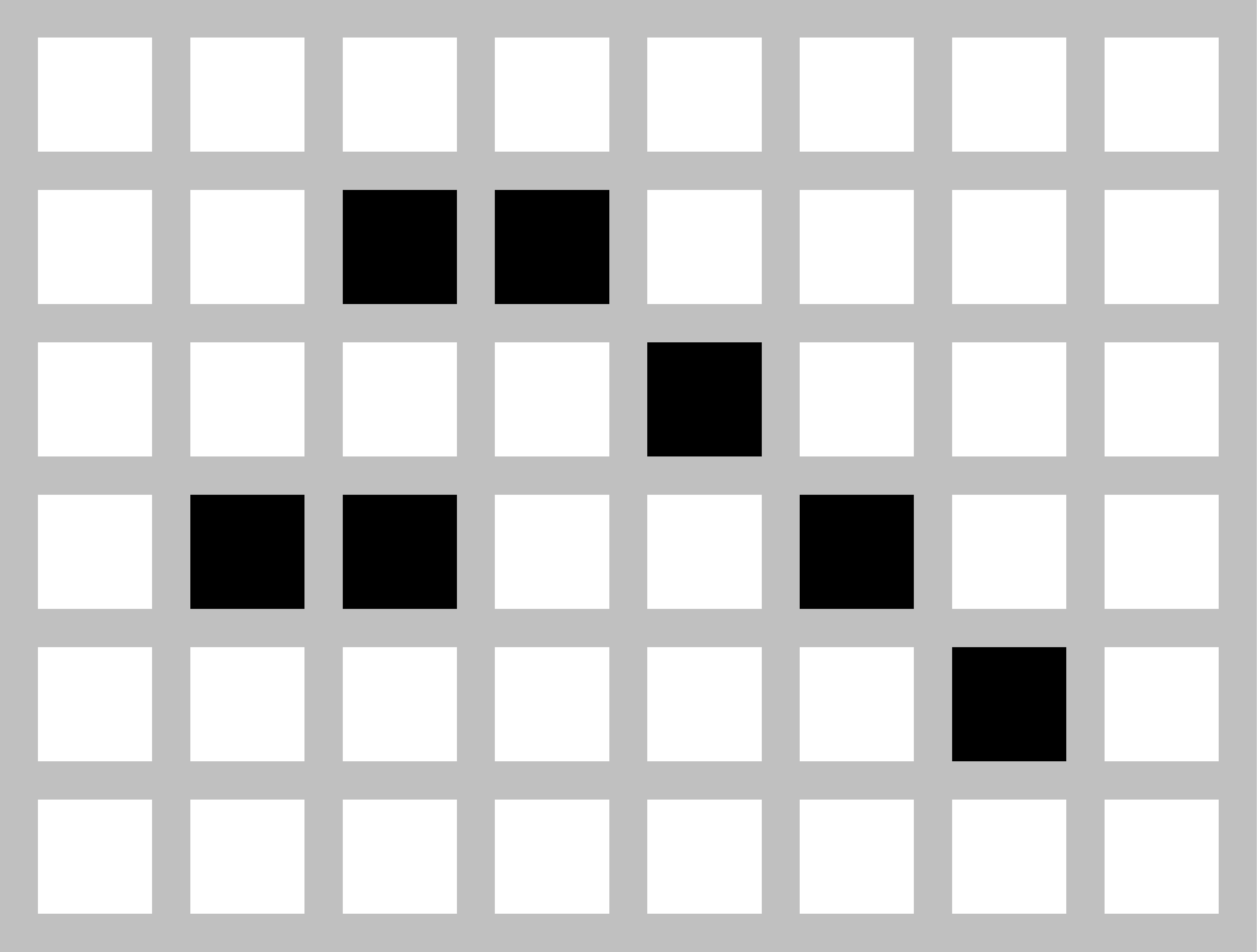} & $2$ & $1.559 \times 10^{-3}$ \\
		$\mathbf{17}$ & \includegraphics[scale=0.02]{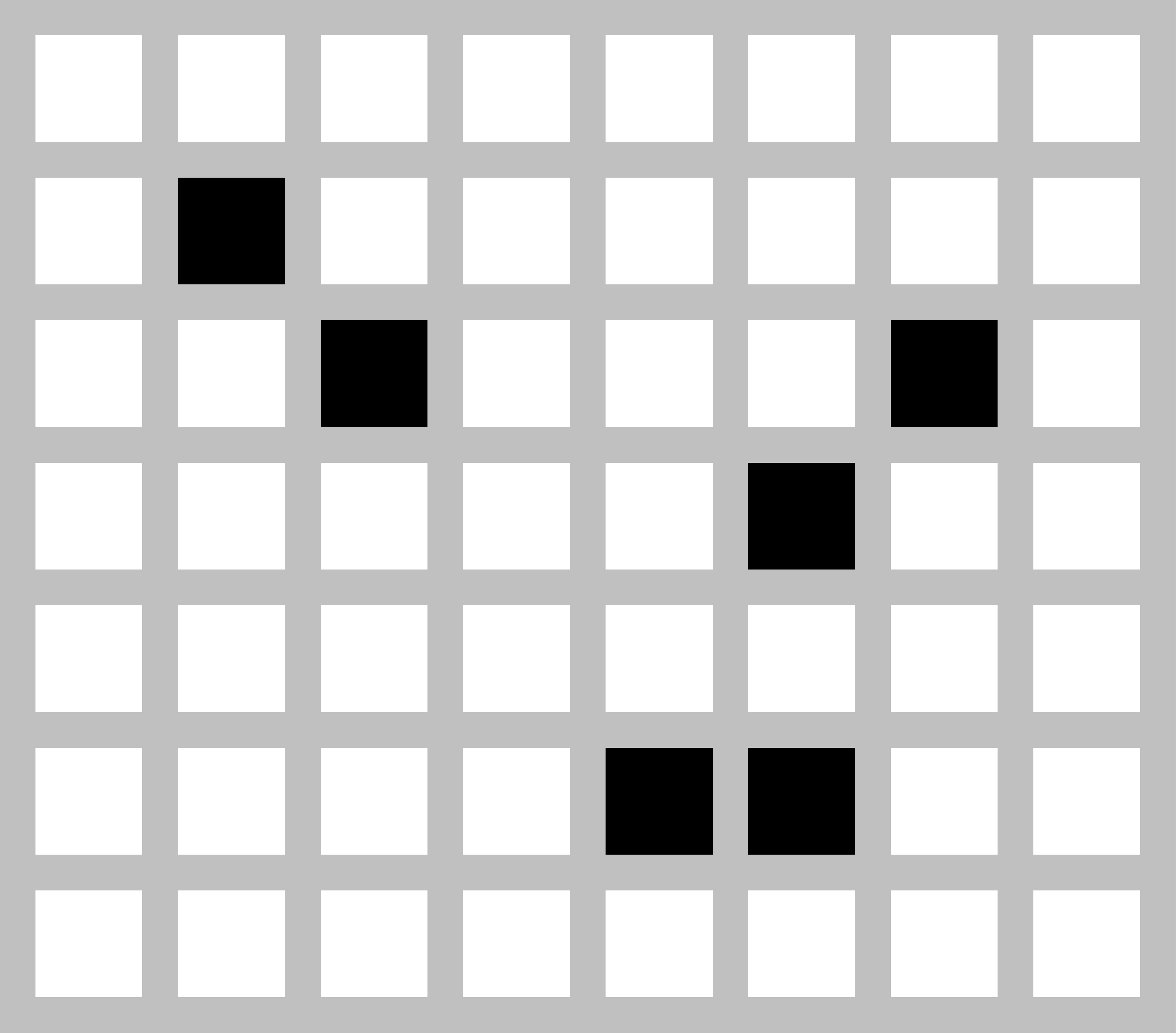} & $14$ & $1.182 \times 10^{-3}$ \\
		$\mathbf{18}$ & \includegraphics[scale=0.02]{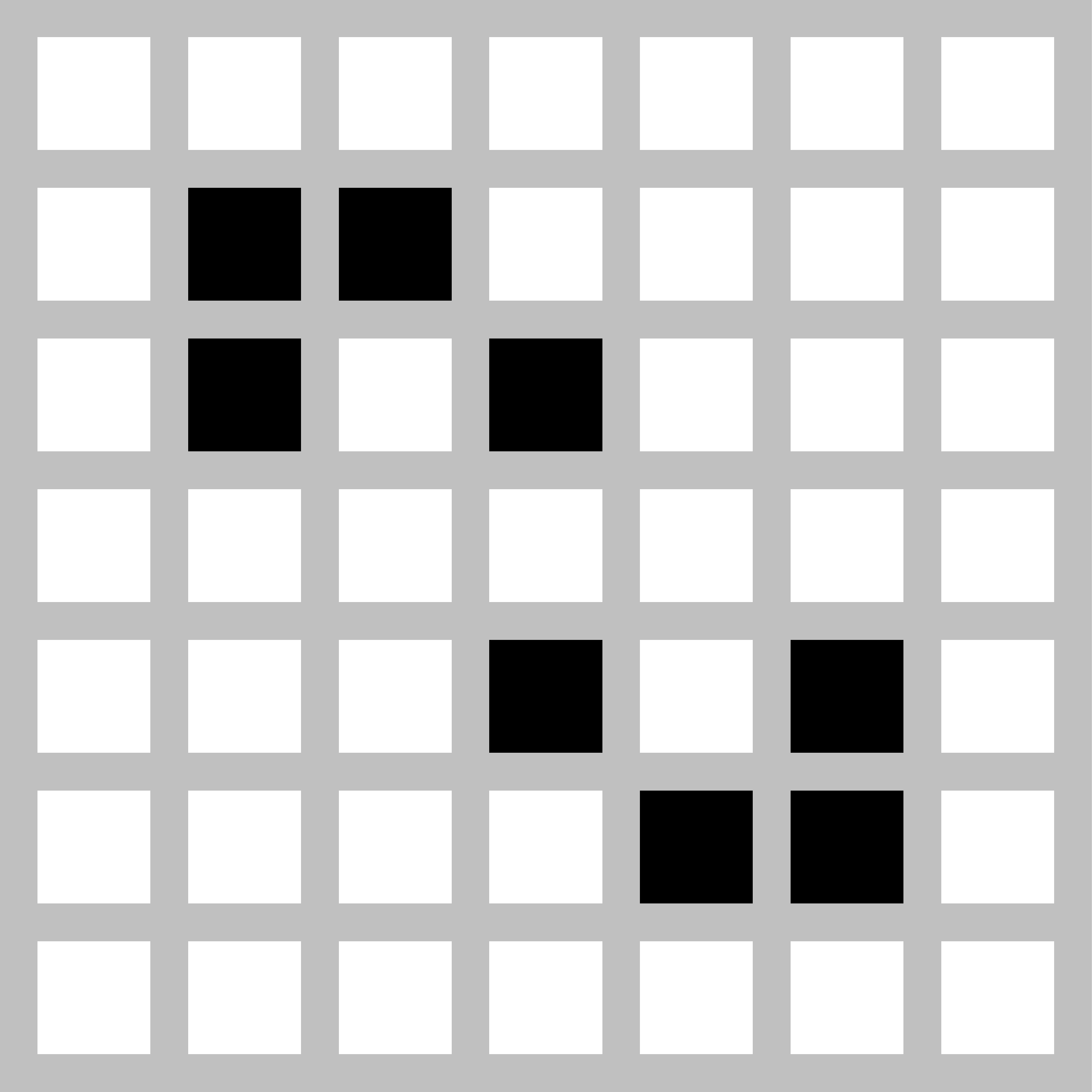} & $2$ & $9.618 \times 10^{-4}$ \\
		$\mathbf{19}$ & \includegraphics[scale=0.02]{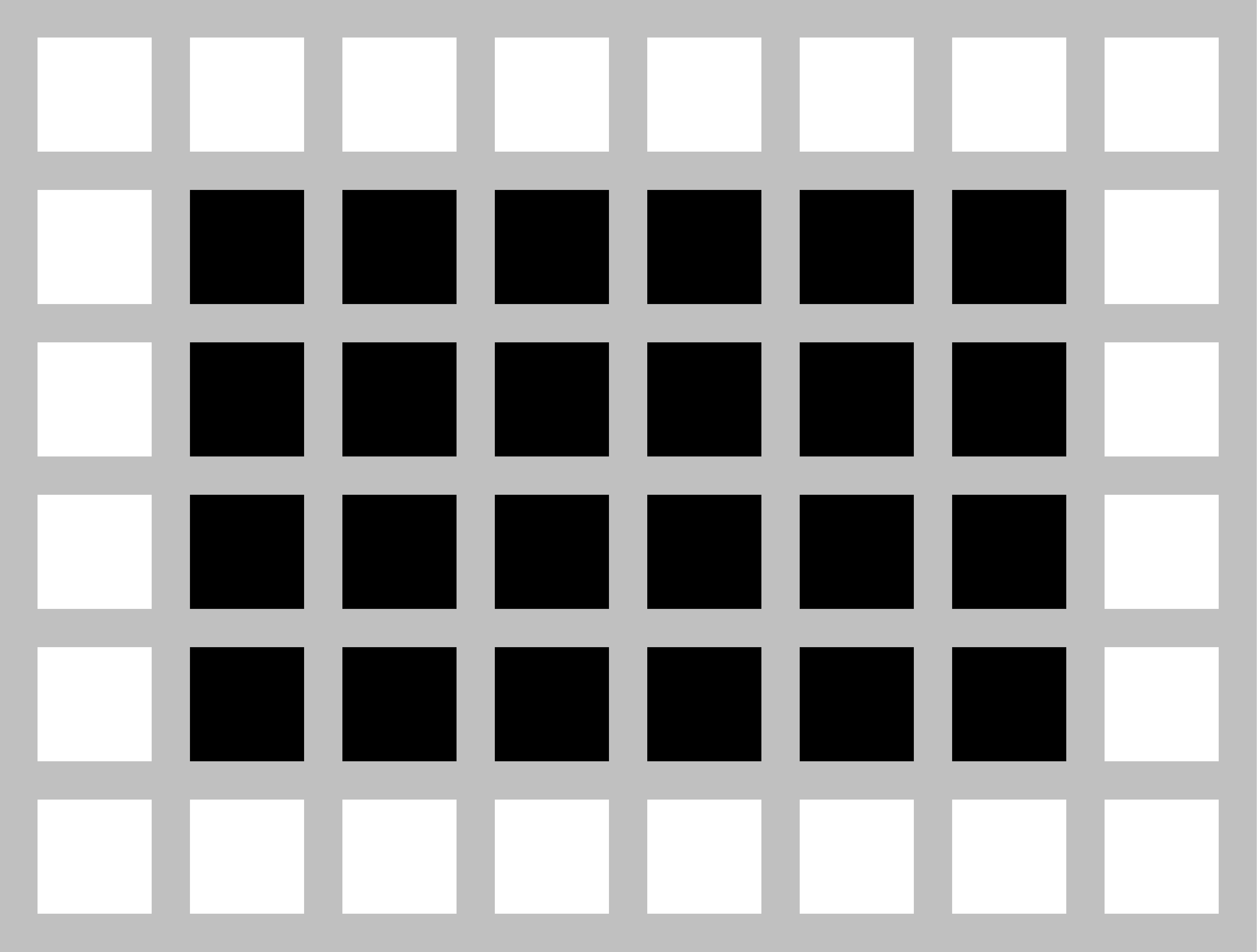} & $6$ & $9.539 \times 10^{-4}$ \\
		$\mathbf{20}$ & \includegraphics[scale=0.02]{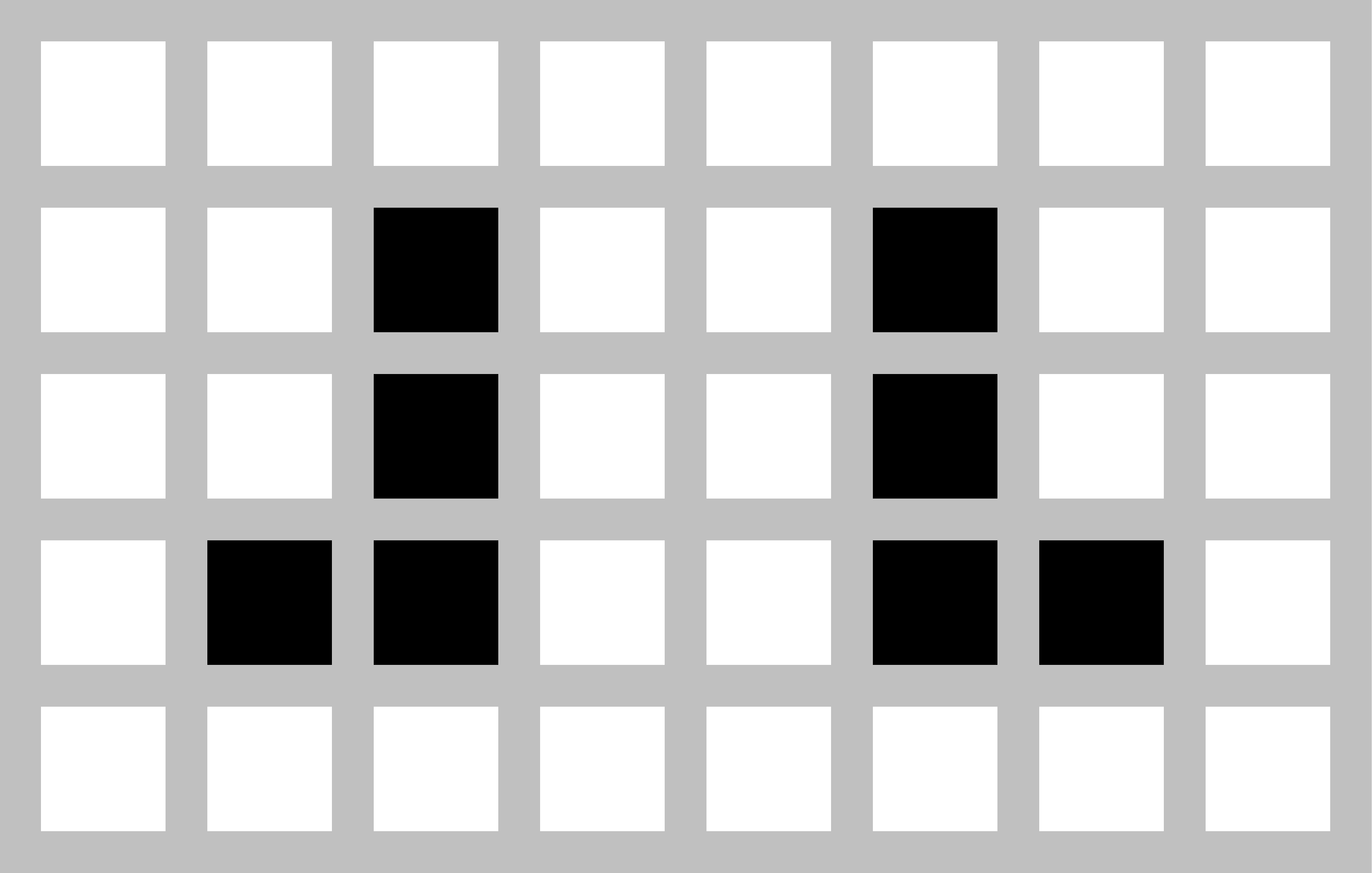} & $22$ & $4.423 \times 10^{-4}$ \\
		\noalign{\smallskip}\hline
		\end{tabular}
		\end{tabular}
		\end{table}

		The only other particularly notable patterns that have been known to appear spontaneously from random soup are a fairly common $c/8$ diagonal glider and a related $c/8$ wickstretcher. Although the glider itself was known of by no later than 1993, the wickstretcher, which works simply by placing the glider next to a diagonal wick, was not found until June 2009. In fact, the $c/8$ wickstretcher and its slight modifications are currently the only known infinitely-growing patterns in 2x2.

  \section{Oscillators and Spaceships}\label{sec:patterns}
	
		Beyond the standard oscillators that appear naturally, many oscillators of period 2 through 4 have been constructed by hand and computer search by Alan Hensel, Dean Hickerson, and Lewis Patterson over the years. In 1993, Hensel discovered the first known oscillator with odd period, the small period 5 pattern shown in Figure~\ref{fig:2x2_odd_period}. David Bell soon thereafter noticed that it can be combined with itself and extended in a variety of different ways, creating the first known extensible oscillator. Hensel discovered the first known period 3 oscillator in 1994, and Hickerson found the first known period 11 and 17 oscillators later that same year. To date, the only odd periods for which there are known oscillators are 3, 5, 11, and 17, which shows that odd-period oscillators seem to be much more difficult to construct in this rule than even-period oscillators. We will see later that there is an infinite family of even periods that are easily realized by simple block oscillators. The least period for which there is no known oscillator is 7, while the least even period for which there is no known oscillator is 18.

		\begin{figure}[ht]
		\center
		\includegraphics[width=0.9\textwidth]{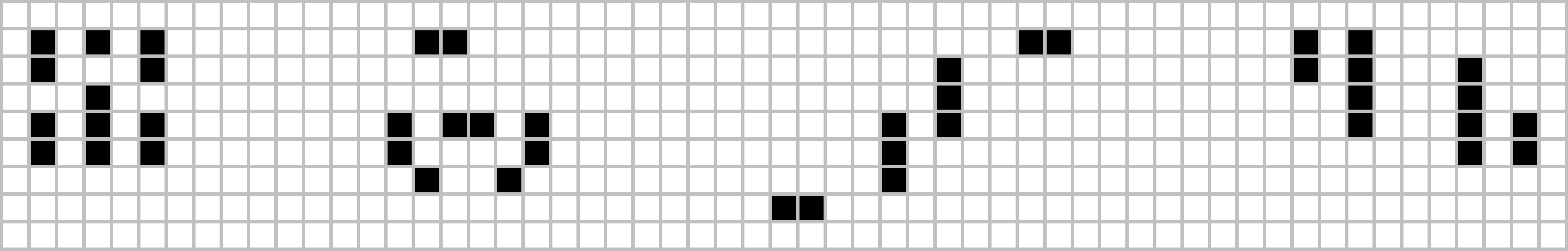}
		\caption{The first known oscillators of period 3, 5, 11, and 17, respectively}
		\label{fig:2x2_odd_period}
		\end{figure}
						
			The $c/8$ diagonal glider is the only spaceship that has ever been seen to occur naturally in 2x2, though several others have been found via computer search. In February 1994, Hensel found the first such spaceship -- the $c/3$ orthogonal pattern shown in Figure~\ref{fig:2x2_c3spaceship}. He also found the next three spaceships, which were orthogonal with speed $c/3$, $c/3$, and $c/4$. As a result of the relative ease of finding slow spaceships in this rule, it was initially suspected that it does not contain spaceships that travel as fast as their Life counterparts ($c/2$ orthogonally and $c/4$ diagonally). However, David Eppstein found a $c/2$ orthogonal spaceship in October 1998 using his \verb|gfind| program \cite{E02}, and several others have been found since then (see Appendix I). Hickerson found the first $c/4$ diagonal spaceship in 1999. Eppstein has since found a $c/3$ diagonal spaceship, which shows that it is possible for spaceships in 2x2 to travel \emph{faster} than spaceships in Life.

\begin{figure}[ht]
\begin{minipage}[t]{0.485\textwidth}
\centering
		  \includegraphics[width=1.1in]{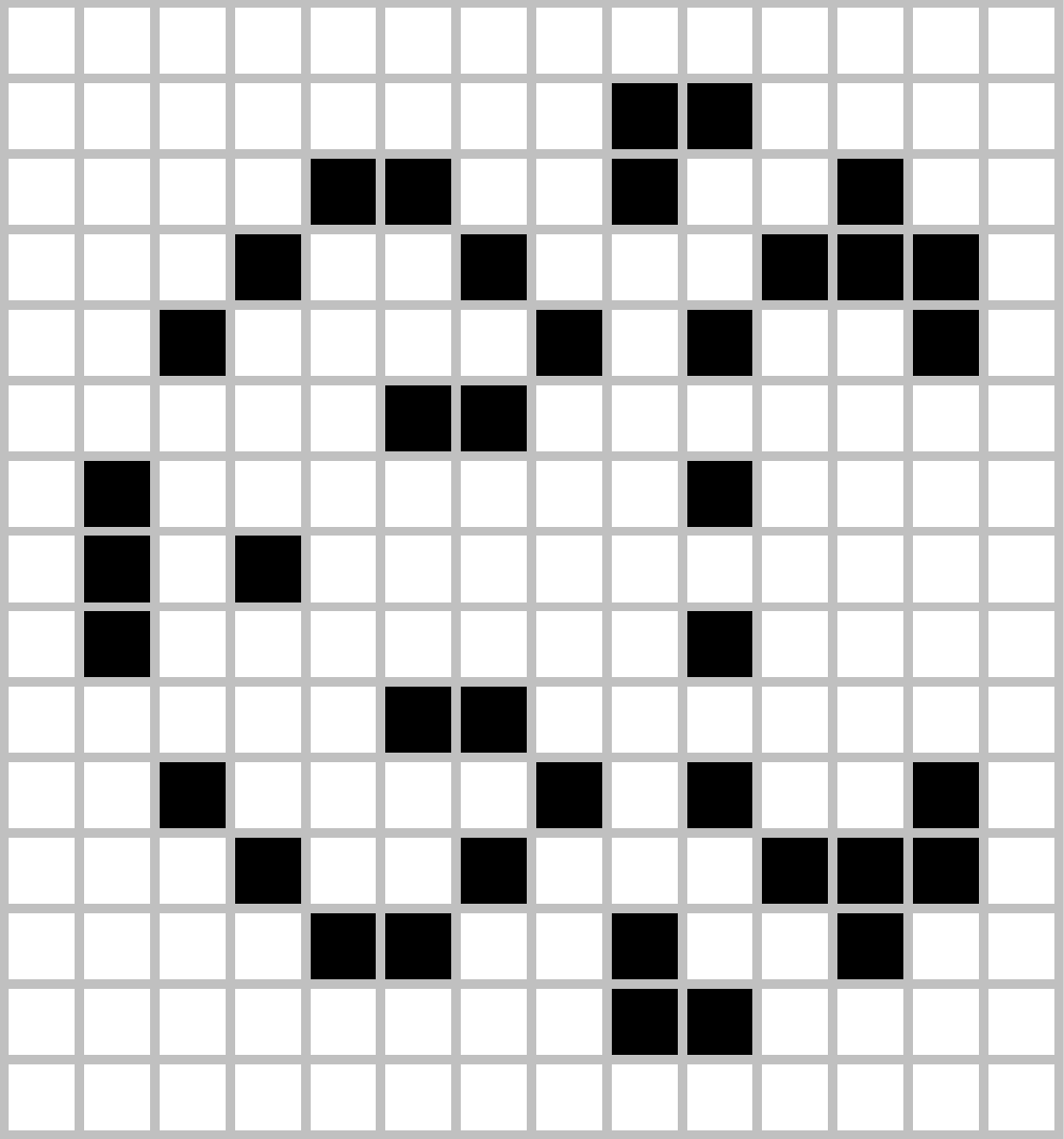}
		  \caption{The first discovered spaceship other than the $c/8$ glider}\label{fig:2x2_c3spaceship} 
\end{minipage}
\hspace{0.03\textwidth}
\begin{minipage}[t]{0.485\textwidth}
\centering
		  \includegraphics[width=1.2in]{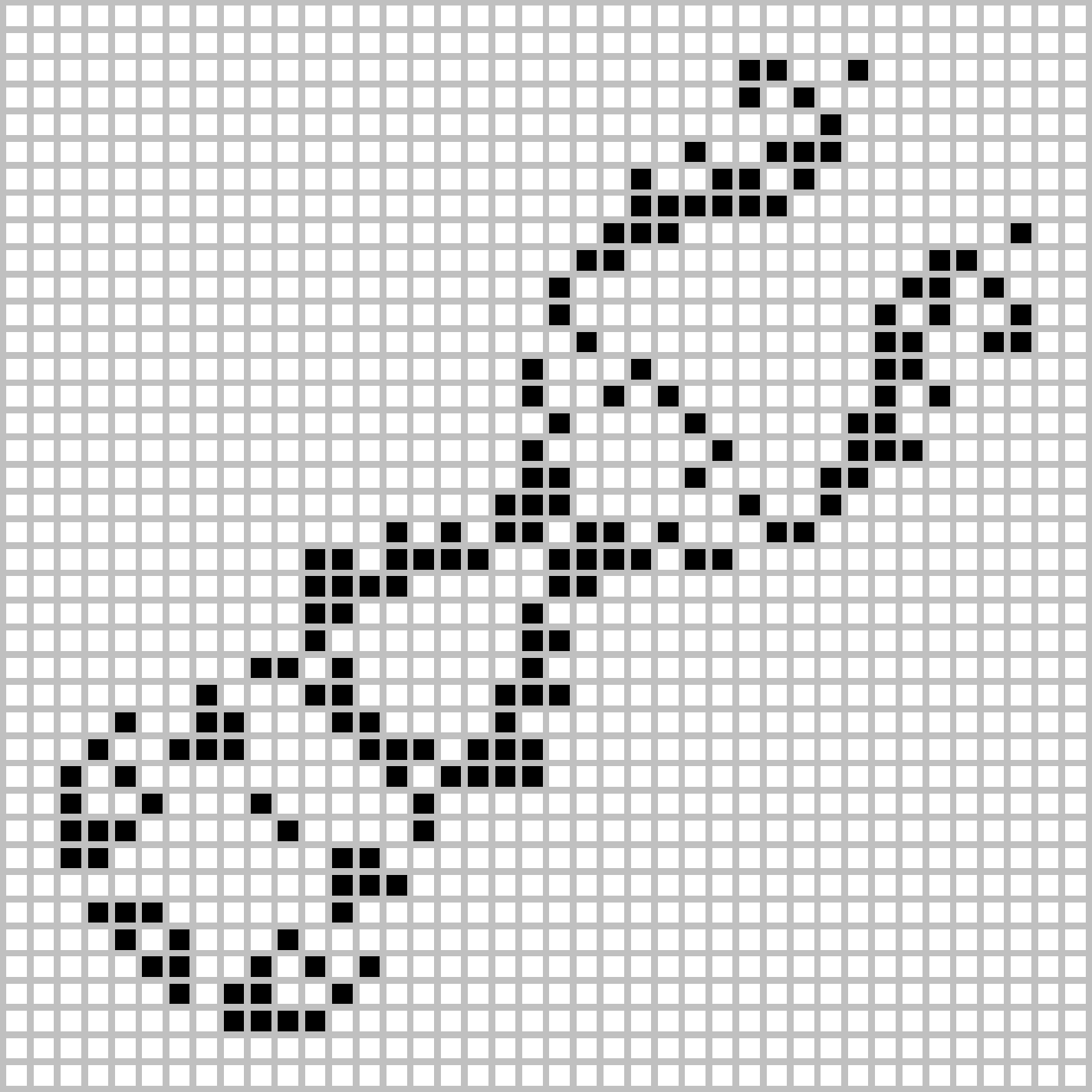}
		  \caption{The $c/3$ diagonal spaceship}\label{fig:2x2_c3dspaceship}
\end{minipage}
\end{figure}

With the discovery of the $c/3$ diagonal spaceship comes the natural question of what the speed limits in 2x2 are for spaceships. It was proved by Conway in the early 1970s that spaceships in Life can not travel faster than $c/2$ orthogonally or $c/4$ diagonally. Using similar methods it is possible to prove that spaceships in 2x2 can not travel faster than $c/2$ orthogonally or $c/3$ diagonally.

\begin{theorem}\label{thm:speedlimit}
	Spaceships in 2x2 can not travel faster than $c/3$ diagonally or $c/2$ orthogonally.
\end{theorem}
\begin{proof}
	The result for diagonal spaceships is proved first. Assume the contrary; assume that there exists a spaceship that travels diagonally at a speed faster than $c/3$. It then must contain a phase such that if that phase is on or to the left of the diagonal line given by the cells A, B, C, D, E, and F in Figure~\ref{fig:2x2_letter_grid} in generation $0$, then the spaceship is \emph{not} entirely on or to the left of the diagonal line defined by the cells V and Y in generation $3$. We can thus assume without loss of generality that the cell X is born in generation $3$. Cells V, W, and Y then must be alive in generation $2$. It is clear that V and Y can not already be alive in generation $1$, so they must have three alive neighbours each in generation $1$. It follows that each of C, D, U, W, and Z must be alive in generation $1$.
		\begin{figure}[ht]
		\center
	    \includegraphics[width=0.25\textwidth]{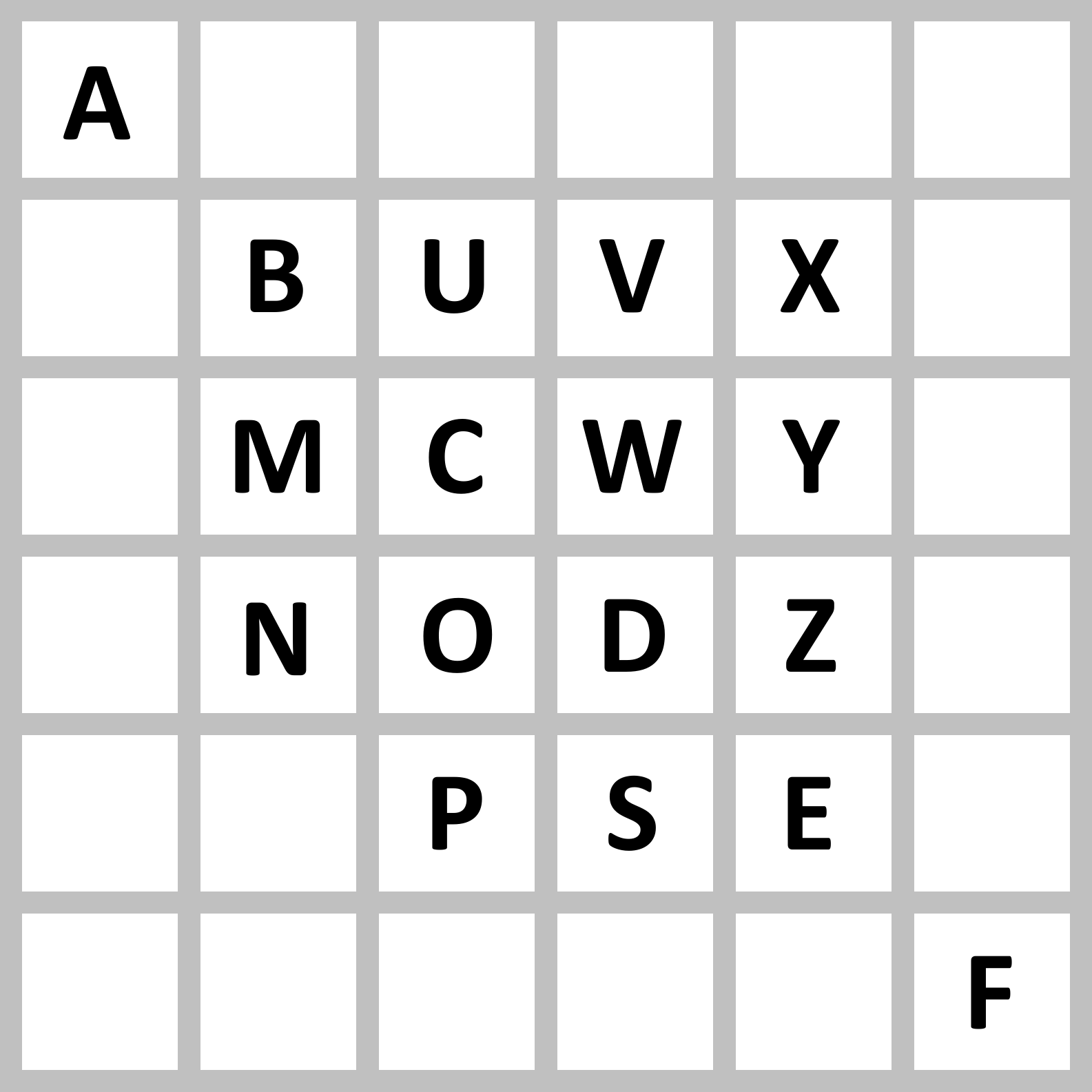}
		\caption{If a spaceship is on or to the left of the line defined by A, B, C, D, E, and F in generation $0$, supposing X is alive in generation $3$ gives a contradiction}
		\label{fig:2x2_letter_grid}
		\end{figure}
		
	Now observe that, because cell W is alive in generation $1$ and generation $2$, and we know that its four neighbours C, D, U and Z are alive in generation $1$, W must have a fifth alive neighbour in generation $1$: O. We will now derive a contradiction by showing that it is impossible for O to be alive in generation $1$.
	
	For U, W, and Z to be alive in generation $1$, it is necessary that B, C, D, E, M, O, and S all be alive in generation $0$. For C and D to survive to generation $1$, N and P must also be alive in generation $0$; otherwise C and D have only four alive neighbours. In particular, this implies that C, D, M, N, P, and S must all be alive in generation $0$, which implies that O has at least $6$ alive neighbours. Since O also must be alive in generation $0$, it can not be alive in generation $1$ -- a contradiction. It follows that the spaceship must be on or to the left of the diagonal line defined by the cells V and Y in generation $3$, which shows that it can not travel faster than $c/3$.\footnote{In 1994, Dean Hickerson used a similar method to prove that spaceships in the rule B3/S135 can not travel faster than $c/3$ diagonally or $2c/3$ orthogonally.}
	
		\begin{figure}[ht]
		\center
	    \includegraphics[width=0.25\textwidth]{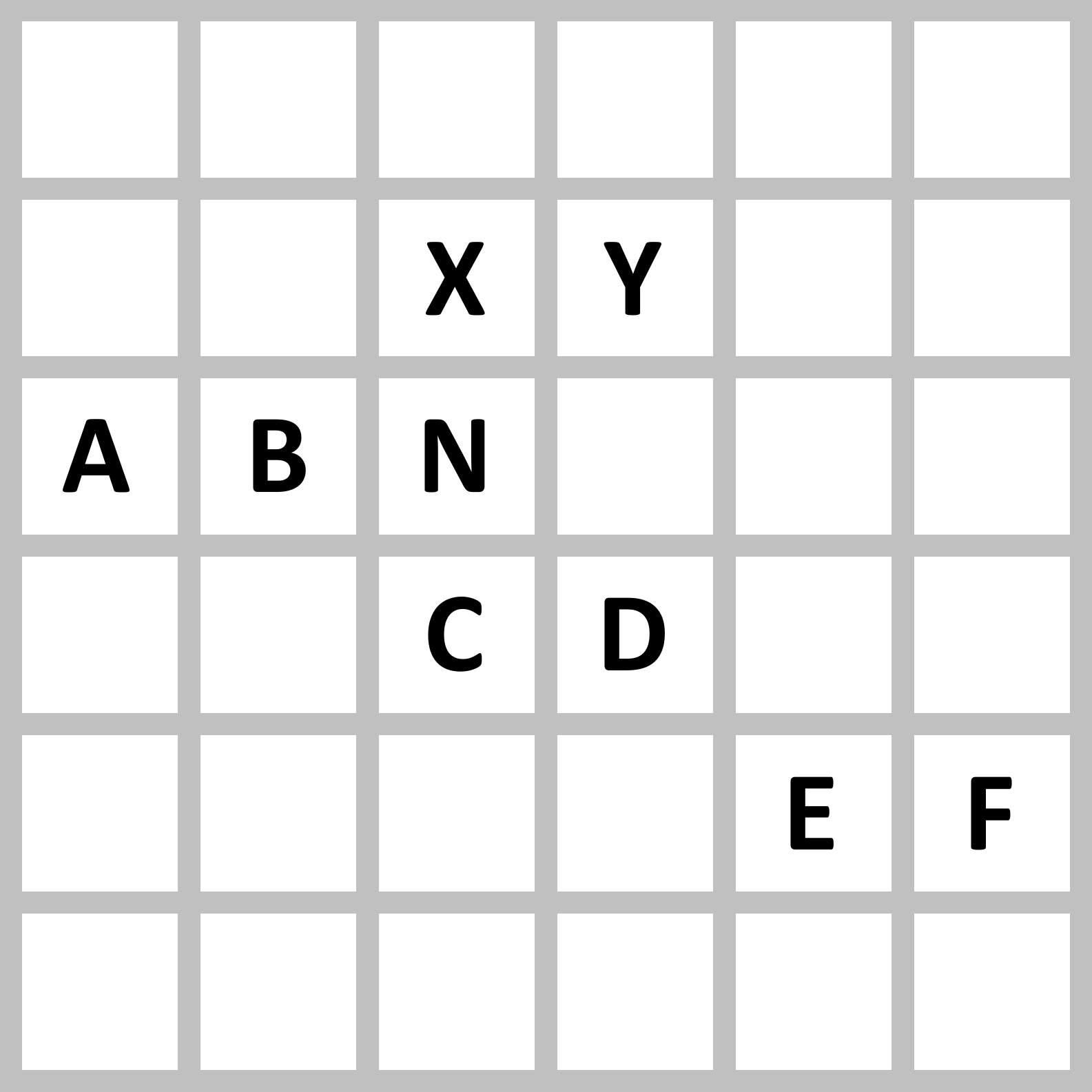}
		\caption{If a spaceship is on or below the line defined by A, B, C, D, E, and F in generation $0$, supposing X or Y is alive in generation $2$ gives a contradiction}
		\label{fig:2x2_letter_grid2}
		\end{figure}

	To see the result for orthogonal spaceships, we perform a similar argument, but instead assume that the spaceship is on or below the line of slope $-1/2$ defined by cells A, B, C, D, E, and F in Figure~\ref{fig:2x2_letter_grid2} in generation $0$. Assuming that orthogonal spaceships can travel faster than $c/2$ orthogonally, it must be possible for either X or Y to be alive in generation $2$. However, the only one of Y's neighbours that can possibly be alive in generation $1$ is N, so Y can not be alive in generation $2$. Similarly, the only of X's neighbours that can be alive in generation $1$ are B and N, so X can not possibly be alive in generation $2$. This shows that neither X nor Y can be alive in generation $2$, so spaceship speeds are limited to $c/2$ orthogonally. \qed
\end{proof}

Because a $c/3$ diagonal spaceship is already known, as are several $c/2$ orthogonal spaceships, the speed limit question for 2x2 has been answered. Additionally, the proof of Theorem~\ref{thm:speedlimit} only relies on a couple of properties of 2x2 and so the upper bounds apply to various other rules as well. In fact, the upper bounds apply to any of the $2^{14} = 16384$ rules in which birth occurs when a cell has $3$ neighbours but does not occur for $2$ or fewer neighbours.

Finally, it is worth noting that the second half of the proof of Theorem~\ref{thm:speedlimit} implies not only the $c/2$ speed limit result for orthogonal spaceships, but also that any spaceship that travels $a$ cells vertically for every $b$ cells horizontally (if such a spaceship exists) can not travel faster than $\max\{a,b\}c/(2a+b)$.\footnote{David Eppstein was aware of this speed upper bound for rules in which birth occurs when a cell has three live neighbours back in 1999 and he used it as an alternative method of proving the $c/3$ upper bound for diagonal spaceships.} In particular, this captures the $c/3$ diagonal speed limit and says that knightships (which travel two cells vertically for every one cell horizontally), if they exist, can not travel faster than $2c/5$.

  \section{As a Block Cellular Automaton}\label{sec:block}

		One of the most interesting properties of 2x2 is that it emulates a simpler cellular automaton that acts on $2 \times 2$ blocks of cells. A bit more specifically, it emulates the block cellular automaton that makes use of the Margolus neighbourhood and evolves according to the six rules given by Figure~\ref{fig:block_evolve}.

		\begin{figure}[ht]
		\center
		\includegraphics[width=2in]{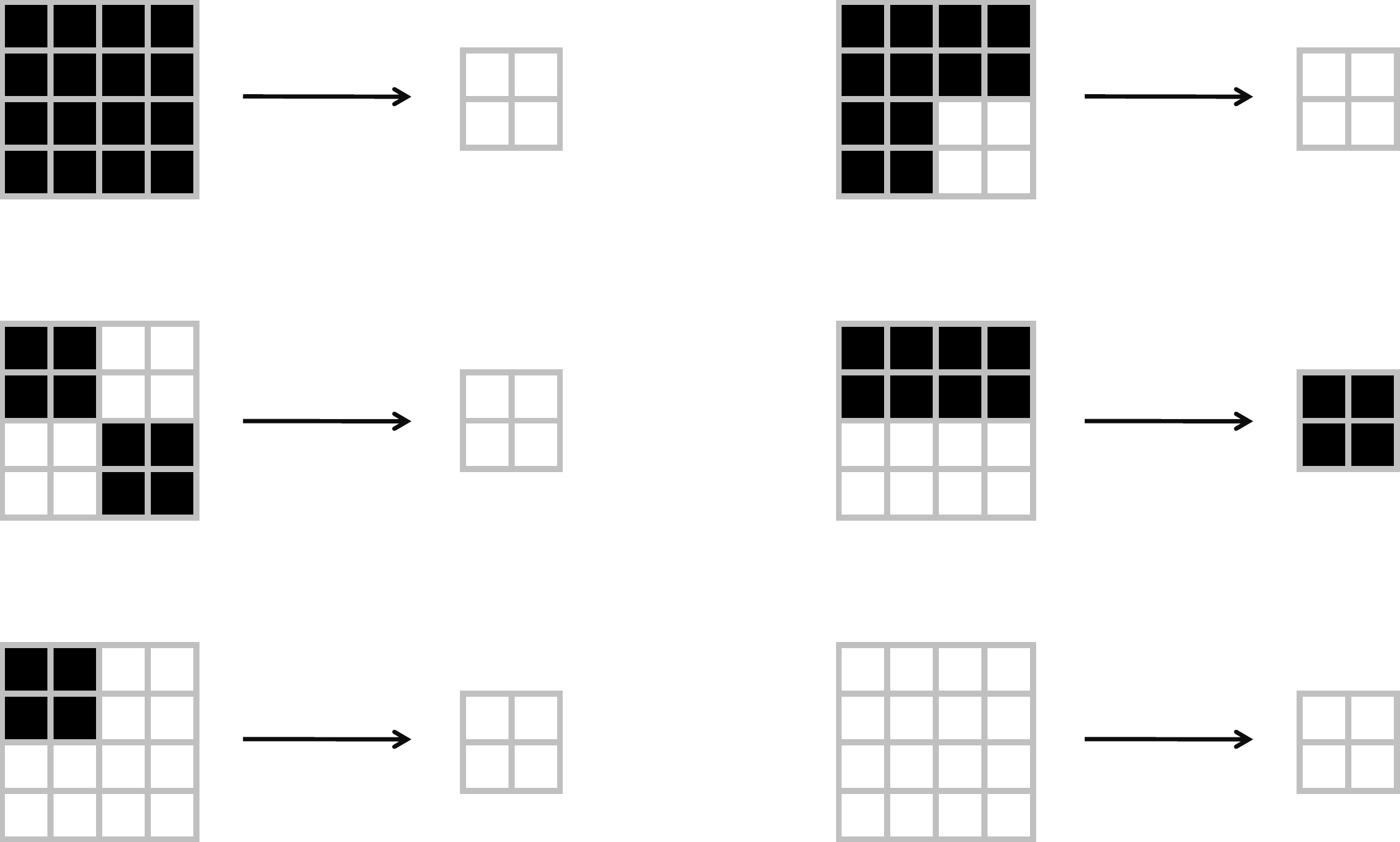}
		\caption{The block cellular automaton emulated by 2x2}
		\label{fig:block_evolve}
		\end{figure}
		
		By saying that 2x2 emulates a Margolus block cellular automaton, we mean that the resulting block appears at the center of the original four blocks. Thus, patterns that are originally made up of $2 \times 2$ blocks will forever be made up of $2 \times 2$ blocks, but the block partition will be offset diagonally by one cell in the odd generations from the even generations. Of course, 2x2 is not the only Life-like cellular automaton that emulates a Margolus block cellular automaton; such a cellular automaton is emulated if and only if the rule satisfies the following four conditions:
		\begin{itemize}
			\item Birth occurs for $3$ neighbours if and only if survival occurs for $5$ neighbours.
			\item Birth occurs for $4$ neighbours if and only if survival occurs for $4$ neighbours.
			\item Birth occurs for $5$ neighbours if and only if survival occurs for $6$ neighbours if and only if survival occurs for $7$ neighbours.
			\item Birth occurs for $1$ neighbour if and only if birth occurs for $2$ neighbours if and only if survival occurs for $3$ neighbours.
		\end{itemize}
	
		More succinctly, a Life-like cellular automaton emulates a Margolus block cellular automaton if and only if, in its rulestring, B3 = S5, B4 = S4, B5 = S6 = S7, and B1 = B2 = S3. 2x2 can be seen to satisfy these conditions because 4 is neither a birth condition nor a survival condition, 5 is not a birth condition and 6 and 7 are not survival conditions, 3 is a birth condition and 5 is a survival condition, and 3 is not a survival condition and 1 and 2 are not birth conditions. There are $2^{12} = 4096$ Life-like cellular automata that emulate $2^6 = 64$ different Margolus block cellular automata. The $64$ Life-like cellular automata from B3/S5 to B3678/S0125 all emulate the same Margolus block cellular automaton given by Figure~\ref{fig:block_evolve}.
		
		In fact, it was noticed by David Eppstein in 1998 that the Margolus block cellular automaton that 2x2 emulates also emulates itself via $2 \times 2$ blocks \emph{of} $2 \times 2$ blocks, but with a slowdown of one simulated generation per real generation. That is, each $2 \times 2$ block in Figure~\ref{fig:block_evolve} can be replaced by the corresponding $4 \times 4$ block as long as it is understood that the transition time indicated by the arrows is two generations instead of one.
				
		It follows naturally that $8 \times 8$ blocks can be used to simulate $4 \times 4$ blocks, again doubling the number of real generations per simulated generation. In general, $2^k \times 2^k$ blocks can be used to simulate $2 \times 2$ blocks, with each simulated generation requiring $2^{k-1}$ real generations. Because the $2 \times 4$ rectangle is an oscillator with period 2, it follows that the $4 \times 8$ rectangle is an oscillator with period 4, the $8 \times 16$ rectangle is an oscillator with period 8, and in general the $2^k \times 2^{k+1}$ rectangle is an oscillator with period $2^k$. This was the first known proof that 2x2 contains oscillators of arbitrarily large period -- we will see another proof (also related to $2 \times 2$ block oscillators) in the next section.
		
		\begin{figure}[ht]
		\center
		\includegraphics[width=0.9\textwidth]{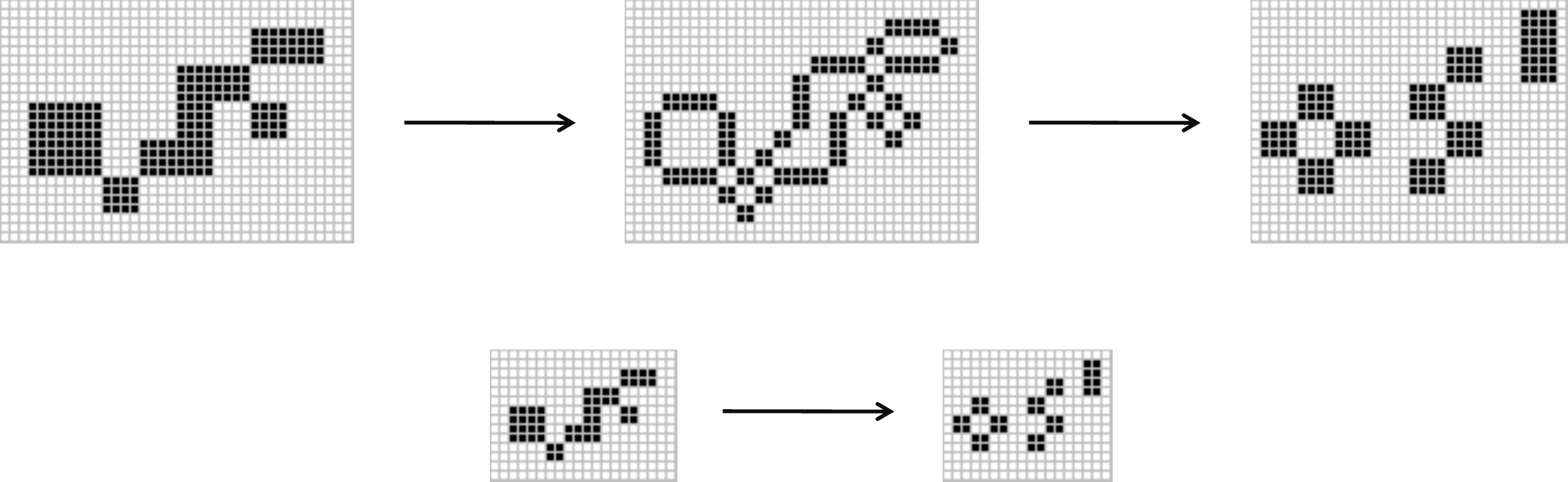}
		\caption{An example of $4 \times 4$ blocks (top) taking two generations to emulate one generation of a $2 \times 2$ block pattern (bottom)}
		\label{fig:block_evolve2}
		\end{figure}

		One might wonder what types of patterns exist in this block cellular automaton -- after all, it is a simpler rule so perhaps we can prove the existence of certain types of patterns in 2x2 by simply trying to find interesting block patterns. However, no block pattern can ever escape its initial ``bounding diamond,'' so we can not hope to find spaceships or infinitely-growing patterns in this manner. Additionally, because the block partition in even generations is offset by one cell from the block partition in odd generations, we can't hope to find odd-period oscillators. Thus, we investigate an infinite family of even-period block oscillators.

  \section{Block Oscillators}\label{sec:blockosc}
		
		One particularly interesting family of oscillators in 2x2 are those that in one phase are a $2 \times 4n$ rectangle of alive cells for some integer $n \geq 1$. For $n = 1$, the oscillator has period 2 and simply rotates by 90 degrees every generation. That is, it starts as a $2 \times 4$ rectangle, evolves into a $4 \times 2$ rectangle after one generation, and then evolves back into a $2 \times 4$ rectangle after the next generation. As $n$ increases, these oscillators behave more and more unpredictably. For $n = 2$, the oscillator has period 6 as shown in Figure~\ref{fig:period6block}.

\begin{figure}[ht]
\center
\includegraphics[width=2in]{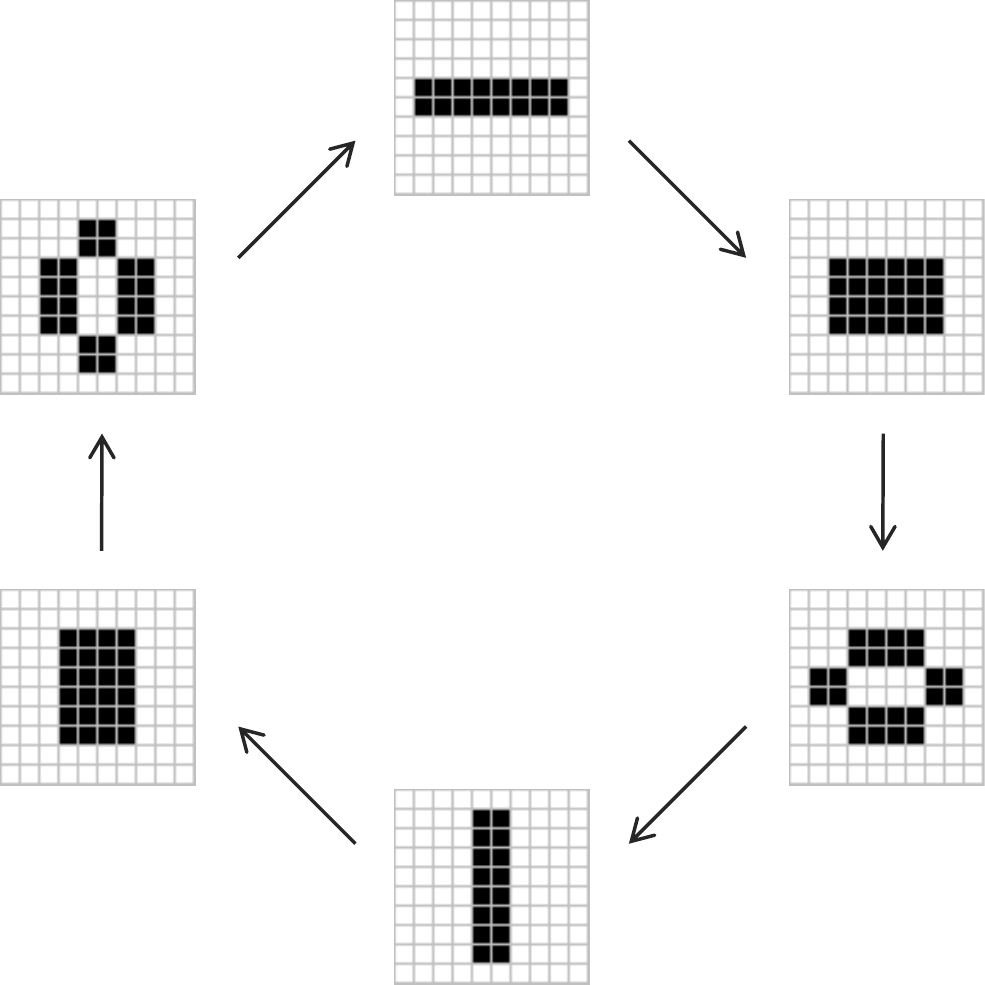}
\caption{The $2 \times 8$ period 6 block oscillator}
\label{fig:period6block}
\end{figure}

		By simply playing around with Life simulation software, it is not difficult to compute the period of the $2 \times 4n$ block oscillator for $n = 1, 2, 3, \ldots$ to be $2, 6, 14, 14, 62, 126, 30, 30, 1022, 126, \ldots$\footnote{Sloane's A160657}. Other than the fact that each of these periods is two less than a power of two, this sequence does not have any obvious pattern. The following theorem shows that the sequence is in fact related to a well-studied mathematical phenomenon, albeit one that does not admit what most people would consider a closed form solution. Before presenting the theorem, it is perhaps worth noting that B36/S125 is not the only rule in which these oscillators work; these rectangular oscillators behave the same in the $2^9 = 512$ rules from B3/S5 to B35678/S012567.
		
		\begin{theorem}\label{thm:block_period}
			The period of the $2 \times 4n$ block oscillator is $2(2^k - 1)$, where $k$ is the multiplicative suborder of $2 \text{ {\rm (mod $(2n+1)$)}}$.\footnote{The values of $k$ given by Theorem~\ref{thm:block_period} for $n = 1, 2, 3, \ldots$ are $1, 2, 3, 3, 5, 6, 4, 4, 9, 6, \ldots$ (Sloane's A003558).}
		\end{theorem}
		
		The multiplicative suborder of $a$ (mod $b$), denoted ${\rm sord}_b(a)$, is defined as the least natural number $k$ such that $a^k \equiv \pm 1 \ (\text{{\rm mod} }$b$)$. There exists a $k$ satisfying this condition if and only if $GCD(a,b) = 1$. Since $2n + 1$ is odd we know that the $k$ mentioned by Theorem~\ref{thm:block_period} is well-defined. We now sketch a proof of the theorem.
		
		\begin{proof}
			The key step in the proof is to notice that each phase of these oscillators can be described as an XOR of rectangles -- that is, an intersection of rectangles where we keep the sections that consist of an odd number of overlapping rectangles and we discard the sections that consist of an even number of overlapping rectangles. For example, Figure~\ref{fig:block_xor} shows a phase of the $n = 3$ block oscillator represented as the XOR of a $2 \times 12$ rectangle, a $6 \times 8$ rectangle, and a $10 \times 4$ rectangle.

\begin{figure}[ht]
\center
\includegraphics[width=2in]{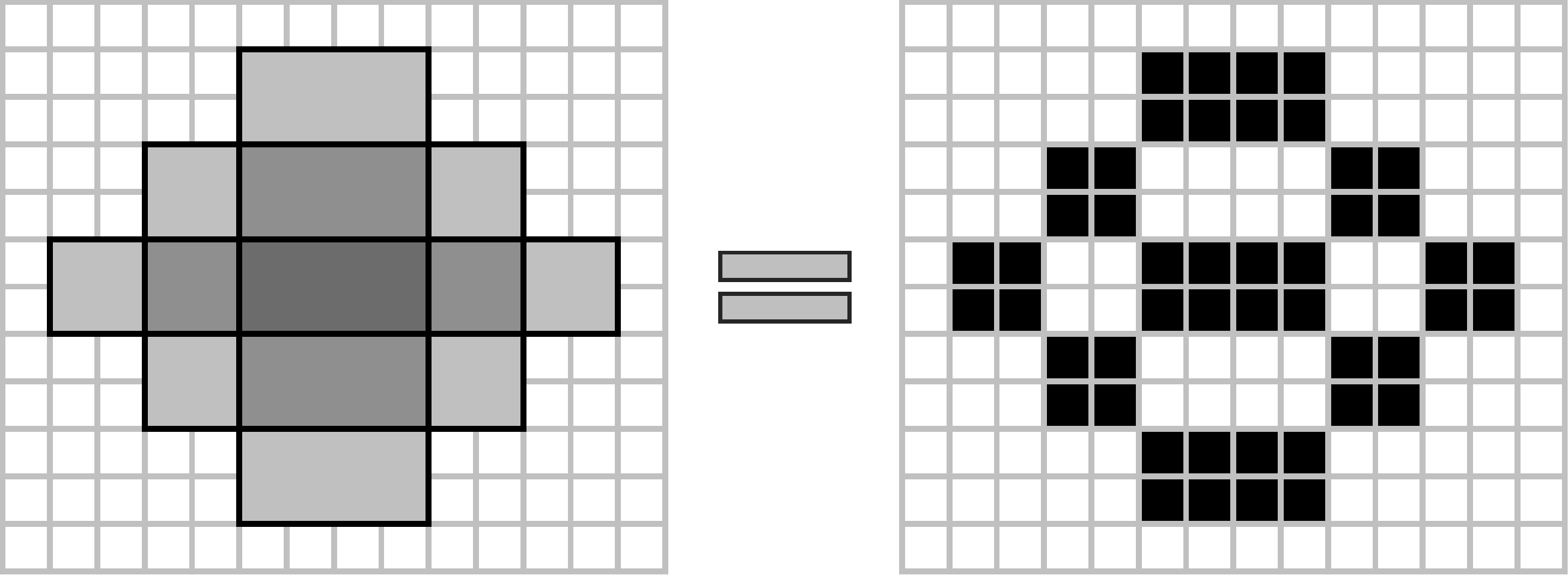}
\caption{Generation $6$ of the $n = 3$ block oscillator, depicted as the XOR of a $2 \times 12$ rectangle, a $6 \times 8$ rectangle, and a $10 \times 4$ rectangle}
\label{fig:block_xor}
\end{figure}

Different phases may require a different number of rectangles to be XORed, but every phase can always be represented in this way. More important, however, is the fact that the evolution of the oscillators occurs in a predictable way when modeled like this. Consider the grid given in Figure~\ref{fig:block_sierpinski}, which is helpful in analyzing the evolution of the block oscillators. The fact that it resembles the Sierpinski triangle is no coincidence; this block automaton is, in a sense, emulating the Rule 90 elementary cellular automaton.

The way to read Figure~\ref{fig:block_sierpinski} is that that row represents the generation number and the column represents the size of the rectangle that is being XORed. The first column represents a $2 \times 4n$ rectangle, the second column represents a $4 \times (4n - 2)$ rectangle, the third column represents a $6 \times (4n - 4)$ rectangle, and so on. Thus, you ``start'' in the top-left cell, and that represents the $2 \times 4n$ rectangle of generation $0$. To go to generation $1$, go to the next row, where we see that the only filled in cell is in the second column, which is the $4 \times (4n - 2)$ column. Thus, generation $1$ will just be a filled-in $4 \times (4n - 2)$ rectangle. To see what generation 2 will look like, go to the next row, where we see that two cells are filled in, corresponding to $2 \times 4n$ and $6 \times (4n - 4)$ rectangles. Thus, XOR together two rectangles of those sizes (in the sense described earlier) to get what generation 2 looks like.

\begin{figure}[ht]
\center
\includegraphics[scale=0.5]{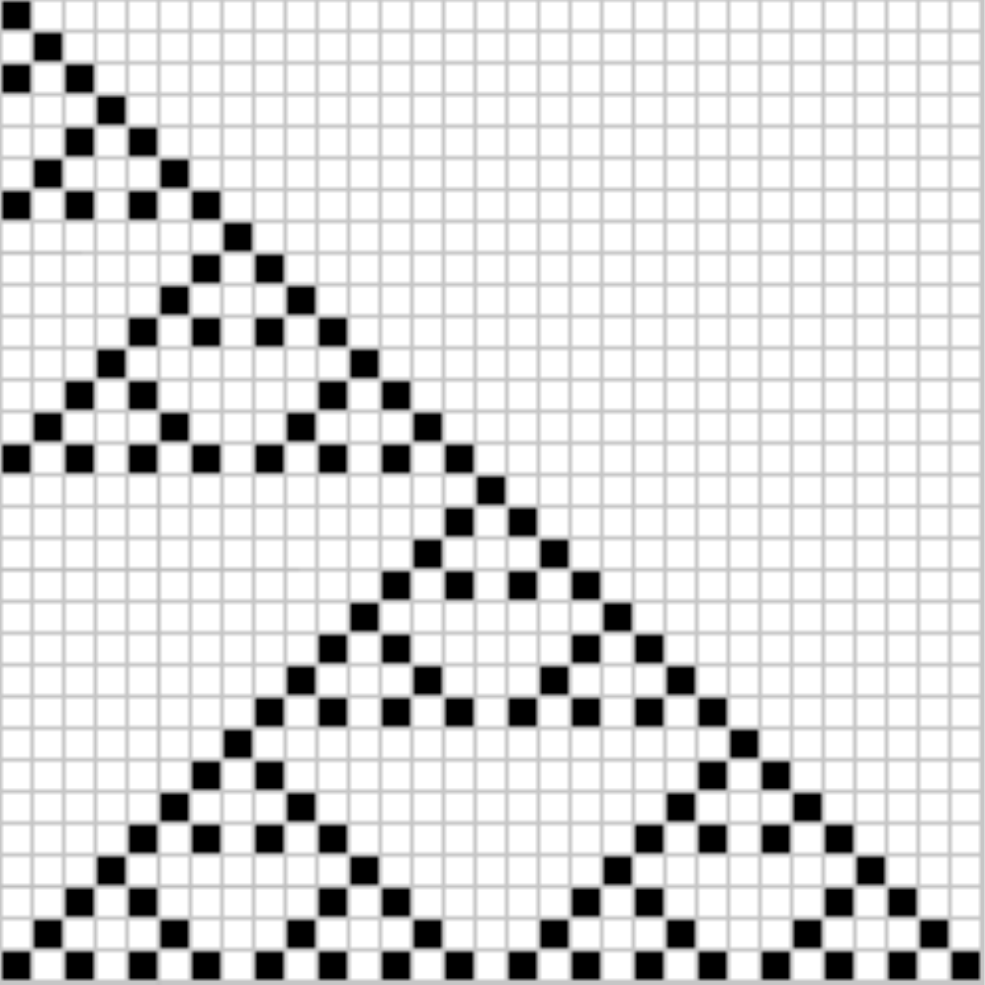}
\caption{A grid that can be used to determine future phases of the block oscillators}
\label{fig:block_sierpinski}
\end{figure}

In order to determine the oscillators' periods, we must observe that if we continue to label the columns in the way described, eventually we hit zero-length rectangles, which does not make a whole lot of sense. Thus, we simply ignore any rectangles that are of length zero. But what about rectangles of negative length? Instead of counting down into negative numbers, start counting back up. Thus, the columns corresponding to a $2k \times (4n - 2(k-1))$ rectangle are given by $2\ell(2n+1) + k$ for $\ell \in \{0, 1, 2, \ldots \}$ and $2\ell(2n+1) - k$ for $\ell \in \{1, 2, 3, \ldots \}$. To help illustrate this idea, consider the grid given in Figure~\ref{fig:block_sierpinski2}, which shows the lengths associated with each column in the $n = 3$ case.

\begin{figure}[ht]
\center
\includegraphics[scale=0.5]{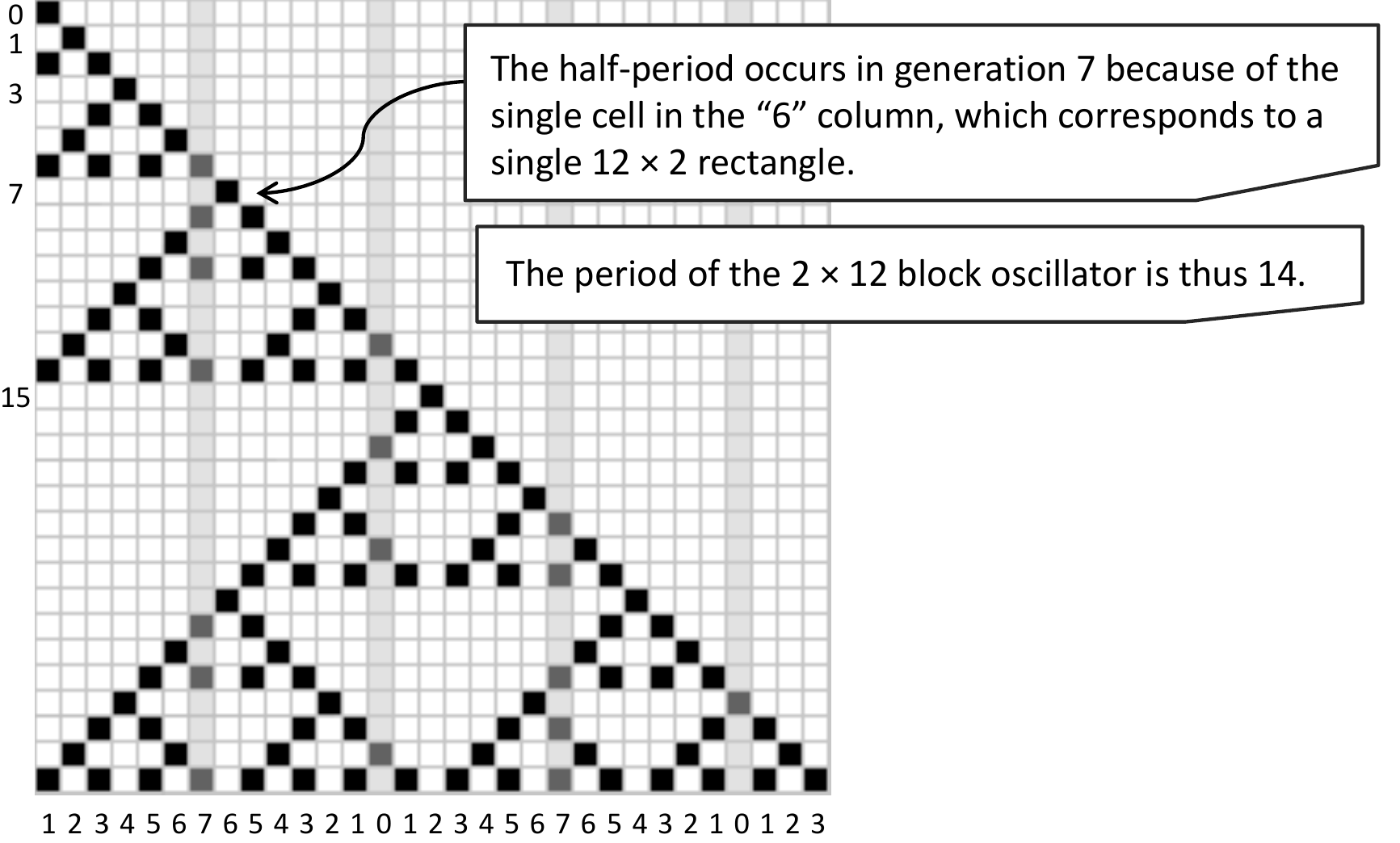}
\caption{For $1 \leq k \leq 6$, the columns marked by the number $k$ correspond to a $2k \times (4n - 2(k-1))$ rectangle -- columns marked $0$ or $7$ are ignored}
\label{fig:block_sierpinski2}
\end{figure}

For example, in generation 6, there is a live cell in columns marked ``$1$'', ``$3$'', ``$5$'', and ``$7$''. The column marked ``$7$'' is ignored, and the columns marked ``$1$'', ``$3$'', and ``$5$'' correspond to rectangles of size $2 \times 12$, $6 \times 8$, and $10 \times 4$, respectively. Generation 6 of the $n = 3$ oscillator should thus be the XOR of these three rectangles, which we saw in Figure~\ref{fig:block_xor} is indeed the case.

Now, notice that, because the multiplicative suborder of $2 \text{ {\rm (mod $(2n+1)$)}}$ is well-defined, there must exist some $\ell,m$ such that $2^\ell = m(2n+1) \pm 1$. Because $2^\ell$ is even, it must be the case that $m$ is odd. In terms of Figure~\ref{fig:block_sierpinski2} this means that there is some row containing only one cell (i.e., one of the rows labelled $3, 7, 15, 31, \ldots, 2^\ell - 1, \ldots$) such that its cell is in a ``$2n$'' column (i.e., a ``$6$'' column in the case of the example provided). This means that, at some point during its evolution, a $2 \times 4n$ rectangle evolves into a $4n \times 2$ rectangle. A simple symmetry argument shows that this must occur at exactly half of its period.

It follows that the oscillator returns to its original $2 \times 4n$ form in generation $2(2^k - 1)$, where $k$ is the multiplicative suborder of $2 \text{ {\rm (mod $(2n+1)$)}}$. Thus, the period of the $2 \times 4n$ block oscillator must be a factor of $2(2^k - 1)$. It is not difficult to see that the period must actually be of the form $2(2^\ell - 1)$ for some integer $\ell$, so the result follows. \qed
		\end{proof}
		
		Even though there is no known closed form formula for the multiplicative suborder of $a$ (mod $b$), it has several simple properties that can be verified without much difficulty. Importantly, it can be bounded as follows:
		\begin{align}\label{eq:sord_ineq}
			\log_a(b-1) \leq {\rm sord}_b(a) \leq \frac{b-1}{2}.
		\end{align}
		
		The lower bound follows simply by noting that if $k = {\rm sord}_b(a)$, then $a^k \geq b - 1$. To see that the upper bound holds, recall that Euler's theorem says that $a^{\phi(b)} \equiv 1 \ (\text{{\rm mod} }$b$)$, where $\phi(b)$ is the totient of $b$. Since $\phi(b) \leq b-1$, the inequality follows. Also, since $\phi(b) = b-1$ exactly when $b$ is prime, we see that equality is attained on the right in Inequality~\eqref{eq:sord_ineq} only if $b$ is prime.
		
		The following corollary of Theorem~\ref{thm:block_period} and Inequality~\eqref{eq:sord_ineq} follows immediately.
		\begin{corollary}\label{cor:block_period}
			Let $p$ be the period of the $2 \times 4n$ block oscillator. Then
			\[
				2(2n - 1) \leq p \leq 2^{n+1} - 2.
			\]
		\end{corollary}
		\begin{proof}
			By simply letting $a = 2$, $b = 2n+1$, and $k = {\rm sord}_b(a)$, Inequality~\eqref{eq:sord_ineq} says
			\[
				\log_2(2n) \leq k \leq n.
			\]
			
			\noindent Rearranging the inequality gives
			\[
				2(2n - 1) \leq 2^{k+1} - 2 \leq 2^{n+1} - 2.
			\]
			
			\noindent The result follows from Theorem~\ref{thm:block_period}.\qed
		\end{proof}
		
		This result provides, for example, a second proof that 2x2 contains oscillators with arbitrarily large period. In fact, any period of the form $2(2^k - 1)$ for an integer $k \geq 1$ is attainable as a $2 \times 4n$ block oscillator by simply choosing $n = 2^{k-1}$. Combining this with the period-doubling method of Section~\ref{sec:block}, we can construct a block oscillator with period $2^{\ell}(2^k - 1)$ for any integers $k,\ell \geq 1$ -- one such oscillator is a solid rectangle of alive cells of size $2^\ell \times 2^{k+\ell}$. It is still an open question whether or not 2x2 is omniperiodic; that is, whether or not it contains an oscillator of any given period.
		
		Table~\ref{tab:block_period} shows the period of the $2 \times 4n$ block oscillator for several values of $n$, as well as the bounds given by Corollary~\ref{cor:block_period}. Values for which either of the bounds are attained have been highlighted.
		\begin{table}
		\center
		\caption{The period of the $2 \times 4n$ block oscillator, as well as the bounds of Corollary~\ref{cor:block_period}, for small values of $n$}
		\label{tab:block_period}
		\begin{tabular}{cc}
		\begin{tabular}{lccc}
		\hline\noalign{\smallskip}
		$\mathbf{n}$ \ & \ $\mathbf{2(2n - 1)}$ \ & \ {\bf Period} \ & \ $\mathbf{2^{n+1} - 2}$ \\
		\noalign{\smallskip}\hline\noalign{\smallskip}
		$\mathbf{1}$ & \colorbox{lightbg}{$2$} & \colorbox{lightbg}{$2$} & \colorbox{lightbg}{$2$} \\
		$\mathbf{2}$ & \colorbox{lightbg}{$6$} & \colorbox{lightbg}{$6$} & \colorbox{lightbg}{$6$} \\
		$\mathbf{3}$ & $10$ & \colorbox{lightbg}{$14$} & \colorbox{lightbg}{$14$} \\
		$\mathbf{4}$ & \colorbox{lightbg}{$14$} & \colorbox{lightbg}{$14$} & $30$ \\
		$\mathbf{5}$ & $18$ & \colorbox{lightbg}{$62$} & \colorbox{lightbg}{$62$} \\
		$\mathbf{6}$ & $22$ & \colorbox{lightbg}{$126$} & \colorbox{lightbg}{$126$} \\
		$\mathbf{7}$ & $26$ & $30$ & $254$ \\
		$\mathbf{8}$ & \colorbox{lightbg}{$30$} & \colorbox{lightbg}{$30$} & $510$ \\
		$\mathbf{9}$ & $34$ & \colorbox{lightbg}{$1022$} & \colorbox{lightbg}{$1022$} \\
		$\mathbf{10}$ & $38$ & $126$ & $2046$ \\
		$\mathbf{11}$ & $42$ & \colorbox{lightbg}{$4094$} & \colorbox{lightbg}{$4094$} \\
		$\mathbf{12}$ & $46$ & $2046$ & $8190$ \\
		\noalign{\smallskip}\hline
		\end{tabular} \quad \ & \quad \ \begin{tabular}{lccc}
		\hline\noalign{\smallskip}
		$\mathbf{n}$ \ & \ $\mathbf{2(2n - 1)}$ \ & \ {\bf Period} \ & \ $\mathbf{2^{n+1} - 2}$ \\
		\noalign{\smallskip}\hline\noalign{\smallskip}
		$\mathbf{13}$ & $50$ & $1022$ & $16382$ \\
		$\mathbf{14}$ & $54$ & \colorbox{lightbg}{$32766$} & \colorbox{lightbg}{$32766$} \\
		$\mathbf{15}$ & $58$ & $62$ & $65534$ \\		
		$\mathbf{16}$ & \colorbox{lightbg}{$62$} & \colorbox{lightbg}{$62$} & $131070$ \\
		$\mathbf{17}$ & $66$ & $8190$ & $262142$ \\
		$\mathbf{18}$ & $70$ & \colorbox{lightbg}{$524286$} & \colorbox{lightbg}{$524286$} \\
		$\mathbf{19}$ & $74$ & $8190$ & $1048574$ \\
		$\mathbf{20}$ & $78$ & $2046$ & $2097150$ \\
		$\mathbf{21}$ & $82$ & $254$ & $4194302$ \\
		$\mathbf{22}$ & $86$ & $8190$ & $8388606$ \\
		$\mathbf{23}$ & $90$ & \colorbox{lightbg}{$16777214$} & \colorbox{lightbg}{$16777214$} \\
		$\mathbf{24}$ & $94$ & $4194302$ & $33554430$ \\		
		\noalign{\smallskip}\hline
		\end{tabular}
		\end{tabular}
		\end{table}

It is clear via Figure~\ref{fig:block_sierpinski2} that a $2 \times 4n$ oscillator is a solid rectangle in any generation of the form $2^\ell - 1$, where $\ell$ is an integer. In fact, the number of distinct solid rectangles that the oscillator will produce is exactly $2k$ (or simply $k$ if you don't double-count the 90 degree rotations of rectangles that appear during the second half of the oscillator's period), where $k$ is the multiplicative suborder of $2 \text{ {\rm (mod $(2n+1)$)}}$ (as in Theorem~\ref{thm:block_period}). Thus, if the upper bound of Theorem~\ref{cor:block_period} is not attained, then there are rectangular blocks of size $2\ell \times (4n - 2(\ell-1))$ for some integer $\ell$ that are not part of the evolution of the $2 \times 4n$ block oscillator. Table~\ref{tab:block_period} shows us that the smallest example of this happening is in the $n = 4$ case because the $2 \times 16$ oscillator has period 14, which is less than the upper bound of 30. Indeed, the $\ell = 3$ case of a $6 \times 12$ rectangle is another oscillator of period 14. This tells us what happens for $28$ of the $2(2^4 - 1) = 30$ possible nonempty combinations of four rectangles being XORed together (with an overall 90 degree rotation being allowed). So what happens to the two missing combinations? Well, if we XOR a $2 \times 16$ rectangle with a $10 \times 8$ rectangle and a $14 \times 4$ rectangle, we get a seemingly atyptical period $2$ block oscillator that simply rotates 90 degree every generation (see Figure~\ref{fig:2x2_weird_block}).\footnote{In 1993, Dean Hickerson observed that a similar phenomenon occurs in the $n = 7$ case of a $6 \times 24$ rectangle XORed with an $18 \times 12$ rectangle, although it was unknown whether or not there was a smaller block oscillator that did not turn into a single rectangle in one of its phases.}

\begin{figure}[ht]
\center
\includegraphics[width=2in]{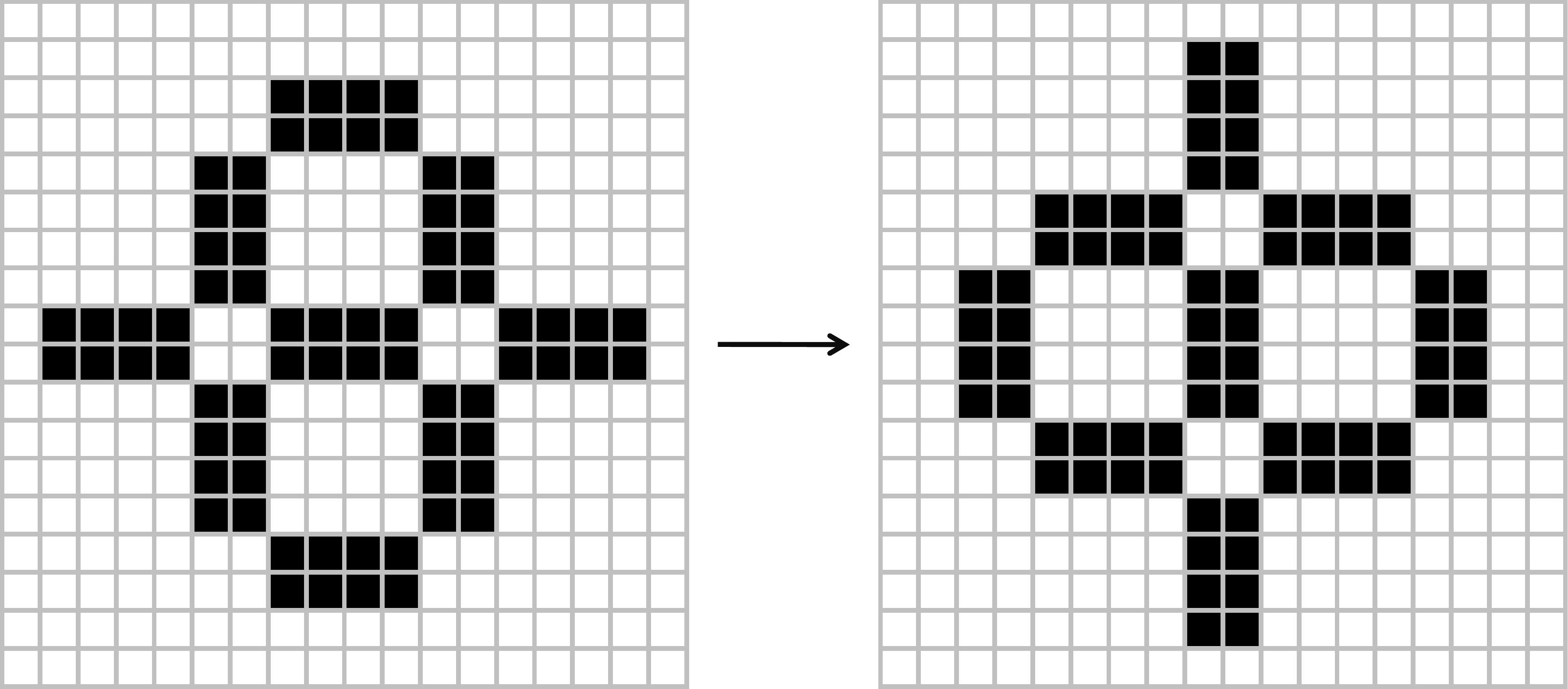}
\caption{The XOR of a $2 \times 16$ rectangle, a $10 \times 8$ rectangle and a $14 \times 4$ rectangle has period 2 and rotates by 90 degrees every generation}
\label{fig:2x2_weird_block}
\end{figure}

While solid rectangles of size $2 \times 4n$ are always an oscillator, this is not the case for rectangles of size $2 \times (4n - 2)$ -- they sometimes evolve into an oscillator and they sometimes vanish completely. It was shown by Dean Hickerson in 1993 they eventually vanish if and only if $n = 2^\ell$ for some integer $\ell \geq 0$. Furthermore, he showed that they must vanish by no later than generation $2^{\ell+1}$. To find out what happens when $n$ is not a power of $2$, we will make use of all of the ideas presented so far in this and the previous section.

If $n$ is odd, say $n = 2m+1$ for some integer $m \geq 1$, then generation $1$ of the oscillator will be a single rectangle of the size $4 \times 8m$. Since this rectangle is made up of $4 \times 4$ blocks, we know from Section~\ref{sec:block} that it is an oscillator with period that is double the period of the $2 \times 4m$ block oscillator. It follows from Theorem~\ref{thm:block_period} that it has period $4(2^k-1)$, where $k$ is the multiplicative suborder of $2 \text{ {\rm (mod $(2m+1)$)}}$.

Now let's suppose $n$ is divisible by $2$ but not by $4$ -- that is, $n = 2(2m+1)$ for some integer $m$. Then generation $3$ of the oscillator will be a single rectangle of size $8 \times 16m$. Since this rectangle is made up of $8 \times 8$ blocks, we can use the same logic as earlier to see that this pattern is an oscillator with period that is quadruple the period of the $2 \times 4m$ block oscillator. That is, it has period $8(2^k-1)$, where $k$ is the multiplicative suborder of $2 \text{ {\rm (mod $(2m+1)$)}}$.

Carrying on in this way, one can easily prove the following result.
	\begin{theorem}\label{thm:block_period_b}
		Let $n = 2^\ell(2m+1)$ for some $\ell,m \geq 0$. If $m = 0$ then the $2 \times (4n-2)$ rectangle vanishes in the $(2^{\ell + 1} - 1)^{th}$ generation. Otherwise, the $2 \times (4n-2)$ rectangle evolves in the $(2^{\ell+1} - 1)^{th}$ generation into a $2^{\ell + 2} \times m2^{\ell + 3}$ block oscillator with period $2^{\ell + 2}(2^k-1)$, where $k$ is the multiplicative suborder of $2 \text{ {\rm (mod $(2m+1)$)}}$.
	\end{theorem}

\vspace{0.1in}

\noindent{\bf Acknowledgements.} Thanks are extended to Alan Hensel and David Eppstein for helpful e-mails as well as Lewis Patterson and the rest of the ConwayLife.com community for helping dig up information about this rule. Thanks also to Dean Hickerson, David Bell, and the other Life enthusiasts who have investigated B36/S125 over the years. The author was supported by an NSERC Canada Graduate Scholarship and the University of Guelph Brock Scholarship.

  \section*{Appendix I: Pattern Collection}\label{sec:patterns_app}
  
\begin{figure}[ht]
\vspace{-0.2in}
\center
\includegraphics[width=0.95\textwidth]{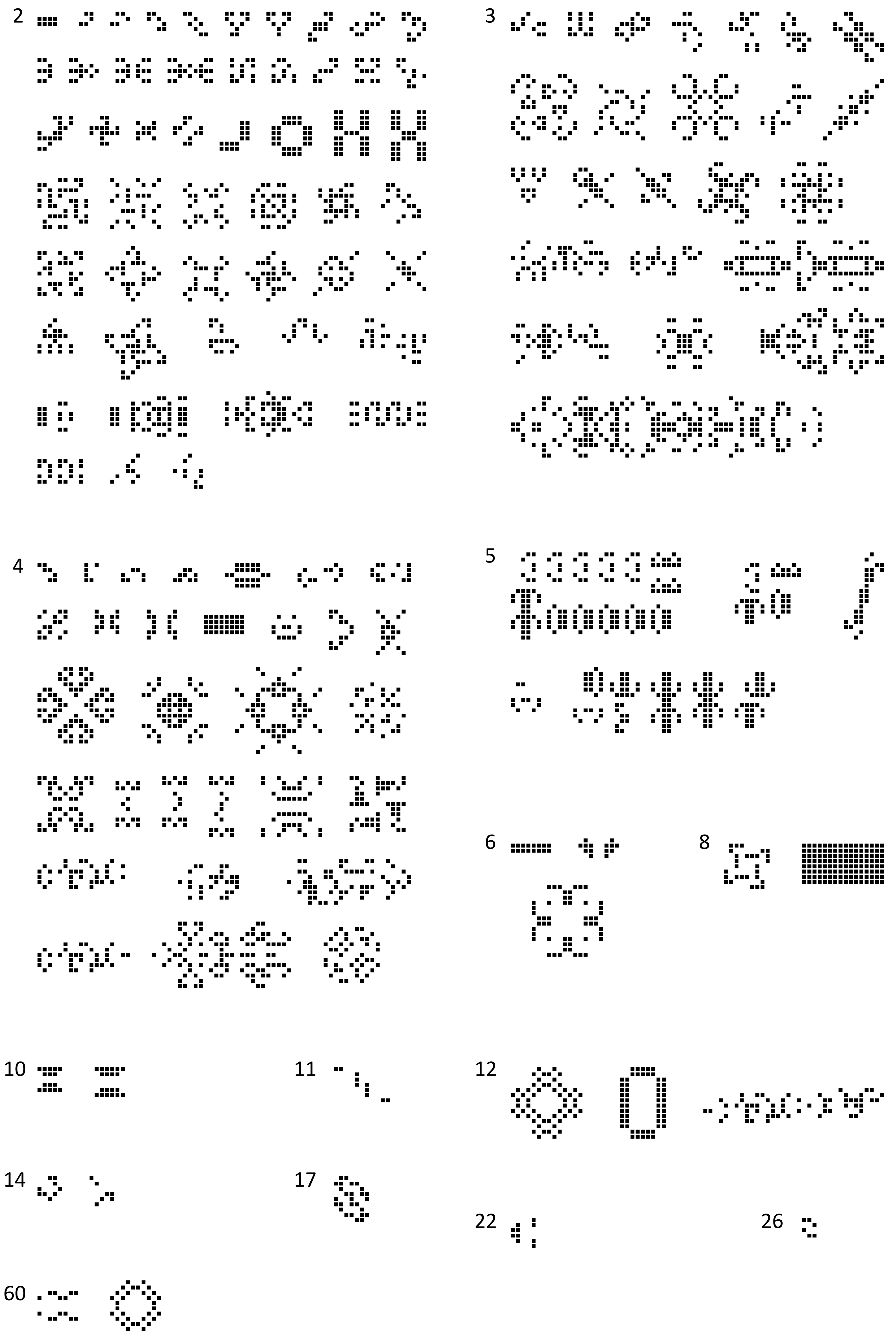}
\caption{Many oscillators of periods $2$ through $60$, most of which were found by Alan Hensel, Dean Hickerson, and Lewis Patterson, with contributions by David Bell}
\vspace{-1in}
\end{figure}

\begin{figure}[ht]
\center
\includegraphics[width=0.95\textwidth]{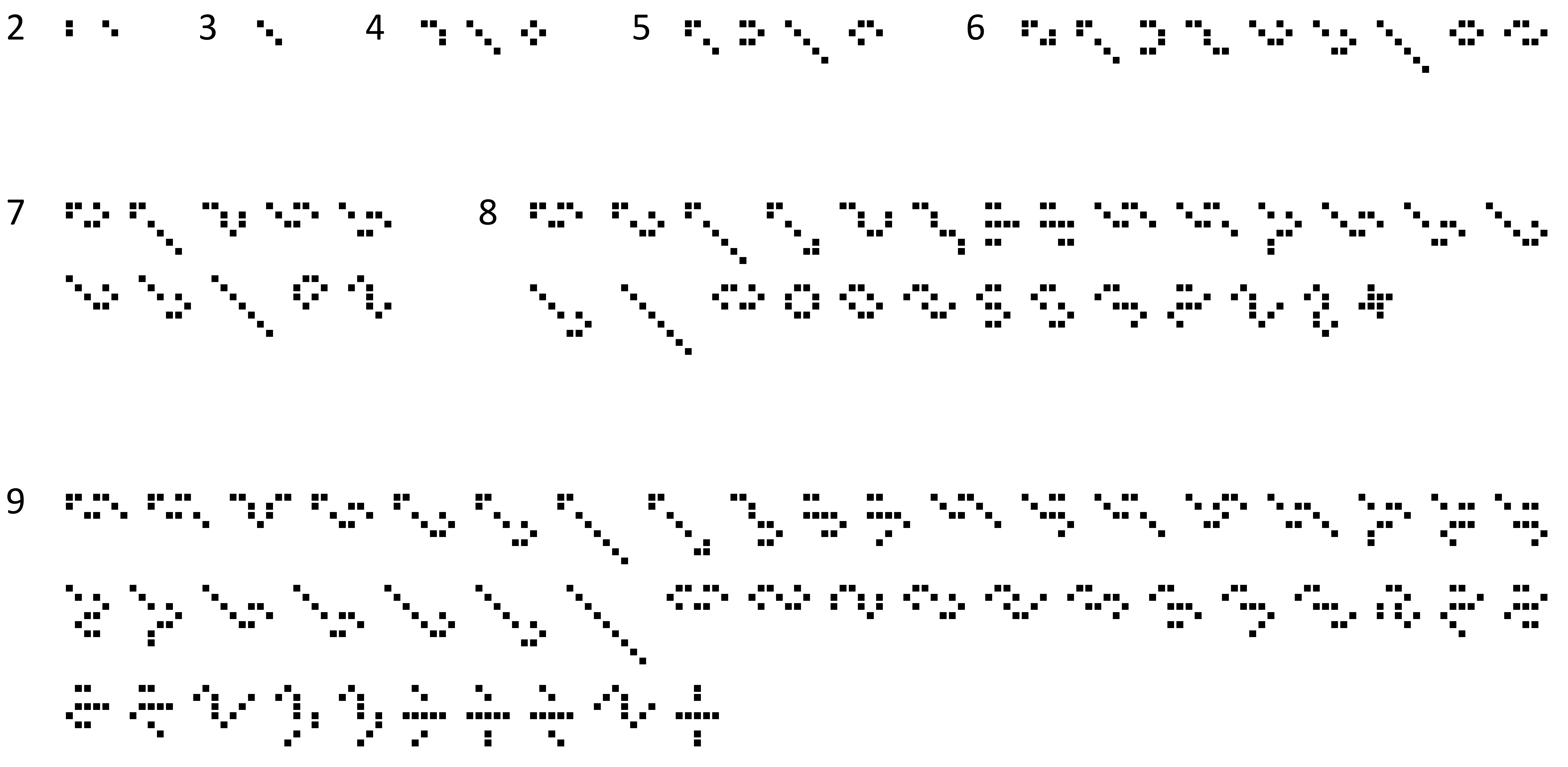}
\caption{The $104$ strict still lifes with $9$ or fewer cells, organized by their cell count}
\vspace{0.4in}
\end{figure}

\begin{figure}[ht]
\center
\includegraphics[width=0.95\textwidth]{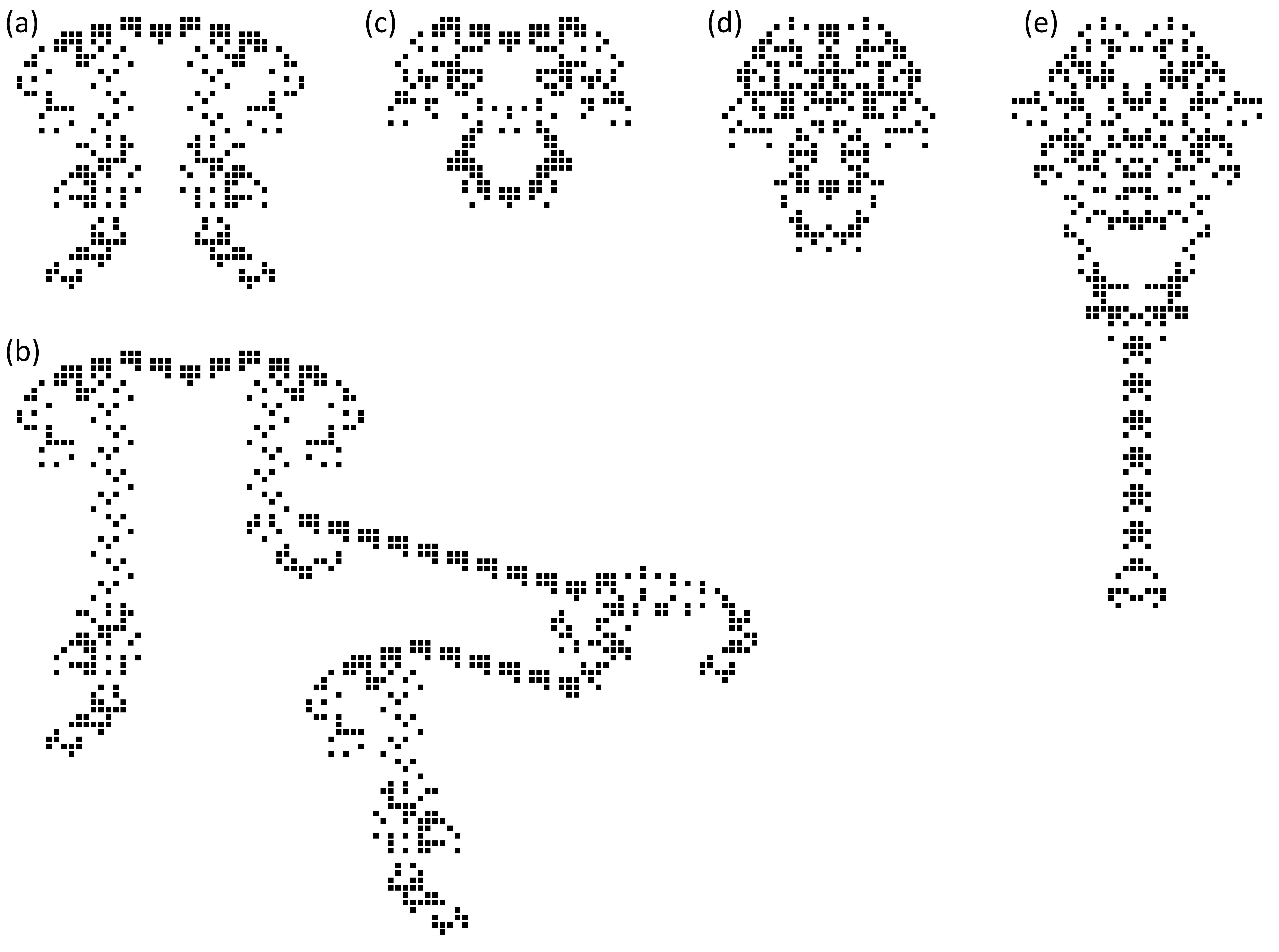}
\caption{Several $c/2$ orthogonal spaceships travelling upward: (a) and (b) are extensible ``jellyfish'' found by David Bell in December 1999, (c), (d) and (e) were found by David Eppstein, with (d) being the first discovered $c/2$ in October 1998}
\label{fig:2x2_c2_spaceships}
\end{figure}

\begin{figure}[ht]
\center
\includegraphics[width=0.95\textwidth]{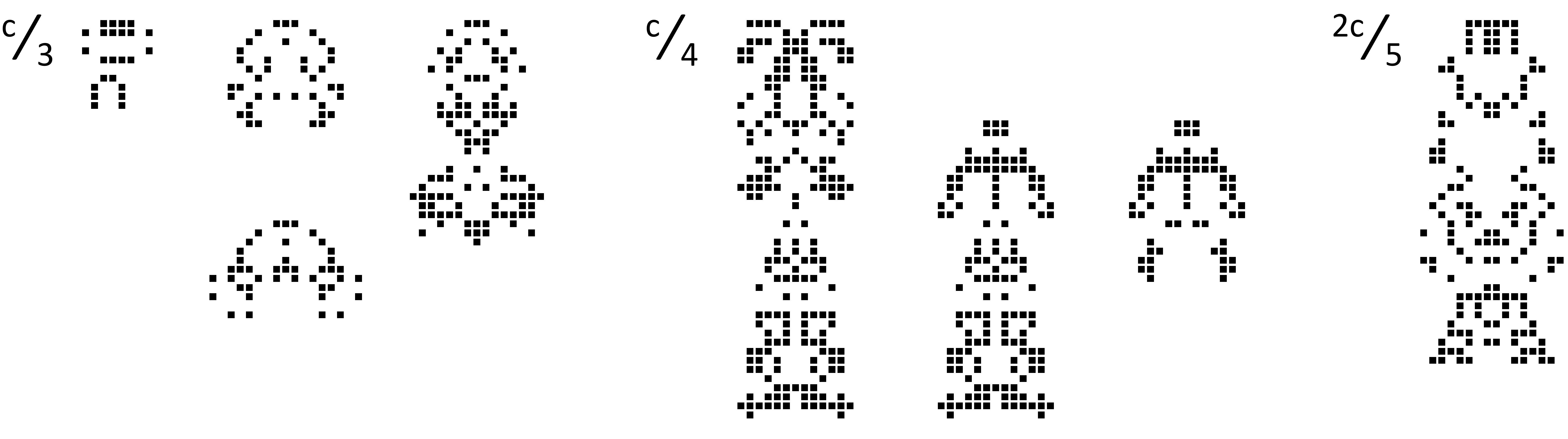}
\caption{Several orthogonal spaceships of various speeds travelling upward found by David Eppstein, Alan Hensel, and Dean Hickerson}
\vspace{0.4in}
\end{figure}

\begin{figure}[ht]
\center
\includegraphics[width=0.95\textwidth]{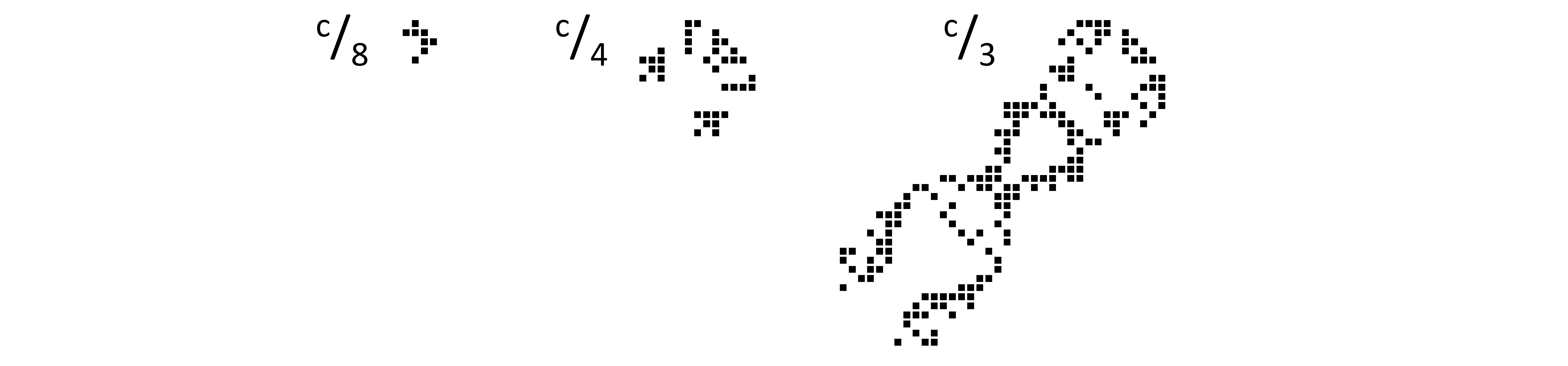}
\caption{Diagonal spaceships of various speeds travelling up and to the right: the $c/8$ spaceship occurs naturally, the $c/4$ spaceship was found by Dean Hickerson in January 1999, and the $c/3$ spaceship was found by David Eppstein}
\end{figure}
\end{document}